\documentclass[a4paper,11pt]{article}

\usepackage[utf8]{inputenc}
\usepackage[english]{babel}
\usepackage{amsmath}
\usepackage{amssymb}
\usepackage{bbm}
\usepackage{hyperref}
\usepackage{amsthm}
\usepackage{diagbox}
\usepackage{graphicx}
\usepackage{listings}
\usepackage{multicol}
\usepackage{color}
\usepackage{fullpage}
\newcommand{\R}{\mathbb R}

\newcommand{\IP}{\mathbb P}
\newcommand{\E}{\mathbb E}

\newcommand{\dr}{\mathrm d r}
\newcommand{\ds}{\mathrm d s}
\newcommand{\dt}{\mathrm d t}
\newcommand{\du}{\mathrm d u}

\newcommand{\ddt}{\partial_t }

\newcommand{\mddx}{\frac{\partial}{\partial x}}

\newcommand{\prog}{\textup{Prog}}
\newcommand{\arcoth}{\textup{arcoth}}

\newcommand{\eps}{\varepsilon}

\theoremstyle{plain}
\newtheorem{theorem}{Theorem}[section]
\newtheorem{example}[theorem]{Example}
\newtheorem{proposition}[theorem]{Proposition}

\newtheorem{lemma}[theorem]{Lemma}
\newtheorem{definition}[theorem]{Definition}
\newtheorem{assumption}[theorem]{Assumption}
\newtheorem{remark}[theorem]{Remark}

\newtheorem*{condition*}{Condition}

\usepackage{parskip}
\makeatletter
\def\thm@space@setup{
\thm@preskip=\parskip \thm@postskip=0pt
}
\makeatother

\numberwithin{equation}{section}

\begin{document}

\title{Optimal trade execution under small market impact and portfolio liquidation with semimartingale strategies\thanks{Financial support through the Stiftung der Deutschen Wirtschaft (sdw) is gratefully acknowledged. We thank Peter Bank, Mikhail Urusov and seminar participants at various institutions for valuable comments and remarks. This is a substantially extended version of a previous preprint entitled {\sl Small impact analysis in stochastically illiquid markets.}}
}

\author{Ulrich Horst\footnote{Department of Mathematics, and School of Business and Economics, Humboldt-Universit\"{a}t zu Berlin, Unter den Linden 6, 10099 Berlin, Germany, \texttt{horst@math.hu-berlin.de}.} \and Evgueni Kivman\footnote{Department of Mathematics, Humboldt-Universit\"{a}t zu Berlin, Unter den Linden 6, 10099 Berlin, Germany, \texttt{kivman@math.hu-berlin.de}.}}
\maketitle

\maketitle

\begin{abstract}
We consider an optimal liquidation problem with instantaneous price impact and stochastic resilience for small instantaneous impact factors. Within our modelling framework, the optimal portfolio process converges to the solution of an optimal liquidation problem with general semimartingale controls when the instantaneous impact factor converges to zero. Our results provide a unified framework within which to embed the two most commonly used modelling frameworks in the liquidation literature and provide a foundation for the use of semimartingale liquidation strategies and the use of portfolio processes of unbounded variation. Our convergence results are based on novel convergence results for BSDEs with singular terminal conditions and novel representation results of BSDEs in terms of uniformly continuous functions of forward processes.
\end{abstract}



\section{Introduction}

The impact of limited liquidity on optimal trade execution has been extensively analyzed in the mathematical finance and stochastic control literature in recent years.
The majority of the optimal portfolio liquidation literature allows for one of two possible price impacts. The first approach, pioneered by Bertsimas and Lo~\cite{bertsimas_lo} and Almgren and Chriss~\cite{almgren_chriss}, divides the price impact in a purely temporary effect, which depends only on the present trading rate and does not influence future prices, and in a permanent effect, which influences the price depending only on the total volume that has been traded in the past. The temporary impact is typically assumed to be linear in the trading rate, leading to a quadratic term in the cost functional. The original modelling framework has been extended in various directions including general stochastic settings with and without model uncertainty and multi-player and mean-field-type models by many authors including Ankirchner et al.~\cite{ankirchner_jeanblanc_kruse}, Cartea et al. ~\cite{cartea_donelly_jaimungal_2017}, Fu et al.~\cite{fu_graewe_horst_popier}, Gatheral and Schied~\cite{gatheral_schied}, Graewe et al.~\cite{graewe_horst_qiu}, Horst et al.~\cite{horst_xia_zhou}, Kruse and Popier~\cite{kruse_popier} and Neuman and Vo{\ss}~\cite{neuman_voss}.

A second approach, initiated by Obizhaeva and Wang~\cite{obizhaeva_wang}, assumes that price impact is not permanent, but transient with the impact of past trades on current prices decaying over time. When impact is transient, one often allows for both absolutely continuous and singular trading strategies. When singular controls are admissible, optimal liquidation strategies usually comprise large block trades at the initial and terminal time.
The work of Obizhaeva and Wang has been extended by Alfonsi et al.~\cite{alfonsi_fruth_schied}, Chen et al.~\cite{chen_kou_wang}, Fruth et al.~\cite{fruth_schoneborn_urusov}, Gatheral~\cite{gatheral}, Guéant~\cite{gueant}, Horst and Naujokat~\cite{horst_naujokat}, Lokka and Xu~\cite{lokka_xu} and Predoiu et al.~\cite{predoiu_shaikhet_shreve}, among many others.

Single and multi-asset liquidation problems with instantaneous and transient market impact and stochastic resilience where trading is confined to absolutely continuous strategies have been analyzed in Graewe and Horst~\cite{graewe_horst} and Horst and Xia~\cite{horst_xia}, respectively. This is consistent with the empirical work of Large~\cite{large} and Lo and Hall~\cite{lo_hall}, which suggests that this resilience does indeed vary stochastically. Although only absolutely continuous trading strategies were admissible in \cite{graewe_horst,horst_xia}, numerical simulations reported in \cite{graewe_horst} suggest that if all model parameters are deterministic constants, then the optimal portfolio process converges to the optimal solution in \cite{obizhaeva_wang} with two block trades and a constant trading rate as the instantaneous impact parameter converges to zero. Cartea and Jaimungal~\cite{cartea_jaimungal_2016} provide empirical evidence that the instantaneous price impact is indeed (much) smaller than permanent (or transient) price impact. The numerical simulations in \cite{graewe_horst} suggest that the model in \cite{graewe_horst} provides a common framework within which to embed the two most commonly used liquidation models \cite{almgren_chriss,obizhaeva_wang} as limiting cases.

This paper provides a rigorous convergence analysis within a Markovian factor model. It turns out that the stochastic setting is quite different from the deterministic one. Most importantly, we show that in the stochastic setting, the optimal portfolio processes obtained in \cite{graewe_horst} converge in the Skorohod $\mathcal M_2$ topology to a process of {\sl infinite variation} with jumps as the instantaneous market impact parameter converges to zero. Our second main result is to prove that the limiting portfolio process is optimal in a liquidation model with semimartingale execution strategies and to explicitly compute the optimal trading cost in the semimartingale execution framework.

Showing that the limiting model solves a liquidation model with semimartingale execution strategies is more than mere byproduct. Control problems with semimartingale strategies are usually difficult to solve because there are no canonical candidates for the value function and/or optimal strategies. We show that the optimal solution in the limiting model is fully determined by the unique bounded solution to a one-dimensional quadratic BSDE. Our limit result provides a novel approach to solving control problems with semimartingale strategies that complements the approaches in \cite{ackermann_kruse_urusov} and \cite{garleanu_pedersen}. They solved related models by passing to a continuous time limit from a sequence of discrete time models.

Within a portfolio liquidation framework, inventory processes with infinite variation were first considered by Lorenz and Schied~\cite{lorenz_schied} to the best of our knowledge. Later, Becherer et al.~\cite{becherer_bilarev_frentrup} considered a trading framework with generalized price impact and proved that the cost functional depends continuously on the trading strategy, considered in several topologies. Bouchard et al.~\cite{bouchard_loeper_zou} considered infinite variation inventory processes in the context of hedging.

The paper closest to ours is the recent work by Ackermann et al.~\cite{ackermann_kruse_urusov}. They considered a liquidation model with general RCLL semimartingale trading strategies. Their framework is more general than ours as they allow for more general filtrations and stochastic order book depth. At the same time, their analysis is confined to risk-neutral traders. In our setting, when the model parameters are deterministic and the instantaneous price impact goes to zero, the case of risk neutral traders - which is then a special case to the model studied in \cite{ackermann_kruse_urusov} - is explicitly solvable. Allowing for risk aversion renders the impact model significantly more complicated as it adds a quadratic dependence of the integrated trading rate into the HJB equation, cf. \cite{graewe_horst} for details.

Our work also complements the work of G\^arleanu and Pedersen~\cite{garleanu_pedersen}. They consider an array of market impact models, including a model with purely transient costs. They write [p.497] that ``with purely [transient] price-impact costs, the optimal portfolio policy can have jumps and infinite quadratic variation.'' As in \cite{ackermann_kruse_urusov}, they justify portfolio holdings with infinite quadratic variation by taking a limit of a sequence of discrete time models with increasing trading frequency. They also prove that the optimal portfolio processes in the discrete time models converge to the optimal portfolio process in the corresponding continuous time model if either the instantaneous price impact converges to a positive constant or the instantaneous price impact factor multiplied by the (increasing) trading frequency converges to zero. However, they do not consider the general case of an instantaneous price impact factor converging to zero. Most importantly, they consider a portfolio choice problem on an infinite time horizon, which avoids the liquidation constraint at the terminal time.

Last but not least our work complements the work of Carmona and Webster~\cite{carmona_webster}, who provide strong evidence that inventories of large traders often do have indeed a non-trivial quadratic variation component. For instance, for the Blackberry stock, they analyze the inventories of ``the three most active Blackberry traders'' on a particular day, namely CIBC World Markets Inc., Royal Bank of Canada Capital Markets, and TD Securities Inc. From their data, they ``suspect that RBC (resp. TD Securities) were trading to acquire a long (resp. short) position in Blackberry'' and found that the corresponding inventory processes were with infinite variation. More generally, they find that systematic tests ``on different days and other actively traded stocks give systematic rejections of this null hypothesis [quadratic variation of inventory being zero], with a $p$-value never greater than $10^{-5}$.'' Our results suggest that inventories with non-trivial quadratic variation arise naturally when market depth is high and resilience and/or market risk fluctuates stochastically. This is very intuitive; in deep markets it is comparably cheap to frequently adjust portfolios to stochastically varying market environments.

The main technical challenge in establishing our convergence results is that the optimal solution to the limiting model cannot be obtained by taking the limit of the three-dimensional quadratic BSDE system that characterizes the optimal solution in the model with positive instantaneous impact. Instead, we prove that the limit is fully characterized by the solution to a one-dimensional quadratic BSDE. Remarkably, this BSDE is independent of the liquidation requirement. As a result, full liquidation takes place if the instantaneous impact parameter converges to zero even it is not strictly required. The reason is a loss in book value of the remaining shares that outweighs the liquidation cost for small instantaneous impact. Our convergence result is based on a novel representation result for solutions for BSDEs driven by It\^o processes in terms of \textit{uniformly} continuous functions of the forward process and on a series of novel convergence results for sequences of singular stochastic integral equations and random ODEs, which we choose to report in an abstract setting in Appendix \ref{SectAppendixRegularityItoBsde} and \ref{SectAppendixConv}, respectively.

The limiting portfolio process is optimal in a liquidation model with general semimartingale execution strategies. Within our modeling framework where the cost coefficients are driven by continuous factor processes block trades optimally occur only at the beginning and the end of the trading period. This is very intuitive as one would expect large block trades to require some form trigger such as an external shock leading to a discontinuous change of cost coefficients.
The proof of optimality proceeds in three steps. We first prove that the process with jumps can be approximated by absolutely continuous ones. This allows us to approximate the trading costs in the semimartingale model by trading costs in the pre-limit models from which we finally deduce the optimality of the limiting process in the semimartingale model by using the optimality of the approximating inventory processes in the pre-limit models. As a byproduct of our approximation, we also obtain that the optimal costs are given in terms of the aforementioned one-dimensional quadratic BSDE.

The rest of the paper is organized as follows. In Section \ref{SectProblemResults}, we recall the modelling setup from \cite{graewe_horst,horst_xia} and summarize our main results. The proofs are given in Sections \ref{SectProofsPart1} and \ref{SectProofsPart2}. A series of fairly abstract convergence results for various stochastic equations with singularities upon which our convergence results are based is postponed to two appendices.

{\bf Notation.} Throughout, randomness is described by an $\R^m$-valued Brownian motion $\{W_t\}_{t\in[0,T]}$ defined on $(\Omega,\mathcal F, \{\mathcal F_t\}_{t \in [0,T]}, \IP)$, a complete probability space, where $\{\mathcal F_t\}_{t \in [0,T]}$ denotes the filtration generated by $W$, augmented by the $\IP$-null-sets. Unless otherwise specified, all equations and inequalities hold in the $\IP$-a.s.~sense. For a subset $A \subseteq \R^d$, we denote by $\mathcal L_{\prog}^2(\Omega\times[0,T];A)$ the set of all progressively measurable $A$-valued stochastic processes $\{X_t\}_{t\in[0,T]}$ such that
$
\E[ \int^T_0 \Vert X_t\Vert_2^2 \, \dt ] < \infty,
$
while $\mathcal L_{\prog}^\infty(\Omega\times[0,T];A)$ denotes the subset of essentially bounded processes. $\mathcal L_{\mathcal P}^2(\Omega\times[0,T];A)$ and $\mathcal L_{\mathcal P}^\infty(\Omega\times[0,T];A)$ denote the respective subsets of predictable processes. Whenever $T-$ appears, we mean that there exists an $\eps > 0$ such that a statement holds for all $T' \in (T-\eps,T)$.

\section{Problem formulation and main results} \label{SectProblemResults}

In this section, we introduce two portfolio liquidation models with stochastic market impact. In the first model, analyzed in Section \ref{SectModelEtaPositive}, the investor is confined to absolutely continuous trading strategies. For small instantaneous market impact, we prove that the optimal liquidation strategy converges to a semimartingale with jumps. In Section \ref{SectModelEta0}, we therefore analyze a liquidation model with semimartingale trading strategies. We prove that the limiting process obtained in Section \ref{SectModelEtaPositive} is optimal in a model where semimartingale strategies that satisfy a suitable regularity condition are admissible.

\subsection{Portfolio liquidation with absolutely continuous strategies} \label{SectModelEtaPositive}

We take the liquidation model analyzed in \cite{graewe_horst,horst_xia} as our starting point and consider an investor that needs to close within a given time interval $[0,T]$ a (single-asset) portfolio of $x_0>0$ shares using a trading strategy $\xi = \{\xi_t\}_{t \in [0,T]}$. If $\xi_t < 0$, the investor is selling the asset at a rate $\xi_t$ at time $t \in [0,T]$, else she is buying it. For a given strategy $\xi$, the corresponding portfolio process $X^\xi = \{X^\xi_t\}_{t \in [0,T]}$ satisfies the ODE
\begin{align*}
	\mathrm dX^\xi_t = \xi_t \, \dt, \quad X^\xi_0 = x_0.
\end{align*}
The set of {\sl admissible strategies} is given by the set $$\mathcal A := \mathcal L_{\prog}^2(\Omega\times[0,T];\R).$$
For a general inventory process $X \in \mathcal L_{\mathcal P}^2 (\Omega\times[0,T];\R)$, the corresponding {\sl transient price impact} is described by $Y^X = \{Y^X_t\}_{t \in [0,T]}$, the unique stochastic process that satisfies the ODE 	
\begin{align}
	\mathrm dY^X_t
	=
	\gamma \, \mathrm dX_t
	- \rho_t Y^X_t \, \dt,
	\quad
	Y^X_0
	=
	0
	\label{eqYDef}
\end{align}
for some constant $\gamma > 0$ and some essentially bounded, adapted, $(0,\infty)$-valued process $\rho=\{\rho_t\}_{t \in [0,T]}$. The process $Y^X$ may be viewed as describing an additional shift in the unaffected benchmark price process generated by the large investor's trading activity. For $\xi \in \mathcal A$, we write $Y^\xi := Y^{X^\xi}$.

For any {\sl instantaneous impact factor} $\eta > 0$ and any {\sl penalization factor}
$$
N \in \mathcal N = \mathcal N(\eta) := [\underline N(\eta),\infty],
$$
where
$$
\underline N (\eta)
:=
\gamma + 1 + \sqrt{ 2 \eta \max \big( \Vert\lambda\Vert_{L^\infty}, \gamma \Vert\rho\Vert_{L^\infty} \big) },
$$
the {\sl cost functional} is given by
\begin{align*}
	J^{\eta,N}(\xi)
	=&
	\E\bigg[
		\frac \eta 2 \int^T_0 (\xi_t)^2 \, \dt
		+ \int^T_0 Y^\xi_t \, \mathrm dX^\xi_t
		+ \frac 1 2 \int^T_0 \lambda_t \big(X^\xi_t\big)^2 \, \dt
		+ \frac N 2 \big(X^\xi_T\big)^2
		- X^\xi_T Y^\xi_T
	\bigg].
\end{align*}
The first term in the above cost functional captures the {\sl instantaneous trading costs}; the second captures the costs from {\sl transient price impact}; the third captures {\sl market risk} where the adapted and non-negative process $\lambda = \{\lambda_t\}_{t\in [0,T]}$ specifies the degree of risk.
If full liquidation is required ($N=\infty$), the fourth term should formally be read as $+\infty \mathbbm 1_{\{X^\xi_T \neq 0\}}$ with the convention $0 \cdot \infty = 0$. The case $N=\infty$ captures the case where full liquidation is required; this case is analyzed in \cite{graewe_horst}. The case $\gamma + 1 \leq N < \infty$ is analyzed in \cite{horst_xia}. The fifth term captures an additional loss in book value of the remaining shares. It drops out of the cost function if $N = \infty$; see \cite{graewe_horst,horst_xia} for further details on the impact costs and cost coefficients.

It has been shown in \cite{graewe_horst,horst_xia} that the optimization problem
\begin{align}
	\min_{\xi \in \mathcal A} J^{\eta,N}(\xi)
	\label{eqMinimizationEtaPos}
\end{align}
has a solution $\hat\xi^{\eta,N}$ for any $N \in {\cal N}$ and any $\eta > 0$. The solution is given in terms of a backward SDE system with possibly singular terminal condition. We index the optimal trading strategies and state processes by $\eta$ and $N$ as we are interested in their behavior for small instantaneous impact factors for both finite and infinite $N$.
\begin{theorem}[\cite{graewe_horst,horst_xia}] \label{thmgraewexiathm}
With the above assumptions, for all $\eta \in (0,\infty)$ and $N\in \mathcal N$, the following holds.
\begin{itemize}
\item[i)] The BSDE system
\begin{align*}
	-dA^{\eta,N}_t
	=&
	\bigg(\lambda_t - \frac 1 \eta \big( A^{\eta,N}_t - \gamma B^{\eta,N}_t \big)^2 \bigg) \, \dt
	- Z^{\eta,N,A}_t \, \mathrm dW_t,
	\\
	- dB^{\eta,N}_t
	=&
	\bigg( - \rho_t B^{\eta,N}_t + \frac 1 \eta \big(\gamma C^{\eta,N}_t - B^{\eta,N}_t + 1 \big) \big( A^{\eta,N}_t - \gamma B^{\eta,N}_t \big) \bigg) \, \dt
	- Z^{\eta,N,B}_t \, \mathrm dW_t,
	\\
	- dC^{\eta,N}_t
	=&
	\bigg( - 2\rho_t C^{\eta,N}_t - \frac 1 \eta \big(\gamma C^{\eta,N}_t - B^{\eta,N}_t + 1 \big)^2 \bigg) \, \dt
	- Z^{\eta,N,C}_t \, \mathrm dW_t
\end{align*}
with terminal condition
\begin{align*}
	B^{\eta,N}_T = 1, \quad
	C^{\eta,N}_T = 0
\end{align*}
together with $A^{\eta,N}_T = N$ in the case $N < \infty$ and $\lim_{t\nearrow T} A^{\eta,\infty}_t = \infty$ in $L^\infty$ if $N=\infty$
has a solution
\[
	\Big( \big( A^{\eta,N},B^{\eta,N},C^{\eta,N} \big) , \big(Z^{\eta,N,A}, Z^{\eta,N,B}, Z^{\eta,N,C}\big) \Big)
\]
that belongs to the space
$\mathcal L^\infty_{\mathcal P}(\Omega\times[0,T];\R^3) \times \mathcal L^2_{\prog}(\Omega\times[0,T];\R^{3\times m})$
if $N<\infty$ and to the space
$\mathcal L^\infty_{\mathcal P}(\Omega\times[0,T-];\R^3) \times \mathcal L^2_{\prog}(\Omega\times[0,T-];\R^{3\times m})$
if $N=\infty$.

\item[ii)] The liquidation problem \eqref{eqMinimizationEtaPos} has a solution $\hat\xi^{\eta,N}$. The corresponding state process
$$\big(\hat X^{\eta,N},\hat Y^{\eta,N}\big) := \big(X^{\hat\xi^{\eta,N}},Y^{\hat\xi^{\eta,N}}\big)$$
is given by the (unique) solution to the ODE system
\begin{equation}\label{eqOdeXY}
\begin{split}
	\ddt \hat X^{\eta,N}_t
	=& - \frac 1 {\sqrt{\eta}} D^{\eta,N}_t \hat X^{\eta,N}_t
	- \frac 1 {\sqrt{\eta}} E^{\eta,N}_t \hat Y^{\eta,N}_t,
	\\
	\ddt \hat Y^{\eta,N}_t
	=& - \frac \gamma {\sqrt{\eta}} D^{\eta,N}_t \hat X^{\eta,N}_t
	- \frac \gamma {\sqrt{\eta}} E^{\eta,N}_t \hat Y^{\eta,N}_t
	- \rho_t \hat Y^{\eta,N}_t
\end{split}
\end{equation}
with initial conditions $\hat X^{\eta,N}_0 = x_0$ and $\hat Y^{\eta,N}_0 = 0$ where
\begin{align*}
	D^{\eta,N} :=
	\frac 1 {\sqrt{\eta}} \big(A^{\eta,N}-\gamma B^{\eta,N}\big),
	\quad
	E^{\eta,N} :=
	\frac 1 {\sqrt{\eta}} \big(\gamma C^{\eta,N} - B^{\eta,N} + 1\big).
\end{align*}
\end{itemize}
\end{theorem}

Let us now define the process
\[
	\hat Z^{\eta,N} := \gamma \hat X^{\eta,N} - \hat Y^{\eta,N}.
\]
The benefit of defining this process is that the terms in the ODE \eqref{eqOdeXY} that are multiplied by $\sqrt{\eta^{-1}}$ drop out so that we expect that the process $\hat Z^{\eta,N}$ remains stable for small values of $\eta$. Next, we state a result on the optimal state process and the previously introduced process $\hat Z^{\eta,N}$ that will be important for our subsequent analysis. In particular, we show that the optimal portfolio process $\hat X^{\eta,N}$ never changes its sign. The proof is given in Section \ref{SectEtaPos}.

\begin{theorem} \label{thmXYEtaPosBdthm}
For all $\eta \in (0,\infty)$, $N\in \mathcal N$, the process $\hat Z^{\eta,N}$ is non-increasing on $[0,T]$. Moreover,
\begin{align*}
	\IP\big[
		\hat X^{\eta,N}_t \in (0,x_0),
		\hat Y^{\eta,N}_t \in (-\gamma x_0,0),
		\hat Z^{\eta,N}_t \in (0,\gamma x_0)
		\text{ for all } t \in (0,T)
	\big]
	= 1.
\end{align*}
\end{theorem}

We are interested in the dynamics of the optimal portfolio processes for small instantaneous price impact. We address this problem within a factor model
where the cost coefficients $\lambda$ and $\rho$ are driven by an It\^o diffusion, which is given by the unique strong solution to the SDE
\begin{align*}
	d\chi_t
	= \mu(t,\chi_t) \, \dt + \sigma(t,\chi_t) \, \mathrm dW_t, \quad \chi_0 = \overline\chi_0
\end{align*}
on $[0,T]$ with $\overline\chi_0 \in \R^n$.
We assume throughout that the function
\begin{align*}
	(\mu, \sigma) \colon [0,T] \times \R^n \to \R^n \times \R^{n\times m}
\end{align*}
is bounded, measurable and uniformly Lipschitz continuous in the space variable:
\[
\big\Vert
(\mu, \sigma)(t,x_1) -
(\mu, \sigma)(t,x_2)
\big\Vert_\infty
\le c \Vert x_1 - x_2\Vert_\infty.
\]

\begin{assumption}\label{assFactorModel}
The processes $\rho$ and $\lambda$ are of the form
\[
	(\rho_t,\lambda_t) = f(t,\chi_t)
\]
for some bounded $C^{1,2}$ function $f = (f^\rho,f^\lambda) \colon [0,T] \times \R^n \to [0,\infty)^2$ with bounded derivatives. Moreover, the function $f^\rho$ is bounded away from zero.
\end{assumption}

For convenience, we define the stochastic process $\varphi = \{\varphi_t\}_{t \in [0,T]}$ by
$$
\varphi_t := \sqrt{\lambda_t + 2 \gamma \rho_t},
$$
choose constants $\underline\rho, \overline\rho, \underline\varphi,\overline\varphi, \overline\lambda \in (0,\infty)$ such that
\begin{align*}
	\underline\rho \le \rho \le \overline\rho,
	\quad
	0 \le \lambda \le \overline\lambda,
	\quad
	\underline\varphi \le \varphi \le \overline\varphi.
\end{align*}

In what follows, we heuristically argue that the processes $\hat X^{\eta,N}$ converge to a limit process $\hat X^0$ (independent of $N$) as $\eta \to 0$ and identify the limit $\hat X^0$.
Since the ODE system \eqref{eqOdeXY} is not defined for $\eta =0$, we cannot define the limiting process as the solution to this system. Instead,
we first identify the limits of the coefficients of the ODE system and then derive candidate limits for the state processes in terms of the limiting coefficients.

\subsubsection{Convergence of the coefficient processes}

In this section, we state the convergence results for the coefficient processes $D^{\eta,N}$ and $E^{\eta,N}$ of the ODE system \eqref{eqOdeXY} as $\eta \to 0$. In particular, we prove that their limits $D^{0}$ and $E^0$ exist and are driven by a common factor, which is given by the solution of a quadratic BSDE.

Before proceeding to the limit result, we provide some heuristics for the convergence. Assuming for simplicity that all coefficients are deterministic, the dynamics of the coefficient processes satisfy
\begin{align*}
	\sqrt\eta \dot D^{\eta,N}_t
	=&
	\big( D^{\eta,N}_t \big)^2
	- \lambda_t
	- \gamma \dot B^{\eta,N}_t,
	\\
	\sqrt\eta \dot E^{\eta,N}_t
	=&
	2 \sqrt\eta \rho_t E^{\eta,N}_t
	+ \gamma \big( E^{\eta,N}_t \big)^2
	- 2 \rho_t \big( 1 - B^{\eta,N}_t \big)
	- \dot B^{\eta,N}_t.
\end{align*}
Letting $\eta \to 0$, we expect that
\begin{align*}
	0
	=&
	\big( D^0_t \big)^2
	- \lambda_t
	- \gamma \dot B^0_t,
	\\
	0
	=&
	\gamma \big( E^0_t \big)^2
	+ 2 \rho_t \big( B^0_t - 1 \big)
	- \dot B^0_t,
\end{align*}
that is, we expect that
\begin{align} \label{D1}
	D^0_t
	=
	\sqrt{ \gamma \dot B^0_t + \lambda_t }
	\quad \mbox{and} \quad
	E^0_t
	=
	\sqrt{ \gamma^{-1} \big( \dot B^0_t + 2 \rho_t ( 1 - B^0_t ) \big) }.
\end{align}
Moreover, by the choice of the coefficients $D^{\eta,N}$ and $E^{\eta,N}$, we expect that
\begin{align} \label{D2}
	- \dot B^0_t
	=&
	- \rho_t B^0_t
	+ D^0_t E^0_t.
\end{align}
The three equalities combined yield
\begin{equation}
\begin{split}
	\big( D^0_t + \gamma E^0_t \big)^2
	=&
	\gamma \dot B^0_t
	+ \lambda_t
	+ 2 \gamma D^0_t E^0_t
	+ \gamma \big( \dot B^0_t + 2 \rho_t ( 1 - B^0_t ) \big)
	\\=&
	(\varphi_t)^2
	+ 2 \gamma \big(
		\dot B^0_t
		+ D^0_t E^0_t
		- \rho_t B^0_t
	\big)
	\\=&
	(\varphi_t)^2.
\end{split}
\label{eqDPlusGammaEEqPhi}
\end{equation}
Plugging \eqref{D2} and \eqref{eqDPlusGammaEEqPhi} back into \eqref{D1} yields
\begin{align*}
	D^0_t
	=
	(\varphi_t)^{-1} \big( \lambda_t + \gamma \rho_t B^0_t \big) \quad \mbox{and} \quad E^0_t
	=
	(\varphi_t)^{-1} \rho_t \big( 2 - B^0_t \big).
\end{align*}
Hence we expect $D^0$ and $E^0$ to be driven by $B^0$ and $B^0$ to satisfy the ODE
\begin{align*}
	- \dot B^0_t
	=&
	- \rho_t B^0_t
	+ D^0_t E^0_t
	\\=&
	- \frac 1 {\varphi_t^2} \Big(\gamma \big(\rho_t B^0_t\big)^2 - 2 \lambda_t \rho_t \big(1 - B^0_t\big)\Big)
	\end{align*}
with terminal condition $B^0_T=1$ (because $B^{\eta,N}_T = 1$). Our heuristic also suggests that the limit processes are independent of the liquidation requirement.

\begin{example}
If $\lambda = C \rho$ for some constant $C \geq 0$, then the process $B^0$ can be computed explicitly. If $C = 0$, then
$
B^0_t
= 2/({2 + \int_t^T \rho_s \, \ds }).
$
If $C > 0$, then
\begin{align*}
	B^0_t
	=&
	- \frac C \gamma
	+ \frac 1 \gamma \sqrt{C(C+2\gamma)} \coth \bigg(
		\arcoth\bigg( \frac{C+\gamma}{\sqrt{C(C+2\gamma)}} \bigg)
		+ \sqrt{\frac C {C+2\gamma}} \int_t^T \rho_s \, \ds
	\bigg).
\end{align*}
\end{example}

The preceding heuristic suggests that the limiting coefficient processes are driven by a solution to the BSDE corresponding to the above ODE for $B^0$. The following lemma is proven in Section \ref{SectEta0}.

\begin{lemma} \label{thmexistsB0lemma}
There exists a unique solution
$
(B^0,Z^{0,B})
$
in the space
\[
\mathcal L^\infty_{\mathcal P} \big( \Omega\times[0,T];(0,\infty) \big)
\times \mathcal L^2_{\mathcal P} \big( \Omega\times[0,T];\R^m \big)
\]
to the BSDE
\begin{align}
	- \, \mathrm dB^0_t
	=
	- \frac 1 {\varphi_t^2} \Big(\gamma \big(\rho_t B^0_t\big)^2 - 2 \lambda_t \rho_t \big(1 - B^0_t\big)\Big) \, \dt
	- Z^{0,B}_t \, \mathrm dW_t,
	\quad
	B^0_T
	=
	1.
	\label{eqBsdeB0}
\end{align}
The process $B^0$ is bounded from above by $1$ and bounded from below by
\begin{align} \label{lower-bound-B0}
	\underline{B}^0_t := \exp\big( - \underline\varphi^{-2} \gamma \overline\rho^2 (T-t)\big).
\end{align}
Moreover, there exists a {\sl uniformly} continuous function (``decoupling field'') $h\colon [0,T]\times\R^n \to \R$ such that $B^0$ and $\{h(t,\chi_t)\}_{t \in [0,T]}$ are indistinguishable.
\end{lemma}

We prove below that the process $B^{\eta,N}$ converges to $B^0$ as $\eta \to 0$ and that $D^{\eta,N}$ and $E^{\eta,N}$ converge to the processes
\begin{equation}
\begin{split} \label{eqD0def}
	D^0_t := (\varphi_t)^{-1} \big( \lambda_t + \gamma \rho_t B^0_t \big) \quad \mbox{and} \quad
	E^0_t := (\varphi_t)^{-1} \rho_t \big( 2 - B^0_t \big),
\end{split}
\end{equation}
respectively.
In view of Lemma \ref{thmexistsB0lemma}, these processes are well-defined and so the dynamics of the process $B^0$ can be rewritten as
\begin{align*}
	- \, \mathrm dB^0_t
	=
	- \big(
		\rho_t B^0_t
		- D^0_t E^0_t
	\big) \, \dt
	- Z^{0,B}_t \, \mathrm dW_t, \quad B^0_T = 1.
\end{align*}
Likewise, the BSDE for the process $B^{\eta,N}$ can be rewritten as
\begin{align}
	- \, \mathrm dB^{\eta,N}_t
	=&
	- \big( \rho_t B^{\eta,N}_t - D^{\eta,N}_t E^{\eta,N}_t \big) \, \dt
	- Z^{\eta,N,B}_t \, \mathrm dW_t,
	\quad
	B^{\eta,N}_T = 1.
	\label{eqSdeB}
\end{align}
This suggests that the process $B^{\eta,N}$ converges to $B^0$ on the entire interval $[0,T]$. By contrast, convergence of the processes $D^{\eta,N}$ and $E^{\eta,N}$ can only be expected to hold on compact subintervals of $[0,T)$ because the terminal conditions of the limiting and the approximating processes are different. Specifically, we have the following result. Its proof is given in Section \ref{SectBDEFconv}.
\begin{proposition} \label{thmbdeconv0prop}
Let
\begin{align*}
b^{\eta,N} := B^{\eta,N} - B^0, \quad
d^{\eta,N} := D^{\eta,N} - D^0, \quad
e^{\eta,N} := E^{\eta,N} - E^0.
\end{align*}
Then, for all $\eps > 0$, there exists an $\eta_0>0$ such that, for all $\eta \in (0,\eta_0]$ and for all $N \in \mathcal N$,
\begin{align*}
	\IP\bigg[
	&\sup_{t \in [0,T)} \max\Big( \big|b^{\eta,N}_t\big| , -d^{\eta,N}_t , e^{\eta,N}_t, - d^{\eta,N}_t - \gamma e^{\eta,N}_t \Big) \le \eps,
	\\&\quad\quad
	\inf_{t \in [0,T-\eps]} \min\big(-d^{\eta,N}_t, e^{\eta,N}_t, - d^{\eta,N}_t - \gamma e^{\eta,N}_t \big) \ge -\eps
	\bigg]
	= 1.
\end{align*}
\end{proposition}

\subsubsection{Convergence of the state process}

Having derived the limits of the coefficient processes, we can now heuristically derive the limits of the processes $\hat X^{\eta,N}$, $\hat Y^{\eta,N}$ and $\hat Z^{\eta,N}$, which we denote by $\hat X^0$, $\hat Y^0$ and $\hat Z^0$, respectively.


Since $\hat X^{\eta,\infty}_T = 0$ for all $\eta > 0$, we expect that $\hat X^0_T = 0$. We prove in Lemma \ref{thmXTnFiniteTo0lemma} that this convergence also holds if $N$ is finite. The proof heavily relies on the optimality of $\hat X^0$ in the semimartingale portfolio liquidation model.

Assuming that the optimal trading strategy remains stable if $\eta \to 0$, the ODE \eqref{eqOdeXY} suggests that the term $D^{\eta,N} \hat X^{\eta,N} + E^{\eta,N} \hat Y^{\eta,N}$ is small for small $\eta$ and hence that
\[
	D^0_t \hat X^0_t = - E^0_t \hat Y^0_t \quad \mbox{on} \, [0,T).
\]
We do not conjecture the above relation at the terminal time because the convergence of $D^{\eta,N}$ and $E^{\eta,N}$ only holds on $[0,T)$. Assuming that
\begin{align*}
	\hat Z^0 = \gamma \hat X^0 - \hat Y^0
	\quad \mbox{on} \,
	[0,T),
\end{align*}
the equation \eqref{eqDPlusGammaEEqPhi} implies
\[
	\hat Z^0 = - \frac{\varphi}{D^0} \hat Y^0 \quad \mbox{and hence} \quad
	\hat Y^0 = - \frac{D^0}{\varphi} \hat Z^0.
\]	
On the other hand, by definition,
$
\ddt \hat Z^{\eta,N}
= \rho \hat Y^{\eta,N}
$.
As $\eta \to 0$, this suggests that $\ddt \hat Z^0 = \rho \hat Y^0$ and hence that
\begin{align*}
	\ddt \hat Z^0_t
	= - \frac {\rho_t D^0_t} {\varphi_t} \hat Z^0, \quad \hat Z^0_0 = \gamma x_0.
\end{align*}

This motivates us to {\sl define the process}
\begin{align*}
	\hat Z^0_t
	 : =
	\gamma x_0 \exp\bigg( -\int^t_0 \frac{\rho_s}{\varphi_s} D^0_s \, \ds \bigg), \quad
	t \in [0,T].
\end{align*}
Since we expect that $\hat X^0_T = 0$ and that $\hat Z^0 = \gamma \hat X^0 - \hat Y^0$, we now introduce the {\sl candidate limiting state processes}
\begin{equation}
\label{X0}
\begin{split}
	\hat X^0_t
	& : =
	\mathbbm 1_{[0,T)}(t) \frac {E^0_t}{\varphi_t} \hat Z^0_t, \\
	\hat Y^0_t
	& : =
	- \mathbbm 1_{[0,T)}(t) \frac {D^0_t}{\varphi_t} \hat Z^0_t
	- \mathbbm 1_{\{T\}}(t) \hat Z^0_T,
	\quad
	t \in [0,T].
\end{split}
\end{equation}
Since $(\hat X^0_{0-},\hat Y^0_{0-}) = (x_0,0)$, we expect that the limiting state process jumps at the initial and the terminal time. In particular, we cannot expect uniform convergence on $[0,T]$.

We also expect the limiting state processes to be of unbounded variation; this can already be deduced from Figure 1. The figure also suggests that the portfolio process is more or less monotone for large $\eta$, while this property is lost for small $\eta$. When $\eta\to 0$, adjustments to small changes in market environments are cheap. This is very different from round-trip strategies where own impact is used to drive market prices into a favorable direction.

Figure \ref{fig:1} also suggests that the limiting portfolio process jumps only at times $0$ and $T$. This is consistent with the definition of candidate processes \eqref{X0} as well as the observation in \cite{horst_naujokat}, according to which jumps in the optimal strategy can only be triggered by exogenous shocks like jumps in the cost coefficients, which are absent in the present model.

\begin{figure}
\includegraphics{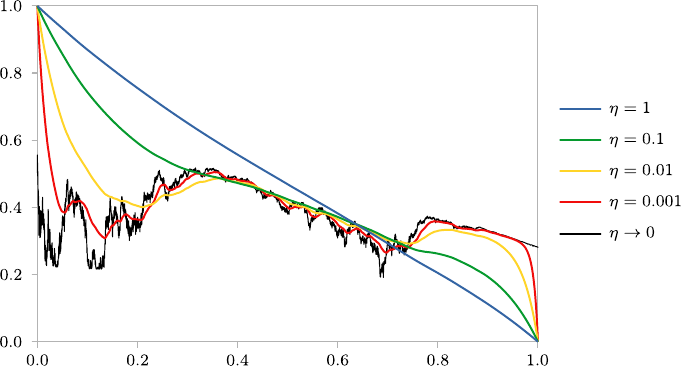}
\caption{Optimal trading strategies $\hat X^{\eta,\infty}$ for the liquidation model for different instantaneous impact factors and their limit $\hat X^0$ for $m=n=T=x_0=1$, $\gamma=3$, $\rho_t = 1 + 0{.}9 \sin(2{.}5 W_t)$, $\lambda\equiv1$.}
\label{fig:1} 
\end{figure}

It remains to clarify in which sense the state processes converge. Contrary to the convergence result stated in Proposition \ref{thmbdeconv0prop}, we can only expect convergence in probability because the state process follows a forward ODE while the coefficient processes follow backward SDEs; see also Appendix \ref{SectAppendixConvOde}. The following theorem establishes uniform convergence in probability on compact subintervals on $(0,T)$ along with some ``upper/lower convergence'' at the initial and terminal time. The proof is given in Section \ref{SectXYconv}.

\begin{theorem} \label{thmxConv0Specificthm}
For all $\eps>0$ and $\delta > 0$, there exists an $\eta_0>0$ such that, for all $\eta \in (0,\eta_0]$ and all $N \in \mathcal N$,
\begin{align*}
	\IP\Big[
		&
		\sup_{t\in [\eps,T)} \max \big( \hat X^{\eta,N}_t - \hat X^0_t , \hat Y^{\eta,N}_t - \hat Y^0_t \big) \le \eps ,
\\&
		\quad
		\inf_{t\in [0,T-\eps]} \min \big( \hat X^{\eta,N}_t - \hat X^0_t , \hat Y^{\eta,N}_t - \hat Y^0_t \big) \ge - \eps
	\Big]
	\ge 1 - \delta.
\end{align*}
\end{theorem}

The preceding theorem does not provide a convergence result on the whole time interval, due to the jumps of the limit processes at the initial and terminal time. However, along with our results from Section \ref{SectModelEta0}, it allows us to prove the convergence of the graphs of the state processes on the entire time interval. The completed graph of a RCLL function $X \colon \{0-\}\cup[0,T] \to \R$ with finitely many jumps is defined by
\begin{align*}
	G_X
	&:=
	\big\{
		(t,x) \in [0,T]\times\R
		:
		x = X_t = X_{t-} \text{ or }
		X_{t-} < X_t, x \in [X_{t-}, X_t] \\
		& \quad
		\text{or} \,
		X_{t-} > X_t, x \in [X_t, X_{t-}]
	\big\}.
\end{align*}

The Skorohod $\mathcal M_2$ distance between $X$ and $Y$ is defined as the Hausdorff distance between their completed graphs, i.e.
\begin{align*}
	d_{ \mathcal M_2 }(X,Y)
	:=
	\max\Big(
		\sup_{p \in G_X} \min_{q \in G_Y} \Vert p-q\Vert_\infty,
		\sup_{q \in G_Y} \min_{p \in G_X} \Vert p-q\Vert_\infty
	\Big)
	\in [0,\infty],
\end{align*}

where
\begin{align*}
	\big\Vert (s,y) \big\Vert_\infty
	:=
	\max \big( |s|,|y| \big).
\end{align*}

If strict liquidation is required, then Theorem \ref{thmxConv0Specificthm} is sufficient to prove convergence of the state processes in the Skorohod $\mathcal M_2$ sense.
Even if liquidation is not required, it turns out that the terminal position converges to zero as $\eta \to 0$. Heuristically, this can be seen as follows.

Let $t_0 \in (0,T)$. Disregarding market risk costs, which we expect to be of order $\mathcal O(T-t_0)$ and hence negligible if $t_0 \to T$, and disregarding instantaneous impact costs for the moment, the cost functional for any given admissible strategy $\xi$ is given by
\begin{align*}
	&
	\int^T_{t_0} Y^\xi_s \xi_s \, \ds
	- X^\xi_T Y^\xi_T
	+ \frac N 2 \big( X^\xi_T \big)^2
	\\=&
	Y^\xi_T X^\xi_T
	- Y^\xi_{t_0} X^\xi_{t_0}
	- \int^T_{t_0} \dot Y^\xi_s X^\xi_s \, \ds
	- X^\xi_T Y^\xi_T
	+ \frac N 2 \big( X^\xi_T \big)^2
	\\=&
	- Y^\xi_{t_0} X^\xi_{t_0}
	- \int^T_{t_0} \gamma \xi_s X^\xi_s \, \ds
	+ \int^T_{t_0} \rho Y^\xi_s X^\xi_s \, \ds
	+ \frac N 2 \big( X^\xi_T \big)^2
	\\=&
	- Y^\xi_{t_0} X^\xi_{t_0}
	- \frac \gamma 2 \Big( \big( X^\xi_T \big)^2 - \big( X^\xi_{t_0} \big)^2 \Big)
	+ \int^T_{t_0} \rho Y^\xi_s X^\xi_s \, \ds
	+ \frac N 2 \big( X^\xi_T \big)^2
	\\=&
	- Y^\xi_{t_0} X^\xi_{t_0}
	+ \frac \gamma 2 \big( X^\xi_{t_0} \big)^2
	+ \frac {N - \gamma} 2 \big( X^\xi_T \big)^2
	+ \int^T_{t_0} \rho Y^\xi_s X^\xi_s \, \ds.
\end{align*}
Hence, we expect the {\sl controllable} costs to satisfy
\[
	\mathbb E \bigg[ \frac {N - \gamma} 2 \big( X^\xi_T \big)^2 + \int^T_{t_0} \rho Y^\xi_s X^\xi_s \, \ds \bigg]
	= \mathbb E \bigg[ \frac {N - \gamma} 2 \big( X^\xi_T \big)^2 \bigg] + \mathcal{O}(T-t_0)
\]
plus instantaneous impact costs. Since $N > \gamma$, this suggests to make $X^\xi_T$ small, which is cheap if $\eta$ is small. More precisely, we have the following result; its proof is given in Section \ref{SectXGraphConv}.

\begin{proposition} \label{thmxConv0prop}
For all $\eps>0$ and $\delta>0$, there exists an $\eta_0>0$ such that, for all $\eta\in(0,\eta_0]$ and all $N \in \mathcal N$,
\begin{align*}
	\IP\Big[d_{ \mathcal M_2 } \big(\hat X^{\eta,N},\hat X^0 \big) \le \eps\Big]
	\ge
	1 - \delta.
\end{align*}
\end{proposition}


\subsection{Optimal liquidation with semimartingale strategies} \label{SectModelEta0}

In this section, we prove that the limit process $\hat X^0$ is the optimal portfolio process in a trade execution model with semimartingale trading strategies.

In our semimartingale model, a trading strategy is given by a triple $$\theta=(j^+,j^-,V)$$ where $j^+$ and $j^-$ are real-valued, non-decreasing pure jump processes and $V$ is a real-valued continuous Brownian semimartingale starting in zero. The jump processes $j^+$ and $j^-$ describe the cumulative effects of buying, respectively selling large blocks of shares while the continuous semimartingale $V$ describes the effect of continuously trading small amounts of the stock. The portfolio process $X^\theta = \{X^\theta_t\}_{t\in[0,T]}$ associated with a strategy $\theta$ is then given by
\begin{align*}
	X^\theta_t
	= x_0 + j^+_t - j^-_t + V_t.
\end{align*}
The associated price impact process, again given by \eqref{eqYDef}, is denoted $Y^\theta$. We note that $X^\theta$ and $Y^\theta$ are semimartingales.

We assume that strict liquidation is required and that the cost associated with a trading strategy $\theta$ is given by

\begin{align*}
	J^0(\theta)
	:= \E \bigg[
		\int_{(0,T]} Y^\theta_{s-} \, \mathrm dX^\theta_s
		+ \int^T_0 \frac 1 2 \lambda_s \big( X^\theta_s \big)^2 \, \ds
		+ \frac \gamma 2 \big[ X^\theta \big]_T
	\bigg].
\end{align*}
The first term captures the transient price impact cost; the second term captures market risk. The third term emerges as an additional cost term when passing from discrete to continuous time, as shown in \cite{ackermann_kruse_urusov}. Moreover, in the absence of this term, arbitrarily low costs can be achieved; see \cite{ackermann_kruse_urusov} for details.

The cost function can be conveniently rewritten as

\begin{align*}
	J^0(\theta)
	=&
	\E \bigg[
		\frac \gamma 2 (j^+_0 - j^-_0)^2
		+ \int_{(0,T]} \frac 1 2 \big( Y^\theta_{s-} + Y^\theta_s \big) \, \mathrm dX^\theta_s
		+ \int^T_0 \frac 1 2 \lambda_s \big( X^\theta_s \big)^2 \, \ds
		+ \frac \gamma 2 [V]_T
	\bigg].
\end{align*}
This representation supports our intuition that the price impact before and the price impact after the jump equally influence the total cost. The first term in this expression captures the cost of the initial block trade at time $t=0$.

The cost functional is well defined under the following admissibility condition.
\begin{definition} \label{assAdmissible}
A trading strategy $\theta = (j^+,j^-,V)$ is called admissible if the liquidation constraint $X^\theta_T = 0$ holds, if $j^\pm$ is a RCLL, predictable, real-valued, non-decreasing and square integrable pure jump process, and $V$ is a continuous semimartingale starting in zero with
\begin{align}
	\E\Big[ \max_{t \in [0,T]} \vert V_t\vert^2 \Big]
	<
	\infty.
	\label{eqVMaxLp}
\end{align}
The set of all admissible trading strategies is denoted $\mathcal A^0$.
\end{definition}
Our goal is now to solve the optimization problem
\begin{align*}
	\min_{\theta \in \mathcal A^0} J^0(\theta).
\end{align*}
To this end, we verify directly that the limit process $\hat X^0$ obtained in the previous section is optimal. The results of Section \ref{SectModelEtaPositive} show that the process has the following representation:
\[
	\hat X^0_t = x_0 - \hat j^-_t + \hat V_t
\]	
where the jump process $\hat j^-$ and the continuous part are given by, respectively
\begin{align*}
	\hat j^-_t
	& :=
	x_0 - \hat X^0_0
	+ \mathbbm 1_{\{T\}} (t) \hat X^0_{t-}, \\
	\hat V_t
	& :=
	(\varphi_t)^{-2} \rho_t \big( 2 - B^0_t \big) \hat Z^0_t
	- (\varphi_0)^{-2} \rho_0 \big( 2 - B^0_0 \big) \hat Z^0_0.
\end{align*}
In view of Assumption \ref{assFactorModel} and because $B^0$ is a continuous semimartingale and $\hat Z^0$ is differentiable, the process $\hat V$ is a continuous semimartingale starting in zero. Hence the following holds.
\begin{lemma} \label{thmHatThetaAdmissiblelemma}
The strategy $\hat\theta = (0,\hat j^-, \hat V)$ is admissible.
\end{lemma}
In order to prove that $\hat\theta$ is optimal, we approximate the cost and the portfolio process associated with any strategy $\theta \in \mathcal A^0$ by the cost and portfolio processes corresponding to absolutely continuous trading strategies. To this end, we first approximate the continuous semimartingale part $V$ by differentiable processes. The proof of the following Lemma is given in Section \ref{SectVApproximation}.

\begin{lemma} \label{thmWPApproximationExistslemma}
For all $\theta = (j^+,j^-,V) \in \mathcal A^0$ and for all $\beta,\delta > 0$, there exists a constant $\nu > 0$ and an adapted and continuous $\dot V^{\beta,\nu} \colon \Omega \times [0,T] \to [-\frac \beta \nu,\frac \beta \nu]$ such that
\begin{align}
	\max_{t \in [0,T]} \bigg\vert \int^t_0 \dot V^{\beta,\nu}_s \, \ds \bigg\vert
	\le
	\max_{t \in [0,T]} \vert V_t\vert.
	\label{eqMaxVBetaNuLeMaxV}
\end{align}
and
\begin{align}
	\IP\bigg[ \sup_{t \in [0,T]} \bigg\vert \int^t_0 \dot V^{\beta,\nu}_s \, \ds - V_t \bigg\vert \le 3\beta \bigg]
	\ge
	1 - \delta.
	\label{eqIntDotVMinusVSmall}
\end{align}
\end{lemma}

Next, we approximate the portfolio process $X^\theta$ by a portfolio process associated with an absolutely continuous strategy. To this end, for all $\theta\in \mathcal A^0$ and $\beta,\nu,\eps > 0$, we define the integrable process
\begin{align*}
	\xi^{\theta,\beta,\nu,\eps}_t
	:=
	\begin{cases}
	\frac 1 \eps ( j^+_t - j^+_{t-\eps} )
	- \frac 1 \eps ( j^-_t - j^-_{t-\eps} )
	+ \dot V^{\beta,\nu}_t,
	& t \le T - \eps,
	\\
	- \frac 1 \eps \big( x_0 + \int^{T-\eps}_0 \xi^{\theta,\beta,\nu,\eps}_u \, \du \big),
	& t > T-\eps.
	\end{cases}
\end{align*}
In view of square-integrability of $j^\pm$, we see that $\xi^{\theta,\beta,\nu,\eps}$ belongs to $\mathcal A$.
The corresponding portfolio process is denoted $X^{\theta,\beta,\nu,\eps}$. The proof of the following Lemma is given in Section \ref{SectConvergence}.

\begin{lemma} \label{thmEintDeltaXsqconv0lemma}
For all $\theta = (j^+,j^-,V) \in \mathcal A^0$ and for all $\delta>0$, there exist $\beta>0$, $\nu>0$ and $\eps>0$ such that
\begin{align*}
	\E\bigg[ \int^T_0
		\big( X^{\theta,\beta,\nu,\eps}_t - X^\theta_t \big)^2
	\, \dt \bigg]
	\le \delta.
\end{align*}
\end{lemma}

For all $\xi \in \mathcal A$ with $\int^T_0 \xi_s \, \ds = -x_0$, we can define
$V^\xi_t := \int^t_0 \xi_s \, \ds$
such that $(0,0,V^\xi) \in \mathcal A^0$ and, for all $\eta>0$, we have
\begin{align}
	J^{\eta,\infty}(\xi)
	=
	J^0 \big( 0, 0, V^\xi \big)
	+ \frac \eta 2 \E\bigg[ \int^T_0 (\xi_t)^2 \, \dt \bigg].
	\label{eqJEtaEqJ0PlusEIntXiSq}
\end{align}

The preceding lemma allows us to establish a cost estimate. The proof is given in Section \ref{SectCostApproximation}.

\begin{lemma} \label{thmJetaconvJ0lemma}
For all $C > 0$, there exists a constant $D(C) > 0$ such that the following holds:

For all $\theta = (j^+,j^-,V) \in \mathcal A^0$ with $\E[ \int^T_0 (X^\theta_t)^2 \, \dt ] \le C$ and for all $\xi \in \mathcal A$ with $\int^T_0 \xi_s \, \ds = -x_0$,
\begin{align*}
	&
	\Big\vert J^0(\theta) - J^0\big(0,0,V^\xi\big) \Big\vert
	\le
	 D(C) \cdot \bigg(
		\sqrt{ \E\bigg[ \int^T_0 \big(X^\theta_t - X^{\xi}_t\big)^2 \, \dt \bigg] }
		+ \E\bigg[ \int^T_0 \big(X^\theta_t - X^{\xi}_t\big)^2 \, \dt \bigg]
	\bigg).
\end{align*}
\end{lemma}

As a consequence of the previous results, the optimal instantaneous price impact term converges to zero as $\eta \to 0$. The proof is given in Section \ref{SectInstantaneousCostConvergence}.

\begin{lemma} \label{thmetaEintxisqconv0lemma}
We have
\begin{align*}
	&
	\lim_{\eta \to 0} \frac \eta 2 \E\bigg[ \int^T_0 \big( \hat\xi^{\eta,\infty}_t \big)^2 \, \dt \bigg]
	=
	0.
\end{align*}
\end{lemma}

The cost estimate in Lemma \ref{thmJetaconvJ0lemma} allows us to establish the optimality of the trading strategy $\hat\theta$ by using the optimality of $\hat\xi^{\eta,\infty}$ in the strict liquidation model with absolutely continuous strategies. It turns out that the minimal trading costs are fully determined by the initial value $B^0_0$ of the process $B^0$ along with the impact factor $\gamma$ and the initial portfolio.

\begin{theorem} \label{thmthetaoptimalJ0thm}
It holds that
\begin{align*}
	\min_{\theta \in \mathcal A^0} J^0(\theta)
	=
	J^0 \big( \hat\theta \big)
	=
	\frac 1 2 \gamma (x_0)^2 B^0_0.
\end{align*}
\end{theorem}

\begin{proof}
Let us assume to the contrary that $\hat\theta$ does not minimize the cost functional $J^0$ over the set $\mathcal A^0$.
Then, there exist a strategy $\theta = (j^+,j^-,V) \in \mathcal A^0$ and a constant $\delta \in (0,\infty)$ such that $J^0(\theta) + \delta \le J^0(\hat\theta)$. Moreover, due to \eqref{eqJEtaEqJ0PlusEIntXiSq} and Lemma \ref{thmJetaconvJ0lemma}, for all $\eta,\beta,\nu,\eps \in (0,\infty)$, we have (for convecience, let $p(x) := x + \sqrt x$)
\begin{align*}
	&
	 J^{\eta,\infty}\big( \xi^{\theta,\beta,\nu,\eps} \big)
	- J^{\eta,\infty}\big( \hat\xi^{\eta,\infty} \big)
	\\\le&
	 J^{\eta,\infty}\big( \xi^{\theta,\beta,\nu,\eps} \big)
	- J^0( \theta )
	+ J^0( \hat\theta )
	- \delta
	- J^0\big( 0, 0, V^{\hat\xi^{\eta,\infty}} \big)
	\\\le&
	 - \delta
	+ \frac \eta 2 \E\bigg[ \int^T_0 \big( \xi^{\theta,\beta,\nu,\eps}_t \big)^2 \, \dt \bigg]
	+ \Big|
		J^0( \theta )
		- J^0\big( 0, 0, V^{\xi^{\theta,\beta,\nu,\eps}} \big)
	\Big|
	+ \Big|
		J^0\big( \hat\theta \big)
		- J^0\big( 0, 0, V^{\hat\xi^{\eta,\infty}} \big)
	\Big|
	\\\le&
	 - \delta
	+ \frac \eta 2 \E\bigg[ \int^T_0 \big( \xi^{\theta,\beta,\nu,\eps}_t \big)^2 \, \dt \bigg]
	+ D \bigg(\E\bigg[ \int^T_0 \big( X^\theta_t \big)^2 \, \dt \bigg] \bigg)
	\cdot p\bigg( \E\bigg[ \int^T_0 \big( X^\theta_t - X^{\theta,\beta,\nu,\eps}_t \big)^2 \, \dt \bigg] \bigg)
	\\&
	+ D \bigg(\E\bigg[ \int^T_0 \big( X^{\hat\theta}_t \big)^2 \, \dt \bigg) \bigg]
	\cdot p\bigg( \E\bigg[ \int^T_0 \big( \hat X^0_t - \hat X^{\eta,\infty}_t \big)^2 \, \dt \bigg] \bigg).
\end{align*}
According to Lemma \ref{thmEintDeltaXsqconv0lemma}, we can first choose $\beta,\nu,\eps\in (0,\infty)$ and then, due to $|\hat X^{\eta,\infty}_t| \le x_0$ (cf. Theorem \ref{thmXYEtaPosBdthm}), choose $\eta \in (0,\infty)$ small enough according to Theorem \ref{thmxConv0Specificthm} such that the sum is negative. This contradicts to the optimality of $\hat\xi^{\eta,\infty}$ for $J^{\eta,\infty}$. This proves the optimality of $\hat\theta$.

It remains to compute $J^0 ( \hat\theta )$. In view of \eqref{eqJEtaEqJ0PlusEIntXiSq}, for all $\eta \in (0,\infty)$, we have
\begin{align*}
	J^0 \big( \hat\theta \big)
	=
	J^0 \big( \hat\theta \big)
	- J^0 \big( 0, 0, V^{\hat\xi^{\eta,\infty}} \big)
	- \frac \eta 2 \E\bigg[ \int^T_0 \big( \xi^{\eta,\infty}_t \big)^2 \, \dt \bigg]
	+ J^{\eta,\infty} \big( \hat\xi^{\eta,\infty} \big).
\end{align*}
The difference of the first two terms converge to zero as $\eta \to 0$, which is verified using Lemma \ref{thmJetaconvJ0lemma} and Theorem \ref{thmxConv0Specificthm} as in the proof of Theorem \ref{thmthetaoptimalJ0thm}. The third term converges to zero by Lemma \ref{thmetaEintxisqconv0lemma}. Hence, using the representation of the value function given in \cite{graewe_horst},
\begin{align*}
	J^0 \big( \hat\theta \big)
	=&
	\lim_{\eta \to 0} J^{\eta,\infty} \big( \hat\xi^{\eta,\infty} \big)
	=
	\lim_{\eta \to 0} \frac 1 2 A^{\eta,\infty}_0 (x_0)^2
	=
	\frac {(x_0)^2} 2 \lim_{\eta \to 0} \big(
		\sqrt{\eta} D^{\eta,\infty}_0
		+ \gamma B^{\eta,\infty}_0
	\big)
	=
	\frac 1 2 \gamma (x_0)^2 B^0_0.
\end{align*}

\end{proof}

\begin{remark}
If $\lambda \equiv 0$, then our model is a special case of the model analyzed in \cite{ackermann_kruse_urusov}, which also contains cases when there is no optimal trading strategy. However, since $\gamma$ is constant in our model, the processes ``$\mu$'' and ``$\sigma$'' introduced in \cite{ackermann_kruse_urusov} are equal to zero. This implies that the equation ``$\tilde\beta = Y$'' holds in their notation. As shown in Section 5 of \cite{ackermann_kruse_urusov}, the process ``$M^\perp$'' introduced therein is also equal to zero, which implies that ``$Y$'' and hence ``$\tilde\beta$'' is a semimartingale. Theorem 2.3 (ii) in \cite{ackermann_kruse_urusov} confirms that, under this property, an optimal trading strategy does indeed exist.
\end{remark}


\section{Proofs for Section \ref{SectModelEtaPositive}} \label{SectProofsPart1}

This section proves the results stated in Section \ref{SectModelEtaPositive}. We start with a priori estimates and regularity properties for the coefficient processes that specify the optimal state processes.

\subsection{A priori estimates and regularity properties} \label{SectAprioriRegularity}


\subsubsection{The case $\eta > 0$} \label{SectEtaPos}

The following estimates have been established in \cite{graewe_horst,horst_xia}, except for the upper bound on $E^{\eta,N}$ for finite $N$, which is stronger than the corresponding one in \cite{horst_xia}. It can be established using the same arguments as in the proof of Proposition 3.2 in \cite{graewe_horst} noting that $D^{\eta,N}_T < \infty$ if $N$ is finite.

\begin{lemma} \label{thmBDEEtaPosApriorilemma}
Let
$$\overline\kappa := \sqrt{2 \max(\overline\lambda,\gamma\overline\rho)}.$$
For all $\eta \in (0,\infty)$, $N\in \mathcal N$ and $s,t \in [0,T)$ with $s \le t$, we have that
\begin{align*}
	 e^{ -\overline\rho (T-s) }
	&\le
	B^{\eta,N}_s
	\le
	1,
	\\
	0
	&\le E^{\eta,N}_s
	\le \gamma^{-1} \overline\kappa \tanh\big( \sqrt{\eta^{-1}} \overline\kappa (T-s) \big)
	\le \gamma^{-1} \overline\kappa ,
	\quad
	\\
	0
	&< D^{\eta,N}_s
	\le \overline\kappa \coth\big(\sqrt{\eta^{-1}} \overline\kappa (T-s)\big)
	\le \overline\kappa \coth\big(\sqrt{\eta^{-1}} \overline\kappa (T-t)\big).
\end{align*}
\end{lemma}

The preceding estimates allow us to prove that neither the optimal portfolio process nor the corresponding spread process change sign.
\begin{proof}[Proof of Theorem \ref{thmXYEtaPosBdthm}]
Let us put $V(t) := \hat X^{\eta,N}_t \cdot \hat Y^{\eta,N}_t$. Then $V(0) = 0$ and $V$ satisfies the ODE
\begin{align*}
	\dot V(t)
	=&
	-\Big(
		\sqrt{\eta^{-1}} \big(
			D^{\eta,N}_t
			+ \gamma E^{\eta,N}_t
		\big)
		+ \rho_t
	\Big) V(t)
	- \sqrt{\eta^{-1}} \Big(
		E^{\eta,N}_t \big(\hat Y^{\eta,N}_t\big)^2
		+ \gamma D^{\eta,N}_t \big(\hat X^{\eta,N}_t\big)^2
	\Big).
\end{align*}
As a result,
\begin{align*}
	V(t)
	=&
	- \exp\bigg( - \int^t_0 \Big( \sqrt{\eta^{-1}} \big( D^{\eta,N}_u + \gamma E^{\eta,N}_u \big) + \rho_u \Big) \, \du \bigg) \cdot
	\\& \qquad
	\int^t_0 \sqrt{\eta^{-1}} \Big( E^{\eta,N}_s \big(\hat Y^{\eta,N}_s\big)^2 + \gamma D^{\eta,N}_s \big(\hat X^{\eta,N}_s\big)^2 \Big)
	\\& \qquad\qquad
	\cdot \exp\bigg( \int_0^s \Big( \sqrt{\eta^{-1}} \big( D^{\eta,N}_u + \gamma E^{\eta,N}_u \big) + \rho_u \Big) \, \du \bigg) \, \ds.
\end{align*}

In view of Lemma \ref{thmBDEEtaPosApriorilemma}, this shows that $V(t) < 0$ on $(0,T)$. Hence, strict positivity of $\hat X^{\eta,N}_0$ yields
\begin{align*}
	\IP\big[
		\hat X^{\eta,N}_t > 0,
		\hat Y^{\eta,N}_t < 0
		\text{ for all } t \in (0,T)
	\big]
	= 1.
\end{align*}
Thus, the definition of $\hat Z^{\eta,N}$ along with \eqref{eqOdeXY} yields $\ddt \hat Z^{\eta,N}_t = \rho_t \hat Y^{\eta,N}_t < 0$ on $(0,T)$. Moreover,
\begin{align*}
	&
	\hat X^{\eta,N}_t <
	\gamma^{-1} \hat Z^{\eta,N}_t
	\le \gamma^{-1} \hat Z^{\eta,N}_0
	= x_0
	\quad \mbox{and} \quad
	\hat Y^{\eta,N}_t
	>
	- \hat Z^{\eta,N}_t
	\ge
	- \gamma x_0.
\end{align*}
\end{proof}

Next, we prove that the process $B^{\eta,N}$ satisfies an $L^1$ uniform continuity property. We refer to Appendix \ref{SectAppendixRegularityItoBsde} for a discussion of general regularity properties of stochastic processes.

\begin{lemma} \label{thmBweakstarpropsatisfiedlemma}
The process $B^{\eta,N}$ satisfies Condition \hyperlink{C1}{C.1} stated in Appendix \ref{SectAppendixRegularityItoBsde} on $[0,T]$ uniformly in $N \in \mathcal N$ and $\eta \in (0,H]$, for all $H \in (0,\infty)$.
\end{lemma}

\begin{proof}
Let $\eps, \eps_1 > 0$, $s > T-\eps_1$ and let $V$ and $\tau$ be arbitrary according to the definition of Condition \hyperlink{C1}{C.1}. By Lemma \ref{thmBDEEtaPosApriorilemma}, if $\eps_1$ is small enough,
\begin{align*}
	\bigg| \E\Big[ V \big( B^{\eta,N}_\tau - B^{\eta,N}_s \big) \Big] \bigg|
	\le \Big( 1 - \exp\big( -\overline\rho \eps_1 \big) \Big) \E\big[ |V| \big]
	\le \eps \E\big[ |V| \big].
\end{align*}

If $s \leq T - \eps_1$, the assertion follows from the integral representation \eqref{eqSdeB} along with the estimates established in Lemma \ref{thmBDEEtaPosApriorilemma} using that the stochastic integral in \eqref{eqSdeB} is a martingale on $[0,T-\eps_1]$. We emphasize that $Z^{\eta,N,B}$ is possibly defined only on $[0,T-]$ and so the stochastic integral may be a martingale only away from the terminal time.
\end{proof}


\subsubsection{The case $\eta=0$} \label{SectEta0}

We are now going to establish a priori estimates on the candidate limiting coefficient processes.
First, we show that Assumption \ref{assFactorModel} directly implies the following regularity result for the parameter processes:
\begin{lemma} \label{thmChiRhoLambdaPhilemma}
The processes $\chi$, $\rho$, $\lambda$, $\varphi$ and $\varphi^{-1}$ satisfy Condition \hyperlink{C2}{C.2} introduced in Appendix \ref{SectAppendixRegularityItoBsde}.
\end{lemma}
\begin{proof}
$\chi$ satisfies Condition \hyperlink{C2}{C.2} due to Lemma \ref{thmEcondmaxincXbdlemma}, Lemma \ref{thmPropertyStarComponentslemma} and Lemma \ref{thmSumCondC1C2lemma}. The rest immediately follows by Lemma \ref{thmftWtpropstarlemma}.
\end{proof}

We are now ready to prove that the process $B^0$ is well-defined.
\begin{proof}[Proof of Lemma \ref{thmexistsB0lemma}]
The existence result follows from a standard argument. In fact, it is well known that, for any $\overline b \in [1,\infty)$, the BSDE
\begin{align}
	- \, \mathrm dB^{\overline b}_t
	=
	\psi^{\overline b} \big( t, \chi_t , B^{\overline b}_t \big) \, \dt
	- Z^{\overline b}_t \, \mathrm dW_t,
	\quad
	B^{\overline b}_T = 1
	\label{eqthmexistsB0lemmaBsde}
\end{align}
with Lipschitz continuous driver (we recall $f^\lambda, f^\rho$ that are defined in Assumption \ref{assFactorModel})
\begin{align*}
	\psi^{\overline b}(t,\chi,b)
	:=&
	- \frac{1}{ f^\lambda(t,\chi) + 2 \gamma f^\rho(t,\chi) } \bigg(
		\gamma f^\rho(t,\chi)^2 \big(\overline b^2 \wedge b^2\big)
		- 2 f^\lambda(t,\chi) f^\rho(t,\chi) \Big(
			1
			- \big( 0 \vee (b\wedge \overline b) \big)
		\Big)
	\bigg)
\end{align*}
has a unique solution $(B^{\overline b},Z^{\overline b}) \in \mathcal L_{\mathcal P}^2(\Omega\times[0,T];\R \times \R^m)$ (cf. Theorem \ref{thmBsdeSolExistsUniqueWithHthm}).
Let us then define the functions
$\underline\phi, \overline\phi \colon \Omega \times [0,T] \times \R \times \R^m \to \R$ by
\[
	\underline\phi(t,b,z) := - \underline\varphi^{-2} \gamma \overline\rho^2 b \quad \mbox{and} \quad \overline\phi(t,b,z) := 0.
\]
By definition, $(\underline{B}^0,0)$ is the unique solution to the BSDE with driver $\underline\phi$ and terminal condition $1$, where the lower bound $\underline{B}^0$ on the process $B^0$ was defined in \eqref{lower-bound-B0}. Likewise, $(1,0)$ is the unique solution of the BSDE with driver $\overline\phi$ and the same terminal condition. Since
\[
	\psi^{\overline b}\big(t,\chi_t,\underline{B}^0_t\big) > \underline\phi\big(t,\underline{B}^0_t,0\big) \quad \mbox{and} \quad \psi^{\overline b}(t,\chi_t,1) < \overline\phi(t,1,0),
\]
the standard comparison principle for BSDEs with Lipschitz continuous drivers yields
\[
	\underline{B}^0 \leq B^{\overline b} \leq 1.
\]
This proves that $(B^{1},Z^{1})$ is the desired unique bounded solution to the BSDE \eqref{eqBsdeB0}.

The second assertion follows Theorem \ref{thmBsdeSolExistsUniqueWithHthm} applied to the BSDE \eqref{eqthmexistsB0lemmaBsde} for $\overline b = 1$.
\end{proof}

Having established the existence of the process $B^0$, the processes $D^0$ and $E^0$ are well-defined. The following lemma establishes estimates and regularity properties for $D^0$ and $E^0$.
\begin{lemma} \label{thmD0E0Apriorilemma}
The following a priori estimates hold:
\begin{align*}
	\frac{\gamma \underline\rho}{\overline\varphi} \exp\bigg( - \frac{\gamma \overline\rho^2}{\underline\varphi^{2}} T \bigg)
	&\le D^0_t
	\le \underline\varphi^{-1} \big( \overline\lambda + \gamma \overline\rho \big),
	\\
	\overline\varphi^{-1} \underline\rho
	&\le E^0_t
	\le 2 \underline\varphi^{-1} \overline\rho.
\end{align*}
Moreover, the processes $D^0$ and $E^0$ satisfy Condition \hyperlink{C2}{C.2}.
\end{lemma}
\begin{proof}
The a priori estimates can be obtained by plugging the bounds on $B^0$ (cf. Lemma \ref{thmexistsB0lemma}) into the definitions of $D^0$ and $E^0$ given in \eqref{eqD0def}.
Moreover, if we denote by $h$ the function derived from Lemma \ref{thmexistsB0lemma}, then for all $t \in [0,T]$,
\begin{align*}
	&D^0_t
	= \bigg( \frac{ f^\lambda + \gamma \cdot f^\rho \cdot h }{ \sqrt{ f^\lambda + 2 \gamma \cdot f^\rho } } \bigg) (t,\chi_t),
	\quad
	E^0_t
	= \bigg( \frac{ f^\rho \cdot (2 - h) }{ \sqrt{ f^\lambda + 2 \gamma \cdot f^\rho } } \bigg) (t,\chi_t).
\end{align*}
In view of Assumption \ref{assFactorModel}, the processes $D^0$ and $E^0$ can be represented as uniformly continuous functions of the factor process $\chi$ and hence the assertion follows from Lemma \ref{thmChiRhoLambdaPhilemma} and Lemma \ref{thmftWtpropstarlemma}.
\end{proof}

The next lemma can be viewed as the analogue to Theorem \ref{thmXYEtaPosBdthm} in the case $\eta = 0$.
\begin{lemma} \label{thmXY0BdStarStarlemma}
It holds
\begin{align*}
	\IP\big[
		\hat X^0_t \in (0,x_0),
		\hat Y^0_t \in (-\gamma x_0,0)
		\text{ for all } t \in (0,T)
	\big]
	= 1.
\end{align*}
\end{lemma}
\begin{proof}
Due to Lemma \ref{thmD0E0Apriorilemma}, $D^0$ is positive and hence, for $t \in (0,T)$,
\begin{align*}
	0
	<
	\hat Z^0_t
	<
	x_0 \gamma.
\end{align*}
Since $D^0 + \gamma E^0 = \varphi$ on $[0,T)$, it thus follows from Lemma \ref{thmD0E0Apriorilemma} that, for all $t\in (0,T)$,
\begin{align*}
	0 < \hat X^0_t
	& = \frac{\varphi_t - D^0_t}{\gamma \varphi_t} \hat Z^0_t
	< \gamma^{-1} \hat Z^0_t < x_0, \\
	0 > \hat Y^0_t
	& = - \frac{\varphi_t - \gamma E^0_t}{\varphi_t} \hat Z^0_t
	> - \hat Z^0_t
	> -\gamma x_0.
\end{align*}
\end{proof}


\subsection{Proof of the convergence results} \label{SectConv}

In this section, we prove our main convergence results. We start with the convergence of the coefficients of the ODE system \eqref{eqOdeXY}. Subsequently, we prove that the convergence of the coefficients yields convergence of the state process.

\subsubsection{Proof of Proposition \ref{thmbdeconv0prop}} \label{SectBDEFconv}

The proof of Proposition \ref{thmbdeconv0prop} is split into a series of lemmas. In a first step, we establish the convergence as $\eta \to 0$ of the auxiliary processes
\begin{align*}
	F^{\eta,N}_t
	:=
	D^{\eta,N}_t + \gamma E^{\eta,N}_t \quad \mbox{to} \quad F^0 := \varphi
\end{align*}
and
\begin{align*}
	G^{\eta,N} := \rho B^{\eta,N} + \varphi E^{\eta,N} \quad \mbox{to} \quad
	G^0
	:=
	\rho B^0 + \varphi E^0
	=
	2 \rho.
\end{align*}

On $[0,T)$, the processes $F^{\eta,N}$ and $G^{\eta,N}$ satisfy the dynamics
\begin{equation}
\begin{split}
	\mathrm dF^{\eta,N}_t
	=&
	\sqrt{\eta^{-1}} \Big(
		\big( F^{\eta,N}_t \big)^2
		- (\varphi_t)^2
	\Big) \, \dt
	+ 2 \gamma \rho_t E^{\eta,N}_t \, \dt
	\\&
	\quad + \sqrt{\eta^{-1}} \big(
		Z^{\eta,N,A}_t
		- 2 \gamma Z^{\eta,N,B}_t
		+ \gamma^2 Z^{\eta,N,C}_t
	\big) \, \mathrm dW_t,
\label{eqBsdeF}
\end{split}
\end{equation}
respectively,
\begin{equation}
\begin{split}	
		- \, \mathrm d \big( \varphi^{-1} G^{\eta,N} \big)_t
	=&
	\sqrt{\eta^{-1}} \Big(
		- G^{\eta,N}_t
		+ 2 \rho_t
		- E^{\eta,N}_t \big( F^{\eta,N}_t - \varphi_t \big)
	\Big) \, \dt
	- 2 \rho_t E^{\eta,N}_t \, \dt
	\\&
	\quad - \, \mathrm d \big( \varphi^{-1} \rho B^{\eta,N} \big)_t
	+ \sqrt{\eta^{-1}} \big( Z^{\eta,N,B}_t - \gamma Z^{\eta,N,C}_t \big) \, \mathrm dW_t.
	\label{eqBsdeG}
\end{split}
\end{equation}

A general convergence result for integral equations of the above form is established in Appendix \ref{SectAppendixConvSdeScaled}. It allows us to prove the following two lemmas.
\begin{lemma} \label{thmfconv0lemma}
Let $f^{\eta,N} := F^{\eta,N} - \varphi$. For all $\eps>0$, there exists an $\eta_0>0$ such that, for all $\eta\in(0,\eta_0]$ and all $N \in \mathcal N$,
\begin{align*}
	\IP \Big[
		\sup_{s \in [0,T-\eps]} f^{\eta,N}_s \le \eps,
		\inf_{s \in [0,T)} f^{\eta,N}_s \ge -\eps
	\Big]
	=
	1.
\end{align*}
\end{lemma}

\begin{proof}
For every $N \in \mathcal N$,
we apply Lemma \ref{thmfggeneralconvlemma} to $P^\eta := F^{\eta,N}$ with
\begin{align*}
	\psi &:= 1, \quad a(x,y) := (x\vee 0)^2 - y^2, \quad P^0 := \varphi,
	\\
	q^{\eta} &:= 0, \quad
	L^{\eta}_t := 2\gamma \int^t_0 \rho_s E^{\eta,N}_s \, \ds,
\end{align*}
noticing that $\eps \mapsto \eta_0 = \eta_0(\eps)$ is independent of $N$.
Assumption \ref{assthmfggeneralconvlemma} \textit{i)} is satisfied with $\overline P := \overline\varphi + 3 \overline\kappa$, due to Lemma \ref{thmBDEEtaPosApriorilemma}.
If $N = \infty$, then Assumption \ref{assthmfggeneralconvlemma} \textit{ii)} \textit{a)} is satisfied due to Theorem \ref{thmgraewexiathm} and due to the a priori estimates derived in Lemma \ref{thmBDEEtaPosApriorilemma} and Lemma \ref{thmD0E0Apriorilemma}.
If $N<\infty$ and $\eta \le \overline \varphi^{-2}$, then
\begin{align*}
	f^{\eta,N}_T
	= D^{\eta,N}_T + \gamma \underbrace{E^{\eta,N}_T}_{=0} - D^0_T - \gamma E^0_T
	= \sqrt{\eta^{-1}} (N - \gamma) - \varphi_T
	\ge \sqrt{\eta^{-1}} - \overline\varphi
	> 0,
\end{align*}
which shows Assumption \ref{assthmfggeneralconvlemma} \textit{ii)} \textit{b)}.
Assumption \ref{assthmfggeneralconvlemma} \textit{iv)} follows by direct computation using $\overline\eps := \underline\varphi/2$ and $\beta := \underline\varphi/2$.
Assumption \ref{assthmfggeneralconvlemma} \textit{v)} follows from Lemma \ref{thmChiRhoLambdaPhilemma} and Lemma \ref{thmBDEEtaPosApriorilemma}.
\end{proof}

It is not difficult to show that similar arguments as those used to prove the converegence of $f^{\eta,N}$ can be applied to $-g^{\eta,N} := -G^{\eta,N} + 2 \rho$. As a result, the intervals $[0,T)$ and $[0,T-\eps]$ in the statement of the convergence result for $g^{\eta,N}$ need to be swapped.

\begin{lemma} \label{thmgconv0lemma}
For all $\eps>0$, there exists an $\eta_0>0$ such that, for all $\eta\in(0,\eta_0]$ and all $N \in \mathcal N$,
\begin{align*}
	\IP \Big[
		\sup_{s \in [0,T)} g^{\eta,N}_s \le \eps,
		\inf_{s \in [0,T-\eps]} g^{\eta,N}_s \ge -\eps
	\Big]
	=
	1.
\end{align*}
\end{lemma}

We are now going to prove the almost sure convergence to zero of the process $b^{\eta,N} := B^{\eta,N} - B^0$. To this end, we first observe that
\begin{align*}
	D^{\eta,N}
	=
	F^{\eta,N} - \gamma E^{\eta,N}
	\quad \mbox{and} \quad
	\varphi E^{\eta,N}
	=
	G^{\eta,N} - \rho B^{\eta,N}
\end{align*}
yields
\begin{align*}
	&
	- D^{\eta,N} E^{\eta,N}
	=
	- f^{\eta,N} E^{\eta,N}
	- G^{\eta,N}
	+ \rho B^{\eta,N}
	+ \gamma \varphi^{-2} \big( G^{\eta,N} - \rho B^{\eta,N} \big)^2.
\end{align*}
Plugging this into \eqref{eqSdeB} shows that
\begin{align*}
	\mathrm dB^{\eta,N}_t
	=&
	\Big(
		- G^{\eta,N}_t
		- f^{\eta,N}_t E^{\eta,N}_t
		+ 2 \rho_t B^{\eta,N}_t
		+ \gamma (\varphi_t)^{-2} \big( G^{\eta,N}_t - \rho_t B^{\eta,N}_t \big)^2
	\Big) \, \dt
	+ Z^{\eta,N,B}_t \, \mathrm dW_t
\end{align*}
on $[0,T)$.
Performing an analogous computation for $B^0$ and subtracting the two equations yields
\begin{align*}
	\mathrm db^{\eta,N}_t
	=&
	\bigg(
		\big( b^{\eta,N}_t \big)^2 \cdot \gamma (\varphi_t)^{-2} (\rho_t)^2
		+ b^{\eta,N}_t \Big(
			2 \rho_t
			- 2 \gamma (\varphi_t)^{-2} \rho_t \big( G^0_t - \rho_t B^0_t \big)
		\Big)
		- f^{\eta,N}_t E^{\eta,N}_t
		- g^{\eta,N}_t
	\\&
		+ \gamma (\varphi_t)^{-2} g^{\eta,N}_t \Big(
			g^{\eta,N}_t + 2 G^0_t
			- \rho_t \big( 2 b^{\eta,N}_t + 2 B^0_t \big)
		\Big)
	\bigg) \, \dt
	+ \big( Z^{\eta,N,B}_t - Z^{0,B}_t \big) \, \mathrm dW_t
\end{align*}
on $[0,T)$.
Using that
\begin{align*}
	2 \rho_t
	- 2 \gamma (\varphi_t)^{-2} \rho_t \big( G^0_t - \rho_t B^0_t \big)
	& = 2 \rho_t (\varphi_t)^{-1} D^0_t,
	\\
	g^{\eta,N}_t + 2 G^0_t
	- \rho_t \big( 2 b^{\eta,N}_t + 2 B^0_t \big)
	& =
	- g^{\eta,N}_t
	+ 2 \varphi_t E^{\eta,N}_t,
\end{align*}
shows that
\begin{align*}
	\mathrm db^{\eta,N}_t
	=&
	\Big(
		\big( b^{\eta,N}_t \big)^2 \cdot \gamma (\varphi_t)^{-2} (\rho_t)^2
		+ b^{\eta,N}_t \cdot 2 (\varphi_t)^{-1} \rho_t D^0_t
		- f^{\eta,N}_t E^{\eta,N}_t
	\\&
		+ g^{\eta,N}_t \big(
			- 1
			- \gamma (\varphi_t)^{-2} g^{\eta,N}_t
			+ 2 \gamma (\varphi_t)^{-1} E^{\eta,N}_t
		\big)
	\Big) \, \dt
	+ \big( Z^{\eta,N,B}_t - Z^{0,B}_t \big) \, \mathrm dW_t
\end{align*}
on $[0,T)$.
This BSDE is different from \eqref{eqBsdeF} and \eqref{eqBsdeG}. We apply Lemma \ref{thmbGeneralConvlemma} to prove the following result.

\begin{lemma} \label{thmbconv0lemma}
For all $\eps>0$, there exists an $\eta_0>0$ such that, for all $\eta\in(0,\eta_0]$ and all $N \in \mathcal N$,
\begin{align*}
	\IP \Big[ \sup_{s \in [0,T]} \big| b^{\eta,N}_s \big| \le \eps \Big]
	=
	1.
\end{align*}
\end{lemma}
\begin{proof}
For every $N \in \mathcal N$,
we apply Lemma \ref{thmbGeneralConvlemma} with
\begin{align*}
	a(t,b)
	&:=
	b^2 \cdot \gamma (\varphi_t)^{-2} (\rho_t)^2
	+ b \cdot 2 (\varphi_t)^{-1} \rho_t D^0_t,
	\\
	q^{\eta}_t
	&:=
	- f^{\eta,N}_t E^{\eta,N}_t
	+ g^{\eta,N}_t \big(
		2 \gamma (\varphi_t)^{-1} E^{\eta,N}_t
		- 1
		- \gamma (\varphi_t)^{-2} g^{\eta,N}_t
	\big).
\end{align*}
Assumption \ref{assthmbGeneralConvlemma} \textit{i)} follows from the a priori estimates on $B^{\eta,N}$ (Lemma \ref{thmBDEEtaPosApriorilemma}), $B^0$ (Lemma \ref{thmexistsB0lemma}) and $D^0$ (Lemma \ref{thmD0E0Apriorilemma}). Assumption \ref{assthmbGeneralConvlemma} \textit{ii)} follows from the a priori estimates on $B^{\eta,N}$ and $B^0$, where the mapping $\eps \mapsto \delta$ is independent of $N$.
Assumption \ref{assthmbGeneralConvlemma} \textit{iii)} follows from the same estimates and Lemma \ref{thmfconv0lemma} and Lemma \ref{thmgconv0lemma}, where the choice of $\eta_1$ is independent of $N$.
Assumption \ref{assthmbGeneralConvlemma} \textit{iv)} is satisfied because we can choose $\eps_1 > 0$ small enough s.t.
$$
-\eps_1 \gamma (\varphi)^{-1} \rho
+ 2 D^0
\ge 0.
$$
\end{proof}

\subsubsection{Proof of Theorem \ref{thmxConv0Specificthm}}
\label{SectXYconv}

First we need to prove an auxiliary result:
\begin{lemma} \label{thmStarStarlemma}
Let $I \subset \R$ be a compact interval and $Y = \{Y_t\}_{t\in I}$ a continuous, adapted, $\R^d$-valued stochastic process. Then the modulus of continuity
\[
	\sup_{s,t\in I, |s-t| \le \nu} \Vert Y_t - Y_s \Vert_\infty
\]
converges to $0$ in probability as $\nu \to 0$.
\end{lemma}
\begin{proof}
For all $\omega\in\Omega$, $Y(\omega)$ is uniformly continuous, hence the modulus of continuity converges to $0$ as $\nu\to 0$ $\IP$-a.s. This implies convergence in probability, in particular.
\end{proof}

\begin{proof}[Proof of Theorem \ref{thmxConv0Specificthm}]
Convergence on compact subintervals of $[0,T)$ follows from Theorem \ref{thmXConv0Generalthm} applied to the ODEs
\begin{align}
\begin{split}
	\ddt \hat X^{\eta,N}_t
	=&
	- \sqrt{\eta^{-1}} F^{\eta,N}_t \bigg(
		\hat X^{\eta,N}_t
		- \frac{E^{\eta,N}_t}{F^{\eta,N}_t} \hat Z^{\eta,N}_t
	\bigg),
	\\
	\ddt \hat Z^{\eta,N}_t
	=&
	\gamma \rho_t \hat X^{\eta,N}_t - \rho_t \hat Z^{\eta,N}_t
\end{split}
\label{eqthmxConv0SpecificthmOdeXYEtaPos}
\end{align}
and
\begin{align*}
\begin{split}
	\hat X^0_t
	=
	\frac{E^0_t}{\varphi_t} \hat Z^0_t, \quad
	\hat Z^0_t
	=
	\hat Z^0_0 \exp\bigg( \int^t_0 \bigg( \gamma \rho_s \frac {E^0_s}{\varphi_s} - \rho_s \bigg) \, \ds \bigg)
\end{split}
\end{align*}
on $[0,T)$.
In fact, Condition \eqref{eqassOdeConvABConv} follows from Proposition \ref{thmbdeconv0prop} and Condition \eqref{boundXZ} follows from Theorem \ref{thmXYEtaPosBdthm}. Finally, $D^0 > 0$ (Lemma \ref{thmD0E0Apriorilemma}), from which we deduce that $x_0 > \frac{E^0_0}{\varphi_0} \hat Z^0_0$.

Thus, by Theorem \ref{thmXConv0Generalthm}, applied once for every $N \in \mathcal N$ and every $\nu > 0$, for all $\nu,\delta \in (0,\infty)$, there exists $\eta_1 = \eta_1(\nu,\delta) \in (0,\infty)$ such that, for all $\eta \in (0,\eta_1]$ and all $N \in \mathcal N$, the set
\begin{align*}
	M_{0}^{\nu,\eta,N}
	:=&
	\bigg\{
		\sup_{t\in [0,T-\nu]} \max\Big(
			\big| \hat Z^0_t - \hat Z^{\eta,N}_t \big|,
			\hat X^0_t - \hat X^{\eta,N}_t
		\Big) \le \nu,
		\inf_{t\in [\nu,T-\nu]} \big( \hat X^0_t - \hat X^{\eta,N}_t \big) \ge -\nu
	\bigg\}
\end{align*}
satisfies
\begin{align}
\IP\big[ M_{0}^{\nu,\eta,N} \big] > 1-\delta/2.
\label{eqthmxConv0SpecificthmPM2GeEtc}
\end{align}
Near the terminal time, we cannot use this theorem since the convergence $\frac{E^{\eta,N}}{F^{\eta,N}} \to \frac{E^0}{\varphi}$ holds only on $[0,T-]$. Instead, we apply increment bounds on intervals of the form $[T-\nu,T]$ to prove that,
for all $\eta > 0$, $N \in \mathcal N$ and for all sufficiently small
\begin{align}
	\nu \in (0,\eps/2] \quad\quad \text{ and } \quad\quad \beta \in (0,\eps/3],
	\label{eqthmxConv0SpecificthmNuLeEpsDiv2BetaLeEpsDiv3}
\end{align}
the set
\begin{align*}
	M_{1}^{\nu,\eps,\eta,N}
	:=&
	\Big\{ \inf_{t \in [T-\nu,T)} \big( \hat X^0_t - \hat X^{\eta,N}_t \big) \ge -\eps \Big\}
\end{align*}
contains the set $M_{0}^{\nu,\eta,N} \cap M_2^{\beta,\nu} \cap M_{3}^{\nu,\eta,N}$, where
\begin{align*}
	M_2^{\beta,\nu}
	:=&
	\Big\{
		\sup_{s,t\in [T-\nu,T)}
		| \varphi_s - \varphi_t |
		\vee \big| E^0_s - E^0_t \big|
		\vee \big| \hat X^0_s - \hat X^0_t \big|
		\le \beta
	\Big\},
	\\
	M_{3}^{\nu,\eta,N}
	:=&
	\Big\{
		\sup_{t \in [T-\nu,T)} e^{\eta,N}_t \vee \big(-f^{\eta,N}_t\big) \le \nu
	\Big\}.
\end{align*}
The probability of the last two events can be made large, due to Lemma \ref{thmStarStarlemma} and Proposition \ref{thmbdeconv0prop}: for all $\beta,\delta,\nu \in (0,\infty)$, there exist $\nu_0(\beta,\delta), \eta_2(\nu) \in (0,\infty)$ such that
\begin{align*}
	\inf_{\nu \le \nu_0(\beta,\delta)}
	\IP\big[M_2^{\beta,\nu}\big]
	>
	1 - \delta/2,
	\quad
	\inf_{\eta \le \eta_2(\nu)}
	\inf_{N \in \mathcal N}
	\IP\big[M_{3}^{\nu,\eta,N}\big]
	=
	1.
\end{align*}

In order to see that $M_{1}^{\nu,\eps,\eta,N} \supseteq M_{0}^{\nu,\eta,N} \cap M_2^{\beta,\nu} \cap M_{3}^{\nu,\eta,N}$, we assume to the contrary that there exists
$$
\omega \in \big( M_{0}^{\nu,\eta,N} \cap M_2^{\beta,\nu} \cap M_{3}^{\nu,\eta,N} \big)
\backslash M_1^{\nu,\eps,\eta,N}.
$$
To obtain the desired contradiction, we show that, for any such $\omega$, there exists a time $s(\omega) \in [T-\nu,T)$ for which we can deduce both non-negativity and strict negativity of $\ddt \hat X^{\eta,N}_{s(\omega)}(\omega)$ simultaneously.

We start with the choice of $s(\omega)$. Since $\omega \notin M_1^{\nu,\eps,\eta,N}$, there exists some $t \in [T-\nu,T)$ such that $\hat X^{\eta,N}_t(\omega) - \hat X^0_t(\omega) > \eps$. Since $\omega \in M_2^{\beta,\nu}$ and due to \eqref{eqthmxConv0SpecificthmNuLeEpsDiv2BetaLeEpsDiv3}, this yields that
\begin{align*}
	\hat X^{\eta,N}_t(\omega)
	> \hat X^0_{T-\nu}(\omega) + \frac {2\eps} 3.
\end{align*}
Since $\omega$ belongs to $M_{0}^{\nu,\eta,N}$ and due to \eqref{eqthmxConv0SpecificthmNuLeEpsDiv2BetaLeEpsDiv3}, this implies that
\begin{align*}
	\hat X^{\eta,N}_{T-\nu}(\omega)
	< \hat X^0_{T-\nu}(\omega) + \frac{2\eps} 3 < \hat X^{\eta,N}_t(\omega).
\end{align*}
Now we can choose $s(\omega) \in (T-\nu,t)$ minimal with the property that
\begin{align}
	\hat X^{\eta,N}_{s(\omega)}(\omega)
	=
	\hat X^0_{T-\nu}(\omega) + \frac{2\eps}3.
	\label{eqthmxConv0SpecificthmXEtaEqX0PlusEps}
\end{align}
Due to minimality of $s(\omega)$, we have $\ddt \hat X^{\eta,N}_{s(\omega)}(\omega) \ge 0$.

We now show that this derivative must also be strictly negative.
In fact, due to \eqref{eqthmxConv0SpecificthmOdeXYEtaPos},
\begin{align*}
	\sqrt\eta \ddt \hat X^{\eta,N}_{s(\omega)}(\omega)
		=&
	- F^{\eta,N}_{s(\omega)}(\omega) \hat X^{\eta,N}_{s(\omega)}(\omega) + E^{\eta,N}_{s(\omega)}(\omega) \hat Z^{\eta,N}_{s(\omega)}(\omega).
\end{align*}
Since $\hat Z^{\eta,N}$ is non-increasing (Theorem \ref{thmXYEtaPosBdthm}), using \eqref{eqthmxConv0SpecificthmXEtaEqX0PlusEps} again, the right hand side of the above equation can be bounded from above by
\begin{align*}
	&
	- \varphi_{T-\nu}(\omega) \hat X^0_{T-\nu}(\omega)
	+ \big( \varphi_{T-\nu}(\omega) - \varphi_{s(\omega)}(\omega) \big) \hat X^0_{T-\nu}(\omega)
	- f^{\eta,N}_{s(\omega)}(\omega) \hat X^0_{T-\nu}(\omega)
	\\&
	- F^{\eta,N}_{s(\omega)}(\omega) \frac{2\eps}3
	+ E^0_{T-\nu}(\omega) \hat Z^0_{T-\nu}(\omega)
	+ E^0_{T-\nu}(\omega) \big( \hat Z^{\eta,N}_{T-\nu}(\omega) - \hat Z^0_{T-\nu}(\omega) \big)
	\\&
	+ \big( E^0_{s(\omega)}(\omega) - E^0_{T-\nu}(\omega) \big) \hat Z^{\eta,N}_{T-\nu}(\omega)
	+ e^{\eta,N}_{s(\omega)}(\omega) \hat Z^{\eta,N}_{T-\nu}(\omega).
\end{align*}
Since $\omega \in M_{3}^{\nu,\eta,N}$, we have
$
	F^{\eta,N}_{s(\omega)}(\omega)
	=
	\varphi_{s(\omega)}(\omega) + f^{\eta,N}_{s(\omega)}(\omega)
	\ge
	\underline\varphi - \nu,
$
which is strictly positive if $\nu < \underline\varphi$.
Moreover,
$
	- \varphi \hat X^0 + E^0 \hat Z^0
	= 0
$
on $[0,T)$.
Since $\omega \in M_{0}^{\nu,\eta,N} \cap M_2^{\beta,\nu} \cap M_3^{\nu,\eta,N}$ and since $\hat X^0$, $E^0$ and $\hat Z^{\eta,N}$ are bounded,
this shows that $\ddt \hat X^{\eta,N}_{s(\omega)}(\omega) < 0$ if $\nu$ and $\beta$ are chosen small enough.

The convergence of $\hat Y^{\eta,N} - \hat Y^0$ on $[0,T-]$ follows from
\begin{align*}
	\hat Y^{\eta,N} - \hat Y^0
	= \gamma \big( \hat X^{\eta,N} - \hat X^0 \big)
	- \big( \hat Z^{\eta,N} - \hat Z^0 \big)
\end{align*}
and \eqref{eqthmxConv0SpecificthmPM2GeEtc}. The convergence on $[T-,T]$ follows from the fact that
\begin{align*}
	\sup_{s\in [T-\nu,T]} \Big|
		\big( \hat Z^{\eta,N}_s - \hat Z^0_s \big)
		- \big( \hat Z^{\eta,N}_{T-\nu} - \hat Z^0_{T-\nu} \big)
	\Big|
\end{align*}
can be made arbitrarily small by choosing $\nu$ small since $\ddt \hat Z^{\eta,N} = \rho \hat Y^{\eta,N}$, since $\ddt \hat Z^0 = \rho \hat Y^0$ and because both $\hat Y^{\eta,N}$ and $\hat Y^0$ are bounded.
\end{proof}

\subsubsection{Proof of Proposition \ref{thmxConv0prop}}
\label{SectXGraphConv}

The proof of the convergence of the optimal portfolio processes in Skorohod $\mathcal M_2$ sense follows from Theorem \ref{thmxConv0Specificthm} if strict liquidation is required. If strict liquidation is not required, the results of Section \ref{SectModelEta0} are required to establish the assertion. This is not a circular argument since the proofs of Section \ref{SectModelEta0} only use the results of Section \ref{SectModelEtaPositive} concerning the liquidating case $N=\infty$.

Since $\hat\theta=(0,\hat j^-, \hat V)$ is optimal in the model introduced in Section \ref{SectModelEta0} (cf. Theorem \ref{thmthetaoptimalJ0thm}), it is easy to show that the strategy $\hat\theta^q:=(\hat j^{+,q},\hat j^{-,q},\hat V^q)$ is admissible for every $q \in \mathbb R$ where
\begin{align*}
	\hat V^q_t
	=&
	\hat V_t + q t,
	\\
	\hat j^{+,q}_t - \hat j^{-,q}_t
	=&
	- \hat j^-_t - \mathbbm 1_{\{T\}}(t) q T.
\end{align*}

The following lemma shows that we can express the cost term corresponding to the transient price impact without It\^o integrals. The proof is an immediate consequence of It\^o's formula for semimartingales (Theorem II.32 in \cite{protter}) and the fact that $\gamma^2 [V]$ is equal to the continuous part of $[Y^\theta]$.
\begin{lemma} \label{thmintYdXeqlemma}
For all $\theta = (j^+,j^-,V) \in \mathcal A^0$ and all $t \in [0,T]$, it holds that
\begin{align*}
	& \int_{(0,t]} \frac 1 2 \big( Y^\theta_{s-} + Y^\theta_s \big) \, \mathrm dX^\theta_s
	=
	\frac 1 {2\gamma} \big(Y^\theta_t\big)^2 - \frac \gamma 2 (j^+_0 - j^-_0)^2
	- \frac \gamma 2 [V]_t
	+ \frac 1 \gamma \int^t_0 \rho_s \big(Y^\theta_s\big)^2 \, \ds.
\end{align*}
\end{lemma}

Using the previous lemma, it is not difficult to check that the mapping $q \mapsto J^0(\hat\theta^q)$ is differentiable and that
\begin{align}
\begin{split}
	0
	=
	\frac {\partial J^0(\hat\theta^q)} {\partial q} \bigg|_{q=0}
	=&
	\E\bigg[
		- \hat Y^0_T e^{-\int^T_0 \rho_u \, \du } \int^T_0 t \rho_t e^{\int_0^t \rho_u \, \du} \, \dt
		+ \int^T_0 2 t \rho_t \hat Y^0_t \, \dt
	\\& \quad
		- \int^T_0 2 \rho_t \hat Y^0_t e^{-\int^t_0 \rho_u \, \du } \int^t_0 s \rho_s e^{\int_0^s \rho_u \, \du} \, \ds \, \dt
		+ \int^T_0 t \lambda_t \hat X^0_t \, \dt
	\bigg].
\end{split}
\label{eqthmXTnFiniteTo0lemma0EqDdcPart2}
\end{align}

This allows us to prove that full liquidation is optimal if $\eta \to 0$ even if it is not formally required.

\begin{lemma} \label{thmXTnFiniteTo0lemma}
We have that $\lim_{\eta\to 0} \sup_{N\in \mathcal N} \E[\hat X^{\eta,N}_T] = 0$.
\end{lemma}
\begin{proof}
We assume to the contrary that $\limsup_{\eta \to 0} \sup_{N\in \mathcal N} \E[\hat X^{\eta,N}_T] > 0$ and prove that this contradicts the optimality of $J^{\eta,N}(\hat\xi^{\eta,N})$.
To this end, we consider the admissible trading strategies
\begin{align*}
	\hat\xi^{\eta,N,q}_t
	:=
	\hat\xi^{\eta,N}_t
	+ q,
\end{align*}
compute the derivative of the function $q \mapsto J^{\eta,N}(\hat\xi^{\eta,N,q})$ and show that the derivative at $q=0$ does not vanish for small $\eta$ if
\begin{align*}
	\limsup_{\eta \to 0} \sup_{N\in \mathcal N} \E\big[ \hat X^{\eta,N}_T \big] > 0.
\end{align*}
Obviously, $\hat\xi^{\eta,N,q} \in \mathcal A$ and
\begin{align*}
	\hat X^{\eta,N,q}_t
	:=
	X^{\hat\xi^{\eta,N,q}}_t
	=
	\hat X^{\eta,N}_t + q t.
\end{align*}
Straightforward computations lead to
\begin{align*}
	0
	=
	\frac {\partial J^{\eta,N}(\hat\xi^{\eta,N,q})} {\partial q} \bigg|_{q=0}
	=&
	\E\bigg[
		\eta \big( \hat X^{\eta,N}_T - x_0 \big)
		- \hat Y^{\eta,N}_T e^{ -\int^T_0 \rho_u \, \du } \int^T_0 t \rho_t e^{ \int_0^t \rho_u \, \du } \, \dt
		+ \int^T_0 2 t \rho_t \hat Y^{\eta,N}_t \, \dt
	\\&
		- \int^T_0 2 \rho_t \hat Y^{\eta,N}_t e^{ -\int^t_0 \rho_u \, \du } \int^t_0 s \rho_s e^{\int_0^s \rho_u \, \du } \, \ds \, \dt
		+ \int^T_0 t \lambda_t \hat X^{\eta,N}_t \, \dt
	\\&
		+ (N - \gamma) T \hat X^{\eta,N}_T
		+ \gamma \hat X^{\eta,N}_T e^{ -\int^T_0 \rho_u \, \du } \int^T_0 t \rho_t e^{ \int_0^t \rho_u \, \du } \, \dt
	\bigg].
\end{align*}
Subtracting \eqref{eqthmXTnFiniteTo0lemma0EqDdcPart2} yields
\begin{align*}
	0
	=&
	\E\bigg[
		- \eta x_0
		+ \hat X^{\eta,N}_T \big(
			\eta
			+ (N - \gamma) T
		\big)
		+ \int^T_0 2 t \rho_t \big( \hat Y^{\eta,N}_t - \hat Y^0_t \big) \, \dt
	\\&
		+ \big( \hat Z^{\eta,N}_T - \hat Z^0_T \big) e^{ -\int^T_0 \rho_u \, \du } \int^T_0 t \rho_t e^{ \int_0^t \rho_u \, \du } \, \dt
		+ \int^T_0 t \lambda_t \big( \hat X^{\eta,N}_t - \hat X^0_t \big) \, \dt
	\\&
		- \int^T_0 2 \rho_t \big( \hat Y^{\eta,N}_t - \hat Y^0_t \big) e^{ -\int^t_0 \rho_u \, \du } \int^t_0 s \rho_s e^{\int_0^s \rho_u \, \du } \, \ds \, \dt
	\bigg]
	\\\ge&
	- \eta x_0
	+ \big( \eta + (N - \gamma) T \big) \E\big[ \hat X^{\eta,N}_T \big]
	- T \overline\lambda \E\bigg[ \int^T_0 \big| \hat X^{\eta,N}_t - \hat X^0_t \big| \, \dt \bigg]
	\\&
	- 2 \big( T \overline\rho + \overline\rho^2 T^2 e^{ T \overline\rho } \big) \E\bigg[ \int^T_0 \big| \hat Y^{\eta,N}_t - \hat Y^0_t \big| \, \dt \bigg]
	- T^2 \overline\rho e^{ T \overline\rho } \E\Big[ \big| \hat Z^{\eta,N}_T - \hat Z^0_T \big| \Big].
\end{align*}

In view of Theorem \ref{thmXYEtaPosBdthm}, Theorem \ref{thmxConv0Specificthm}, Lemma \ref{thmXY0BdStarStarlemma} and using
\begin{align*}
	\ddt \big( \hat Z^{\eta,N} - \hat Z^0 \big) = \rho \big( \hat Y^{\eta,N} - \hat Y^0 \big),
\end{align*}
the sum of the three last expected values is small uniformly in $N$ if $\eta$ is small. Hence, if $\limsup_{\eta \to 0} \sup_{N \in \mathcal N} \E[\hat X^{\eta,N}_T] > 0$, then the sum on the right hand side of the above inequality is strictly positive when first choosing $\eta_0>0$ small enough and then choosing $(\tilde\eta,\tilde N)$ with $\tilde\eta \le \eta_0$ such that
$$
\E\big[ \hat X^{\tilde\eta,\tilde N}_T \big]
\ge
\frac 1 2 \limsup_{\eta \to 0} \sup_{N \in \mathcal N} \E\big[ \hat X^{\eta,N}_T \big].
$$.
\end{proof}

We are now ready to prove the convergence of $\hat X^{\eta,N}$ to $\hat X^0$ in the Skorohod $\mathcal M_2$ sense. To this end, we have to bound the distance of each point of any of the graphs to the other graph. In the inner interval $[\eps,T-\eps]$, it is enough to consider $\hat X^{\eta,N}-\hat X^0$, which we have bounded by Theorem \ref{thmxConv0Specificthm}.

\begin{proof}[Proof of Proposition \ref{thmxConv0prop}]
To prove that the probability of
\begin{align*}
	d_{\mathcal M_2}\big(\hat X^{\eta,N},\hat X^0\big) \le \eps
\end{align*}
is large for small $\eta > 0$, we need to prove that the distance of any point $(t,x)$ from either $G_{\hat X^{\eta,N}(\omega)}$ or $G_{\hat X^0(\omega)}$ to the respective other graph is small on a set of large probability. To this end, we fix a small enough $\nu \in (0,\eps)$. If $t\in [\nu,T-\nu]$, this follows directly from Theorem \ref{thmxConv0Specificthm}. For $t \in [0,\nu) \cup (T-\nu,T]$, we use Theorem \ref{thmxConv0Specificthm} along with the facts that (i) the completed graph of a discontinuous function contains the line segments joining the values of the function at the points of discontinuity; (ii) the increments of $\hat X^0$ are small in the sense of Lemma \ref{thmStarStarlemma} and (iii) $\lim_{\eta\to 0} \sup_{N\in\mathcal N} \E[\hat X^{\eta,N}_T] = 0$, due to Lemma \ref{thmXTnFiniteTo0lemma}. For instance, let us consider an $\omega \in \Omega$ with
\begin{align*}
	\sup_{u \in [0,T-\nu]} \big( \hat X^0_u(\omega) - \hat X^{\eta,N}_u(\omega) \big) \le \nu
\end{align*}
and assume that $(t,x) \in G_{\hat X^{\eta,N}(\omega)}$ with $t < \nu$ and
$x < \hat X^0_0(\omega) - \nu$. Since $x = \hat X^{\eta,N}_t(\omega)$, the mean value theorem yields an $s \in [0,t]$ s.t. $\hat X^0_s(\omega) = x + \nu$, which proves that $(s,x+\nu) \in G_{\hat X^0(\omega)}$ and
\begin{align*}
	d_\infty \big( (t,x),(s,x+\nu) \big)
	=
	(t-s) + \nu
	\le
	2 \nu.
\end{align*}
\end{proof}


\section{Proofs for Section \ref{SectModelEta0}} \label{SectProofsPart2}

\subsection{Proof of Lemma \ref{thmWPApproximationExistslemma} } \label{SectVApproximation}

In order to prove Lemma \ref{thmWPApproximationExistslemma} we first define, for all $\beta,\nu \in (0,\infty)$ and $x \in \R$,
\begin{align*}
	f^{\beta,\nu}(x)
	:= \begin{cases}
	-\beta/\nu, & x \le -\beta, \\
	x/\nu, & -\beta \le x \le \beta, \\
	\beta/\nu, & x \ge \beta.
	\end{cases}
\end{align*}

For any admissible strategy $(j^+,j^-,V) \in \mathcal A^0$, there exists a unique
pathwise differentiable, adapted stochastic process $\{\tilde V^{\beta,\nu}_t\}_{t \in [0,T]}$ such that $\tilde V^{\beta,\nu}_0 = 0$ and its time derivative satisfies
\begin{align*}
	\dot V^{\beta,\nu}_t(\omega)
	=
	f^{\beta,\nu} \big( V_t(\omega) - \tilde V^{\beta,\nu}_t(\omega) \big).
\end{align*}

\medskip

Now, \eqref{eqMaxVBetaNuLeMaxV} can easily be verified by the comparison principle: For all $\eps > 0$ and all $t \in [0,T]$, we have
\begin{align*}
	f^{\beta,\nu} \bigg( V_t - \Big( \max_{s \in [0,T]} \vert V_s\vert + \eps \Big) \bigg)
	<
	0
\end{align*}
and hence, for all $\eps > 0$ and $t \in [0,T]$,
\begin{align*}
	\tilde V^{\beta,\nu}_t
	\le
	\max_{s \in [0,T]} \vert V_s\vert + \eps.
\end{align*}

Analogously, we can prove that
	$\tilde V^{\beta,\nu}_t
	\ge
	- \max_{s \in [0,T]} \vert V_s\vert - \eps$.
Now, we are going to prove \eqref{eqIntDotVMinusVSmall}. For $\beta,\nu \in (0,\infty)$, let
\begin{align*}
	M
	:= \Big\{ \sup_{s,t\in[0,T], |s-t| \le \nu} \big| V_s - V_t \big| \le \beta \Big\}.
\end{align*}
In view of Lemma \ref{thmStarStarlemma}, it is enough to prove that
\begin{align*}
	\big| V_t(\omega) - \tilde V^{\beta,\nu}_t(\omega) \big|
	< 3 \beta
\end{align*}
for all $\omega\in M$ and $t\in [0,T]$. In order to see this, let us assume to the contrary that the statement is wrong. Then, by continuity, since $V_0 - \tilde V^{\beta,\nu}_0 = 0$, there exists some $\omega \in M$ and some $t_1 \in [0,T]$ such that
\begin{align*}
	\big| V_{t_1}(\omega) - \tilde V^{\beta,\nu}_{t_1}(\omega) \big|
	= 3 \beta.
\end{align*}
We choose the smallest such $t_1 \in [0,T]$.
Then, for each $t \in [0\vee(t_1-\nu),t_1]$,
\begin{align*}
	\big| V_t(\omega) - \tilde V^{\beta,\nu}_t(\omega) \big|
	\geq &
	\big| V_{t_1}(\omega) - \tilde V^{\beta,\nu}_{t_1}(\omega) \big|
	- \big| V_{t_1}(\omega) - V_t(\omega) \big|
	- \big| \tilde V^{\beta,\nu}_{t_1}(\omega) - \tilde V^{\beta,\nu}_t(\omega) \big|
	\\\ge&
	3 \beta - \beta - (t_1 - t) \frac \beta \nu
	\ge
	\beta.
\end{align*}
Since $V_0 - \tilde V^{\beta,\nu}_0 = 0$, we have $t_1 > \nu$. Since $V_\cdot(\omega) - \tilde V^{\beta,\nu}_\cdot(\omega)$ is continuous and has no roots in $[t_1-\nu,t_1]$, it does not change sign on this interval. We may hence w.l.o.g.~assume that
$V_{t_1}(\omega) - \tilde V^{\beta,\nu}_{t_1}(\omega) = 3 \beta$. This implies that
\begin{align*}
	V_{t}(\omega) - \tilde V^{\beta,\nu}_{t}(\omega)
	\ge \beta
\end{align*}
for all $t \in [t_1-\nu,t_1]$. By definition, for all those $t$, $\dot V^{\beta,\nu}_t(\omega) = \beta/\nu$. This, however, contradicts the minimality of $t_1$ as
\begin{align*}
	V_{t_1-\nu}(\omega) - \tilde V^{\beta,\nu}_{t_1-\nu}(\omega)
	=&
	V_{t_1-\nu}(\omega) - V_{t_1}(\omega)
	+ V_{t_1}(\omega) - \tilde V^{\beta,\nu}_{t_1}(\omega)
	+ \tilde V^{\beta,\nu}_{t_1}(\omega) - \tilde V^{\beta,\nu}_{t_1-\nu}(\omega)
	\\\ge&
	- \beta + 3 \beta + \int_{t_1-\nu}^{t_1} \underbrace{\dot V^{\beta,\nu}_s}_{=\beta/\nu} \, \ds
	= 3 \beta.
\end{align*}
This finishes the proof of Lemma \ref{thmWPApproximationExistslemma}.
\hfill $\Box$

\subsection{Proof of Lemma \ref{thmEintDeltaXsqconv0lemma}} \label{SectConvergence}

The proof of Lemma \ref{thmEintDeltaXsqconv0lemma} requires
the following result on the jump processes.

\begin{lemma} \label{thmEintDeltasqjconv0lemma}
Let $\{j_t\}_{t\in[0,T]}$ be a $[0,\infty)$-valued, non-decreasing and progressively measurable stochastic process with $\E[(j_T)^2] < \infty$. Then (putting $j_t = 0$ for $t < 0$)
\begin{align*}
	\lim_{\eps \to 0} \E\bigg[ \int^T_0 (j_t - j_{t-\eps})^2 \, \dt \bigg]
	= 0.
\end{align*}
\end{lemma}
\begin{proof}
First we want to bound $\int^T_0 (j_t - j_{t-\eps})^2 \, \dt$ through a term that depends on $\int^T_0 (j_t - j_{t-2\eps})^2 \, \dt$ and apply this bound inductively.
For every $\eps > 0$, we have
\begin{align*}
	\int^T_0 (j_t - j_{t-2\eps})^2 \, \dt
	= & \int^T_0 (j_t - j_{t-\eps} + j_{t-\eps} - j_{t-2\eps})^2 \, \dt
	\\ \ge &
	\int^T_0 (j_t - j_{t-\eps})^2 \, \dt
	+ \int_0^{T-\eps} ( j_t - j_{t-\eps} )^2 \, \dt
	\\ = & 2\int^T_0 (j_t - j_{t-\eps})^2 \, \dt
	- \int_{T-\eps}^T ( j_t - j_{t-\eps} )^2 \, \dt
\end{align*}
as well as
\begin{align*}
	\int_{T-\eps}^T ( j_t - j_{t-\eps} )^2 \, \dt
	\le \int_{T-\eps}^T ( j_T - j_{T-2\eps} )^2 \, \dt
	\le \eps ( j_T - j_{T-2\eps} )^2.
\end{align*}
As a result,
\begin{align*}
	\frac 1 \eps \int^T_0 (j_t - j_{t-\eps})^2 \, \dt
	\le \frac 1 {2\eps} \int^T_0 (j_t - j_{t-2\eps})^2 \, \dt + \frac 1 2 ( j_T - j_{T-2\eps} )^2.
\end{align*}
Hence, inductively, for all $\eps \in (0,T)$ and $k \in \{1,2,\ldots\}$, we have
\begin{align*}
	\frac 1 {\eps} \int^T_0 (j_t - j_{t-\eps})^2 \, \dt
	\le&
	\frac 1 {2 \eps} \int^T_0 (j_t - j_{t-2\eps})^2 \, \dt
	+ \frac 1 2 ( j_T - j_{T-2\eps} )^2
	\\\le&
	\frac 1 {4\eps} \int^T_0 (j_t - j_{t-4\eps})^2 \, \dt
	+ \frac 1 4 ( j_T - j_{T-4\eps} )^2
	+ \frac 1 2 ( j_T - j_{T-2\eps} )^2
	\\\le&
	\ldots
	\\\le &
	\frac 1 {2^k\eps} \int^T_0 (j_t - j_{t-2^k\eps})^2 \, \dt
	+ \sum_{l=1}^k \frac 1 {2^l} ( j_T )^2
	\\ \le &
	\frac 1 {2^k\eps} \int^T_0 (j_t - j_{t-2^k\eps})^2 \, \dt
	+ ( j_T )^2.
\end{align*}
Hence, for all $k \in \mathbb N$,
\begin{align*}
	\E\bigg[ \int^T_0 (j_t - j_{t-\eps})^2 \, \dt \bigg]
	\le
	\bigg(
		\frac T {2^k}
		+ \eps
	\bigg) \E\big[ ( j_T )^2 \big].
\end{align*}
This shows the desired result.
\end{proof}

We are now ready to prove the approximation of arbitrary portfolio processes by absolutely continuous ones.

\begin{proof}[Proof of Lemma \ref{thmEintDeltaXsqconv0lemma}]
By the triangle inequality, for all $\beta,\nu,\eps > 0$,
\begin{align}
\begin{split}
	&
	\sqrt{ \E\bigg[ \int^T_0 \big( X^{\theta,\beta,\nu,\eps}_t - X^\theta_t \big)^2 \, \dt \bigg] }
	\\\le&
	\sqrt{ \E\bigg[ \int^T_0 \bigg( \int^t_{(T-\eps)\wedge t} \dot V^{\beta,\nu}_s \, \ds \bigg)^2 \, \dt \bigg] }
	\\&
	+ \sum_{* \in \{+,-\}} \sqrt{ \E\bigg[ \int^T_0 \bigg( \int_0^{(T-\eps)\wedge t} \frac 1 \eps \big(j^*_s - j^*_{s-\eps}\big) \, \ds - j^*_t \bigg)^2 \, \dt \bigg] }
	\\&
	+ \sqrt{ \E\bigg[ \int^T_0 \bigg( \frac 1 \eps \big(t - (T-\eps) \big)^+ X^{\theta,\beta,\nu,\eps}_{T-\eps} \bigg)^2 \, \dt \bigg] }
	+ \sqrt{ \E\bigg[ \int^T_0 \bigg( \int^t_0 \dot V^{\beta,\nu}_s \, \ds - V_t \bigg)^2 \, \dt \bigg] }.
\end{split}
\label{eqthmEintDeltaXsqconv0lemmaSqrtEIntSqDeltaXLeEtc}
\end{align}
We analyse the four terms separately. This first term can be bounded by
\begin{align*}
	\sqrt{ \E\bigg[ \int^T_0 \bigg( \int^t_{(T-\eps)\wedge t} \dot V^{\beta,\nu}_s \, \ds \bigg)^2 \, \dt \bigg] }
	& \le
	\sqrt{ \E\bigg[ \int^T_{T-\eps} \bigg( \int^t_{T-\eps} \dot V^{\beta,\nu}_s \, \ds \bigg)^2 \, \dt \bigg] } \\
	& \le
	\frac {\eps^{\frac 3 2} \beta} \nu
	\to 0 \quad \mbox{for } \eps \to 0.
\end{align*}

Regarding the second term, for all $\eps < T/2$,
\begin{align*}
	\int_0^{(T-\eps)\wedge t} \frac 1 \eps \big(j^*_s - j^*_{s-\eps}\big) \, \ds - j^*_t
	= - \int_{((T-\eps)\wedge t) - \eps}^{(T-\eps)\wedge t} \frac 1 \eps (j^*_t - j^*_s) \, \ds.
\end{align*}

Using the monotonicity of the jump processes it follows from Lemma \ref{thmEintDeltasqjconv0lemma} that
\begin{align}
	&
	\sqrt{\E \bigg[ \int^T_0 \bigg( \int_0^{(T-\eps)\wedge t} \frac 1 \eps \big(j^*_s - j^*_{s-\eps}\big) \, \ds - j^*_t \bigg)^2 \, \dt \bigg] }
	\nonumber
	\\=&
	\sqrt{\E \bigg[ \int^T_0 \bigg( \int_{((T-\eps)\wedge t) - \eps}^{(T-\eps)\wedge t} \frac 1 \eps \big(j^*_t - j^*_s \big) \, \ds \bigg)^2 \, \dt \bigg] }
	\nonumber
	\\\le&
	\frac 1 \eps \sqrt{ \E\bigg[ \int^T_0 \bigg( \int_{((T-\eps)\wedge t) - \eps}^{(T-\eps)\wedge t} \big( j^*_t - j^*_s \big) \, \ds \bigg)^2 \, \dt \bigg] }
	\nonumber
	\\=&
	 \sqrt{
	\E\bigg[ \int_0^{T-\eps} \big( j^*_t - j^*_{t - \eps} \big)^2 \, \dt \bigg]
	+ \E\bigg[ \int_{T-\eps}^T \big( j^*_t - j^*_{T-2\eps} \big)^2 \, \dt \bigg]
	}
	\nonumber
	\\ \le &
	\sqrt{
		\E\bigg[ \int^T_0 \big( j^*_t - j^*_{t - \eps} \big)^2 \, \dt \bigg]
		+ \eps \E\Big[ \big( j^*_T \big)^2 \Big]
	}
	 \to 0 \quad \mbox{ for } \eps \to 0.
	\label{eqprfthmEintDeltaXsqconv0lemmaEintsqintDeltajconv0}
\end{align}

For the third term, we conclude from the It\^o isometry and the definition of $\xi^{\theta,\beta,\nu,\eps}$ that
\begin{align*}
	&
	\sqrt{\E\bigg[ \int^T_0 \bigg( \frac 1 \eps \big(t - (T-\eps) \big)^+ X^{\theta,\beta,\nu,\eps}_{T-\eps} \bigg)^2 \, \dt \bigg]}
	\\=&
	\sqrt{\frac \eps 3 \E\bigg[ \bigg( x_0 + \int_0^{T-\eps} \xi^{\theta,\beta,\nu,\eps}_s \, \ds \bigg)^2 \bigg]}
	\\\le&
	x_0 \sqrt{\frac \eps 3 }
	+ \sum_{* \in \{+,-\}} \sqrt{\frac \eps 3 \E\bigg[ \bigg( \int_0^{T-\eps} \frac 1 \eps \big( j^*_s - j^*_{s-\eps} \big) \, \ds \bigg)^2 \bigg]}
	+ \sqrt{\frac \eps 3 \E\bigg[ \bigg( \int_0^{T-\eps} \dot V^{\beta,\nu}_s \, \ds \bigg)^2 \bigg]}.
\end{align*}
Now we can bound
\begin{align*}
	\sqrt{\frac \eps 3 \E\bigg[ \bigg( \int_0^{T-\eps} \dot V^{\beta,\nu}_s \, \ds \bigg)^2 \bigg]}
	\le
	T \frac \beta \nu \sqrt{\frac \eps 3} \to 0 \quad \mbox{for } \eps \to 0.
\end{align*}
Moreover, in view of \eqref{eqprfthmEintDeltaXsqconv0lemmaEintsqintDeltajconv0},
\begin{align*}
	&
	 \sqrt{\frac \eps 3 \E\bigg[ \bigg( \int_0^{T-\eps} \frac 1 \eps \big( j^*_s - j^*_{s-\eps} \big) \, \ds \bigg)^2 \bigg]} \\
	= &
	\sqrt{\frac 1 3 \E\bigg[ \int_{T-\eps}^T \bigg( \int_0^{(T-\eps)\wedge t} \frac 1 \eps \big( j^*_s - j^*_{s-\eps} \big) \, \ds \bigg)^2 \, \dt\bigg]}
	\\ \le & \sqrt{\frac 1 3 \E\bigg[ \int^T_0 \bigg( \int_0^{(T-\eps)\wedge t} \frac 1 \eps \big( j^*_s - j^*_{s-\eps} \big) \, \ds - j^*_t \bigg)^2 \, \dt\bigg]}
	+ \sqrt{\frac 1 3 \E\bigg[ \int_{T-\eps}^T \big( j^*_t \big)^2 \, \dt\bigg]}
	\\\le&
	 \frac 1 {\sqrt 3} \sqrt{
		\E\bigg[ \int^T_0 \big( j^*_t - j^*_{t - \eps} \big)^2 \, \dt \bigg]
		+ \eps \E\Big[ \big( j^*_T \big)^2 \Big]
	}
	+ \sqrt{\frac \eps 3 \E\Big[ \big( j^*_T \big)^2 \Big]}
	 \to 0
\end{align*}
for $\eps \to 0$.

It remains to consider the fourth term in \eqref{eqthmEintDeltaXsqconv0lemmaSqrtEIntSqDeltaXLeEtc}. To this end, let
$$
M^{\beta,\nu}
:=
\bigg\{ \sup_{t \in [0,T]} \bigg\vert \int^t_0 \dot V^{\beta,\nu}_s \, \ds - V_t \bigg\vert \le 3\beta \bigg\}.
$$
Then, using \eqref{eqMaxVBetaNuLeMaxV} in the last step,
\begin{align*}
	&
	 \sqrt{ \E\bigg[ \int^T_0 \bigg( \int^t_0 \dot V^{\beta,\nu}_s \, \ds - V_t \bigg)^2 \, \dt \bigg] }
	\\\le&
	 \sqrt{ \E\bigg[ \int^T_0 \mathbbm 1_{M^{\beta,\nu}} \bigg( \int^t_0 \dot V^{\beta,\nu}_s \, \ds - V_t \bigg)^2 \, \dt \bigg] }
	+ \sqrt{ \E\bigg[ \mathbbm 1_{\Omega\backslash M^{\beta,\nu}} \int^T_0 ( V_t )^2 \, \dt \bigg] }
	\\&
	 + \sqrt{ \E\bigg[ \mathbbm 1_{\Omega\backslash M^{\beta,\nu}} \int^T_0 \bigg( \int^t_0 \dot V^{\beta,\nu}_s \, \ds \bigg)^2 \, \dt \bigg] } \\
	\le &
	3 \beta \sqrt T +
	\sqrt{ T \E\Big[ \mathbbm 1_{\Omega\backslash M^{\beta,\nu}} \max_{t \in [0,T]} \vert V_t \vert^2 \Big] }
	+ \sqrt{ T \E\bigg[ \mathbbm 1_{\Omega\backslash M^{\beta,\nu}} \max_{t \in [0,T]} \bigg( \int^t_0 \dot V^{\beta,\nu}_s \, \ds \bigg)^2 \bigg] } \\
	\le & 3 \beta \sqrt T
	+ 2 \sqrt{ T \E\Big[ \mathbbm 1_{\Omega\backslash M^{\beta,\nu}} \max_{t \in [0,T]} \vert V_t \vert^2 \Big] }.
\end{align*}
By \eqref{eqVMaxLp},
$\{ \max_{t \in [0,T]} \vert V_t \vert^2 \}$ is uniformly integrable.
Hence we can first choose $\beta>0$ small enough, then, according to Lemma \ref{thmWPApproximationExistslemma}, choose $\nu>0$ small enough such that $\IP[\Omega \backslash M^{\beta,\nu}]$ is sufficiently small and finally choose $\eps>0$ small enough in order to obtain the desired result.
\end{proof}

\subsection{Proof of Lemma \ref{thmJetaconvJ0lemma} } \label{SectCostApproximation}

We start with a technical lemma.

\begin{lemma} \label{thmIntDiffSqUSqVBdlemma}
Let $(M,\mathcal A,\mu)$ be a measure space and let $u,v \in \mathcal L^2(M)$. Then
\begin{align*}
	\int_M | u^2 - v^2 | \, \mathrm d\mu
	\le
	\int_M (u-v)^2 \, \mathrm d\mu
	+ 2 \sqrt{ \int_M (u-v)^2 \, \mathrm d\mu \int_M v^2 \, \mathrm d\mu }.
\end{align*}
\end{lemma}
\begin{proof}
Due to the Hölder inequality and the triangle inequality,
\begin{align*}
	\int_M | u^2 - v^2 | \, \mathrm d\mu
	=&
	\sqrt{ \int_M (u-v)^2 \, \mathrm d\mu } \sqrt{ \int_M (u+v)^2 \, \mathrm d\mu }
	\\\le&
	\sqrt{ \int_M (u-v)^2 \, \mathrm d\mu }
	\bigg(
		\sqrt{ \int_M (u-v)^2 \, \mathrm d\mu }
		+ \sqrt{ \int_M (2v)^2 \, \mathrm d\mu }
	\bigg).
\end{align*}
\end{proof}

The following technical lemma provides useful estimates for the impact process.
\begin{lemma} \label{thmYpropertieslemma}
Let $X \in \mathcal L_{\mathcal P}^2(\Omega\times[0,T];\R)$. Then the transient price impact process $Y^X$ given by \eqref{eqYDef} satisfies $Y^X \in \mathcal L_{\mathcal P}^2(\Omega\times[0,T];\R)$ and
\begin{align}
	\sqrt{\E\Big[ \big(Y^X_T\big)^2 \Big]}
	& \le
	\gamma \bigg(
		X_0
		+ \sqrt{\E\big[ (X_T)^2 \big]}
		+ \overline\rho \exp(T \overline\rho) \sqrt{T \E\bigg[ \int^T_0 (X_s)^2 \, \ds \bigg]}
	\bigg),
	\label{eqEYTsqbd}
\end{align}
\begin{align}
	\sqrt{\E\bigg[ \int^T_0 \big( Y^X_t \big)^2 \, \dt \bigg]}
	& \le
	\gamma X_0 \sqrt T
	+ \gamma \big( 1 + T \overline\rho \exp( T \overline\rho ) \big) \sqrt{ \E\bigg[ \int^T_0 (X_t)^2 \, \dt \bigg] }.
	\label{eqEintYsqbd}
\end{align}
If $X_0 = 0$, then additionally, we have
for all $s,t \in [0,T]$ with $s < t$,
\begin{align}
	\sqrt{ \E\bigg[ \int_s^t \big( Y^X_u \big)^2 \, \du \bigg] }
	\le
	& \gamma \overline\rho \exp\big( T \overline\rho \big) \sqrt{ (t-s) T \E\bigg[ \int_0^s \big( X_u \big)^2 \, \du \bigg] }
	\nonumber
	\\&
	 + \gamma \Big(
		1
		+ \overline\rho \exp\big( T \overline\rho \big) \sqrt{ (t-s) T }
	\Big) \sqrt{ \E\bigg[ \int_s^t \big( X_u \big)^2 \, \du \bigg] }.
	\label{eqEintpartYsqLeEtc}
\end{align}
\end{lemma}
\begin{proof}
Inequality \eqref{eqEYTsqbd} follows from the explicit formula
\begin{align}
	Y^X_t
	=
	\gamma X_t
	- \gamma \exp\bigg( -\int^t_0 \rho_u \, \du \bigg) \bigg( X_0 + \int^t_0 \rho_s X_s \exp\bigg( \int_0^s \rho_u \, \du \bigg) \, \ds \bigg)
	\label{eqYDefExplicit}
\end{align}
and the triangle inequality. Moreover,
\begin{align}
\begin{split}
	&
	\sqrt{\E\bigg[ \int^T_0 \big( Y^X_t \big)^2 \, \dt \bigg]}
	\le \gamma \sqrt{ \E\bigg[ \int^T_0 (X_t)^2 \, \dt \bigg] }
	\\& \quad\quad
	+ \gamma \sqrt{ \E\bigg[ \int^T_0 \exp\bigg( -2 \int^t_0 \rho_u \, \du \bigg) \bigg( X_0 + \int^t_0 \rho_s X_s \exp\bigg( \int_0^s \rho_u \, \du \bigg) \, \ds \bigg)^2 \, \dt \bigg] }.
\end{split}
\label{eqthmYpropertieslemmaEIntSqYLeEtc}
\end{align}
Substituting the inequality
\begin{align*}
	&\sqrt{ \E\bigg[ \int^T_0 \exp\bigg( -2 \int^t_0 \rho_u \, \du \bigg) \bigg( X_0 + \int^t_0 \rho_s X_s \exp\bigg( \int_0^s \rho_u \, \du \bigg) \, \ds \bigg)^2 \, \dt \bigg] }
	\\\le&
	\sqrt{ \E\big[T (X_0)^2\big] }
	+ \overline\rho \exp( T \overline\rho ) \sqrt{ T \E\bigg[ \bigg( \int^T_0 |X_s| \, \ds \bigg)^2 \bigg] }
	\\\le&
	X_0 \sqrt T
	+ T \overline\rho \exp( T \overline\rho ) \sqrt{ \E\bigg[ \int^T_0 (X_s)^2 \, \ds \bigg] }
\end{align*}
into \eqref{eqthmYpropertieslemmaEIntSqYLeEtc} yields \eqref{eqEintYsqbd}. To prove \eqref{eqEintpartYsqLeEtc}, let $X_0 = 0$. Then, due to \eqref{eqYDefExplicit},
\begin{align*}
	\big| Y^X_u - \gamma X_u \big|
	\le
	\gamma \overline\rho \exp\big( T \overline\rho \big) \int_0^u | X_r | \, \dr
\end{align*}
for all $u \in [0,T]$ and so
\begin{align*}
	\E\bigg[ \int_s^t \big(Y^X_u - \gamma X_u \big)^2 \, \du \bigg]
	\le (t-s) T \gamma^2 \overline\rho^2 \exp\big( 2 T \overline\rho \big) \E\bigg[ \int^t_0 (X_u)^2 \, \du \bigg].
\end{align*}
Using the subadditivity of the square root, we now obtain \eqref{eqEintpartYsqLeEtc} from
\begin{align*}
	& \sqrt{ \E\bigg[ \int_s^t \big( Y^X_u \big)^2 \, \du \bigg] }
	\le \sqrt{ \E\bigg[ \int_s^t \big( Y^X_u - \gamma X_u \big)^2 \, \du \bigg] }
	+ \sqrt{ \E\bigg[ \int_s^t (\gamma X_u)^2 \, \du \bigg] }
	\\\le&
	\gamma \overline\rho \exp\big( T \overline\rho \big) \sqrt{ (t-s) T \E\bigg[ \int_0^s (X_u)^2 \, \du \bigg] }
	\\&
	+ \gamma \overline\rho \exp\big( T \overline\rho \big) \sqrt{ (t-s) T \E\bigg[ \int_s^t (X_u)^2 \, \du \bigg] }
	+ \gamma \sqrt{ \E\bigg[ \int_s^t (X_u)^2 \, \du \bigg] }.
\end{align*}
\end{proof}

We are now ready to prove our approximation result for the cost functional.
\begin{proof}[Proof of Lemma \ref{thmJetaconvJ0lemma}]
For $\theta = (j^+,j^-,V) \in \mathcal A^0$ and $\xi \in \mathcal A$,
\begin{align*}
	\big\vert J^0(\theta) - J^0\big(0,0,V^\xi \big) \big\vert
	\le&
	 \bigg| \E \bigg[
		\frac \gamma 2 \big( j^+_0 - j^-_0 \big)^2
		+ \int_{(0,T]} \frac 1 2 \big( Y^\theta_{t-} + Y^\theta_t \big) \, \mathrm dX^\theta_t
		- \int^T_0 Y^{\xi}_t \, \mathrm dX^{\xi}_t
		+ \frac \gamma 2 [V]_T
	\bigg] \bigg| \\
	& + \frac {\overline\lambda} 2 \E\bigg[ \int^T_0 \Big| \big(X^\theta_t\big)^2 - \big(X^{\xi}_t\big)^2 \Big| \, \dt \bigg].
\end{align*}
Due to Lemma \ref{thmIntDiffSqUSqVBdlemma} the last term can be estimated as
\begin{align*}
	\E\bigg[ \int^T_0 \Big| \big(X^\theta_t\big)^2 - \big(X^{\xi}_t\big)^2 \Big| \, \dt \bigg]
	\le&
	\E\bigg[ \int^T_0 \big( X^\theta_t - X^{\xi}_t \big)^2 \, \dt \bigg]
	\\& \quad
	+ 2 \sqrt{
		\E\bigg[ \int^T_0 \big( X^\theta_t - X^{\xi}_t \big)^2 \, \dt \bigg]
		\E\bigg[ \int^T_0 \big( X^\theta_t \big)^2 \, \dt \bigg]
	}.
\end{align*}
Moreover, using first Lemma \ref{thmintYdXeqlemma}, and then Lemma \ref{thmYpropertieslemma} and Lemma \ref{thmIntDiffSqUSqVBdlemma}, we obtain
\begin{align*}
	&
	 \bigg| \E \bigg[
		\frac \gamma 2 \big( j^+_0 - j^-_0 \big)^2
		+ \int_{(0,T]} \frac 1 2 \big( Y^\theta_{t-} + Y^\theta_t \big) \, \mathrm dX^\theta_t
		- \int^T_0 Y^{\xi}_t \, \mathrm dX^{\xi}_t
		+ \frac \gamma 2 [V]_T
	\bigg] \bigg|
	\\ = & \bigg| \E \bigg[
		\frac \gamma 2 \big( j^+_0 - j^-_0 \big)^2
		+ \frac 1 {2\gamma} \big(Y^\theta_T\big)^2 - \frac \gamma 2 \big(j^+_0 - j^-_0\big)^2
		- \frac \gamma 2 [V]_T
		+ \frac 1 \gamma \int^T_0 \rho_t \big(Y^\theta_t\big)^2 \, \dt
	\\ &
		- \frac 1 {2\gamma} \big(Y^{\xi}_T\big)^2
		- \frac 1 \gamma \int^T_0 \rho_t \big(Y^{\xi}_t\big)^2 \, \dt
		+ \frac \gamma 2 [V]_T
	\bigg] \bigg|
	\\\le&
	 \frac 1 {2\gamma} \E\bigg[ \Big| \big(Y^\theta_T\big)^2 - \big(Y^{\xi}_T\big)^2 \Big| \bigg]
	+ \frac {\overline\rho} \gamma \E\bigg[ \int^T_0 \Big| \big(Y^\theta_t\big)^2 - \big(Y^{\xi}_t\big)^2 \Big| \, \dt \bigg]
	\\\le&
	 \bigg(
		\frac 1 2 T \gamma \overline\rho^2 e^{2 T \overline\rho}
		+ \overline\rho \gamma \big(1 + T \overline\rho e^{T \overline\rho} \big)^2
	\bigg) \E\bigg[ \int^T_0 \big(X^\theta_t - X^{\xi}_t\big)^2 \, \dt \bigg]
	\\&
	 + \sqrt T \gamma \overline\rho e^{T \overline\rho} \sqrt{ \E\bigg[ \int^T_0 \big(X^\theta_t - X^{\xi}_t\big)^2 \, \dt \bigg] } \bigg( x_0 + \overline\rho e^{T \overline\rho} \sqrt{T \E\bigg[ \int^T_0 (X^\theta_t)^2 \, \dt \bigg]} \bigg)
	\\&
	 + 2 \overline\rho \big(1 + T \overline\rho e^{T \overline\rho} \big) \sqrt{ \E\bigg[ \int^T_0 \big(X^\theta_t - X^{\xi}_t\big)^2 \, \dt \bigg] } \\
	& \quad \bigg(
		\gamma x_0 \sqrt T
		+ \gamma \big(1 + T \overline\rho e^{T \overline\rho} \big) \sqrt{ \E\bigg[ \int^T_0 (X^\theta_t)^2 \, \dt \bigg] }
	\bigg).
\end{align*}
\end{proof}

\subsection{Proof of Lemma \ref{thmetaEintxisqconv0lemma} } \label{SectInstantaneousCostConvergence}

We assume the contrary, i.e.\ that there exists a constant $c > 0$ such that, for all $H > 0$, there exists some $\eta \in (0,H)$ such that
\begin{align}
	\frac \eta 2 \E\bigg[ \int_0^T \big( \hat\xi^{\eta,\infty}_t \big)^2 \, \dt \bigg]
	>
	c.
	\label{eqthmlimetaEintxisqeq0lemmaEtaEIntEtcGeC}
\end{align}
The optimality of $\hat\xi^{\eta,\infty}$ and \eqref{eqJEtaEqJ0PlusEIntXiSq} imply that, for all $\eta,\nu,\beta,\eps > 0$,
\begin{align}
\begin{split}
	0
	\ge&
	J^{\eta,\infty}\big(\hat\xi^{\eta,\infty}\big) - J^{\eta,\infty}\big(\xi^{\hat\theta,\beta,\nu,\eps}\big)
	\\=&
	\frac \eta 2 \E\bigg[ \int_0^T \big( \hat\xi^{\eta,\infty}_t \big)^2 \, \dt \bigg]
	- \frac \eta 2 \E\bigg[ \int_0^T \big( \xi^{\hat\theta,\beta,\nu,\eps}_t \big)^2 \, \dt \bigg]
	\\&
	+ J^0 \big( 0,0,V^{\hat\xi^{\eta,\infty}} \big)
	- J^0 \big( \hat\theta \big)
	+ J^0 \big( \hat\theta \big)
	- J^0 \big( 0, 0, V^{\xi^{\hat\theta,\beta,\nu,\eps}} \big).
\end{split}
\label{eqthmetaEintxisqconv0lemma0GeEtc}
\end{align}

We now prove that \eqref{eqthmlimetaEintxisqeq0lemmaEtaEIntEtcGeC} contradicts \eqref{eqthmetaEintxisqconv0lemma0GeEtc}. By Theorem \ref{thmJetaconvJ0lemma} and since $|\hat X^0_t| \le x_0$, we obtain (for convenience, let $p(x) := x + \sqrt x$)
\begin{align*}
	\Big|
		J^0 \big( 0,0,V^{\hat\xi^{\eta,\infty}} \big)
		- J^0 \big( \hat\theta \big)
	\Big|
	\le
	D \big( T (x_0)^2 \big) \cdot p \bigg( \E\bigg[ \int^T_0 \big(
		\hat X^{\eta,\infty}_t
		- \hat X^0_t
	\big)^2 \, \dt \bigg] \bigg)
\end{align*}
and
\begin{align*}
	\Big|
		J^0 \big( \hat\theta \big)
		- J^0 \big( 0, 0, V^{\xi^{\hat\theta,\beta,\nu,\eps}} \big)
	\Big|
	\le
	D \big( T (x_0)^2 \big) \cdot p \bigg( \E\bigg[ \int^T_0 \big(
		X^{\hat\theta}_t
		- X^{\hat\theta,\beta,\nu,\eps}_t
	\big)^2 \, \dt \bigg] \bigg).
\end{align*}
Plugging the results into \eqref{eqthmetaEintxisqconv0lemma0GeEtc} yields that, for all $\eta > 0$ that satisfy \eqref{eqthmlimetaEintxisqeq0lemmaEtaEIntEtcGeC}, it holds
\begin{align}
\begin{split}
	0
	>&
	c
	- \frac \eta 2 \E\bigg[ \int_0^T \big( \xi^{\hat\theta,\beta,\nu,\eps}_t \big)^2 \, \dt \bigg]
	\\&
	- D \big( T (x_0)^2 \big) \bigg(
		p \bigg( \E\bigg[ \int^T_0 \big(
			\hat X^{\eta,\infty}_t
			- \hat X^0_t
		\big)^2 \, \dt \bigg] \bigg)
	\\& \quad
		+ p \bigg( \E\bigg[ \int^T_0 \big(
			X^{\hat\theta}_t
			- X^{\hat\theta,\beta,\nu,\eps}_t
		\big)^2 \, \dt \bigg] \bigg)
	\bigg).
\end{split}
\label{eqthmetaEintxisqconv0lemma0GeEtc2}
\end{align}
Due to Lemma \ref{thmEintDeltaXsqconv0lemma}, we can first choose $\beta, \nu, \eps > 0$ sufficiently small such that
\begin{align*}
	D \big( T (x_0)^2 \big)
	\cdot p \bigg( \E\bigg[ \int^T_0 \big(
		X^{\hat\theta}_t
		- X^{\hat\theta,\beta,\nu,\eps}_t
	\big)^2 \, \dt \bigg] \bigg)
	<
	\frac c 2.
\end{align*}
Since $|\hat X^{\eta,\infty}_t| \le x_0$ (cf.\ Theorem \ref{thmXYEtaPosBdthm}) and in view of Theorem \ref{thmxConv0Specificthm}, we can then choose $\eta > 0$ sufficiently small satisfying \eqref{eqthmlimetaEintxisqeq0lemmaEtaEIntEtcGeC} such that the right hand side of \eqref{eqthmetaEintxisqconv0lemma0GeEtc2} is larger than zero, which is a contradiction. This finishes the proof of Lemma \ref{thmetaEintxisqconv0lemma}.


\begin{appendix}

\section{Regularity properties of It\^o processes and BSDEs} \label{SectAppendixRegularityItoBsde}

In this appendix, we introduce some regularity properties of stochastic processes, which we use to prove various convergence results for stochastic processes.

We consider a continuous, adapted, $\R^d$-valued stochastic process $Y = \{Y_t\}_{t\in I}$ on some interval $I$ and introduce the following continuity conditions.
\begin{condition*}[\textbf{C.1}]
\hypertarget{C1}{For} all $\eps \in (0,\infty)$, there exists $\delta \in (0,\infty)$ such that, for all $s \in I$, all $\mathcal F_s$-measurable and integrable $V \colon \Omega \to \R$ and all stopping times $\tau \colon \Omega \to [s,s+\delta]\cap I$,
$$
\Big\Vert \E \big[ V ( Y_\tau - Y_s ) \big] \Big\Vert_\infty
\le \eps \E\big[ |V| \big].
$$
\end{condition*}

\begin{condition*}[\textbf{C.2}]
\hypertarget{C2}{For} all $\eps \in (0,\infty)$, there exists $\delta \in (0,\infty)$ such that, for all $s \in I$,
$$
\E \Big[ \sup_{t\in [s,s + \delta]\cap I} \Vert Y_t - Y_s \Vert_\infty \Big| \mathcal F_s \Big]
\le \eps.
$$
\end{condition*}

\begin{definition}
A family of stochastic processes is said to \textit{uniformly} satisfy Condition \hyperlink{C1}{C.1} or \hyperlink{C2}{C.2} on $I$ if all processes satisfy the respective property and $\delta$ can be chosen uniformly for all processes.
\end{definition}

In what follows, we list some auxiliary results.

\begin{lemma} \label{thmPropertyStarComponentslemma}
An $\R^d$-valued stochastic process $Y = \{(Y^1_t,\ldots,Y^d_t)\}_{t\in I}$ satisfies Condition \hyperlink{C1}{C.1} (Condition \hyperlink{C2}{C.2}) if and only if all components $Y^i$ satisfy Condition \hyperlink{C1}{C.1} (Condition \hyperlink{C2}{C.2}).
\end{lemma}

\begin{lemma} \label{thmC2ImpliesC1lemma}
Condition \hyperlink{C2}{C.2} implies Condition \hyperlink{C1}{C.1}.
\end{lemma}
\begin{proof}
Due to Lemma \ref{thmPropertyStarComponentslemma}, it is enough to consider a real-valued $Y$. The following calculation shows that $Y$ satisfies Condition \hyperlink{C1}{C.1} under the assumption that it satisfies Condition \hyperlink{C2}{C.2}:
\begin{align*}
	&
	\Big| \E \big[ V ( Y_\tau - Y_s ) \big] \Big|
	= \Big| \E \big[ V \cdot \E[ Y_\tau - Y_s | \mathcal F_s ] \big] \Big|
	\le \E\Big[ |V| \cdot \E \big[ | Y_\tau - Y_s |_\infty \big| \mathcal F_s \big] \Big]
	\le \E\big[ |V| \big] \eps.
\end{align*}
\end{proof}

\begin{remark}
If we want to bound a term of the form $\Vert \E[\mathbbm 1_N (Y_\tau - Y_s)] \Vert_\infty$ for some $N \subseteq M \in \mathcal F_s$ with $N \not\in \mathcal F_s$, then Condition \hyperlink{C1}{C.1} is not enough; in this case, we need the stronger Condition \hyperlink{C2}{C.2}.
\end{remark}

\begin{lemma} \label{thmSumCondC1C2lemma}
Let $X = \{X_t\}_{t\in I}$ and $Y = \{Y_t\}_{t\in I}$ be continuous, adapted, real-valued processes, which satisfy Condition \hyperlink{C1}{C.1} (Condition \hyperlink{C2}{C.2}). Then $X + Y$ satisfies Condition \hyperlink{C1}{C.1} (Condition \hyperlink{C2}{C.2}).
\end{lemma}

\begin{theorem} \label{thmXYSatisfiesStarPropertythm}
Let $X = \{X_t\}_{t\in I}$ and $Y = \{Y_t\}_{t\in I}$ be continuous, adapted, essentially bounded, real-valued processes.
\begin{itemize}
\item
If $X$ satisfies Condition \hyperlink{C2}{C.2} and $Y$ satisfies Condition \hyperlink{C1}{C.1}, then $X\cdot Y$ satisfies Condition \hyperlink{C1}{C.1}.
\item
If $X$ and $Y$ both satisfy Condition \hyperlink{C2}{C.2}, then $X\cdot Y$ also satisfies Condition \hyperlink{C2}{C.2}.
\end{itemize}
\end{theorem}

\begin{proof}
The second part of the statement can be proven straightforward using the decomposition
$$
	|X_t Y_t - X_s Y_s|
	\le
	|X_t| |Y_t - Y_s|
	+ |Y_s| |X_t - X_s|.
$$
The first part of the statement follows from
\begin{align*}
	&
	\Big| \E \big[ V ( X_\tau Y_\tau - X_s Y_s ) \big] \Big|
	\le
	\Big| \E \big[ V Y_\tau (X_\tau - X_s) \big] \Big|
	+ \Big| \E \big[ V X_s (Y_\tau - Y_s) \big] \Big|.
\end{align*}
\end{proof}

Next, we prove some properties of the concave envelope of the modulus of continuity of a uniformly continuous function.
We use these results to show that a uniformly continuous function of an It\^o process with bounded coefficients satisfies Condition \hyperlink{C2}{C.2}.
\begin{lemma} \label{thmexistsmodcontconcavelemma}
Let $(X,\Vert\cdot\Vert_X)$ be a normed space with a non-empty and convex subset $D$, let $(Y,d_Y)$ be a metric space, let $f\colon D \to Y$ and, for every $t \in [0,\infty)$ let
\begin{align*}
	\tilde\omega_f(t)
	:=&
	\sup_{x_1,x_2 \in D, \Vert x_1-x_2\Vert_X \le t} d_Y\big( f(x_1),f(x_2) \big),
	\\
	\omega_f(t)
	:=&
	\inf \{ \psi(t) : \psi \colon [0,\infty) \to [0,\infty) \, \text{concave and} \, \psi \ge \tilde\omega_f \}.
\end{align*}
Then $\tilde\omega_f \colon [0,\infty) \to [0,\infty]$, $\tilde\omega_f$ is non-decreasing, $\tilde\omega_f(0) = 0$, and for all $s,t \in [0,\infty)$ with $s > 0$,
\begin{align}
	\tilde\omega_f(t)
	\le \frac{t+s} s \tilde\omega_f(s).
	\label{eqthmexistsmodcontlemmatildeomegabound}
\end{align}

If $\tilde\omega_f$ is finite, then $\omega_f$ takes values in $[0,\infty)$, is concave, non-decreasing and, for all $x_1,x_2 \in D$,
\begin{equation}
	d_Y\big( f(x_1),f(x_2) \big) \le \omega_f\big( \Vert x_1 - x_2\Vert_X \big).
	\label{eqthmexistsmodcontconcavelemmaDeltaFLeOmegaDeltaX}
\end{equation}
If $f$ is uniformly continuous, then $\tilde\omega_f$ is finite, $\tilde\omega_f$ and $\omega_f$ are continuous in $0$ and $\omega_f(0) = 0$.
\end{lemma}
\begin{proof}
By definition, $\tilde{\omega}_f$ is non-decreasing and $\tilde{\omega}_f(0)=0$.
To prove \eqref{eqthmexistsmodcontlemmatildeomegabound}, let $s,t\in [0,\infty)$ with $s > 0$ and $x_1,x_2 \in D$ with $\Vert x_1-x_2\Vert_X \le t$. Let $N := \lceil t/s \rceil$ and let
\begin{align*}
	x^k
	:= \frac k {N} x_1 + \frac {N-k} {N} x_2
	\quad (k=0,\ldots,N).
\end{align*}
Then $x^k \in D$ because $D$ is convex and
$
	\Vert x^k - x^{k+1} \Vert_X
	= \frac 1 N \Vert x_1 - x_2\Vert_X
	\le s.
$
Hence
\begin{align*}
	d_Y\big( f(x_1),f(x_2) \big)
	\le
	\sum_{k=0}^{N-1} d_Y\big( f(x^k), f(x^{k+1}) \big)
	\le
	N \max_{k} \tilde\omega_f\big(\Vert x^k - x^{k+1} \Vert_X\big)
	\le
	(t + s) \frac {\tilde\omega_f(s)} s.
\end{align*}

Having established \eqref{eqthmexistsmodcontlemmatildeomegabound}, we know that, if $\tilde\omega_f$ is finite, then it is bounded above by an affine (hence concave) function. In particular, the concave envelope is well-defined and satisfies \eqref{eqthmexistsmodcontconcavelemmaDeltaFLeOmegaDeltaX}.

Concavity and monotonicity of $\omega_f$ are clear.

Uniform continuity of $f$ implies that $\tilde\omega_f$ is finite and continuous in $0$. In order to see $\omega_f(0) = 0$ and continuity in $0$, let $\eps > 0$ be arbitrary and let $\delta > 0$ be such that $\tilde\omega_f(\delta) \le \eps/2$. Due to \eqref{eqthmexistsmodcontlemmatildeomegabound}, $\tilde\omega_f$ and hence $\omega_f$ is bounded above by the affine and hence concave function $t \mapsto ((t+\delta)/\delta) \cdot (\eps/2)$. In particular, for all $t \le \delta$, we obtain $\omega_f(t) \le \eps$.
\end{proof}

\begin{lemma} \label{thmftWtpropstarlemma}
Let $X = \{X_t\}_{t\in I}$ be a continuous, adapted, $\R^d$-valued process which satisfies Condition \hyperlink{C2}{C.2} and let $f \colon [0,T] \times \R^d \to \R$ be uniformly continuous. Then $\{f(t,X_t)\}_{t\in[0,T]}$ also satisfies Condition \hyperlink{C2}{C.2}.
\end{lemma}
\begin{proof}
There exists a function $\omega_f$ with the properties listed by Lemma \ref{thmexistsmodcontconcavelemma} (with the norm $\Vert\cdot\Vert_\infty$ on $[0,T]\times\R^d$).
Now let $\delta > 0$ be arbitrary. Using \eqref{eqthmexistsmodcontconcavelemmaDeltaFLeOmegaDeltaX} and the concavity and monotonicity of $\omega_f$, we obtain that
\begin{align*}
	&
	\E\Big[ \max_{s \le t\le T\wedge(s + \delta)} \big| f(t,X_t) - f(s,X_s) \big| \Big| \mathcal F_s \Big] \\
	\le &
	\E\Big[ \max_{s \le t\le T\wedge(s + \delta)} \omega_f\big( |t-s| \vee \Vert X_t - X_s \Vert_\infty \big) \Big| \mathcal F_s \Big]
	\\\le&
	\E\bigg[ \omega_f\Big( \delta \vee \max_{s \le t\le T\wedge(s + \delta)} \Vert X_t - X_s \Vert_\infty \Big) \Big| \mathcal F_s \bigg]
	\\ \le &
	\omega_f\bigg( \delta + \E\Big[ \max_{s \le t\le T\wedge(s + \delta)} \Vert X_t - X_s \Vert_\infty \Big| \mathcal F_s \Big] \bigg).
\end{align*}
Since $\omega_f$ is continuous in $0$ and since $X$ satisfies Condition \hyperlink{C2}{C.2}, we can choose $\delta>0$ small enough such that this term is not greater than $\eps$.
\end{proof}

\begin{lemma} \label{thmEcondmaxincXbdlemma}
Let $\{\tilde\sigma_t\}_{t\in[0,T]} \in \mathcal L^\infty_{\mathcal P}(\Omega\times[0,T];\R^m)$. Then $\{\int^t_0 \tilde\sigma_u \, \mathrm dW_u\}_{t\in[0,T]}$ satisfies Condition \hyperlink{C2}{C.2}.
\end{lemma}
\begin{proof}
Let $\overline\sigma > 0$ be a component-wise bound on $\tilde\sigma$. By the Jensen inequality,
\begin{align*}
	\E\bigg[ \max_{s \le t\le T\wedge(s + \delta)} \bigg| \int^t_s \tilde\sigma_u \, \mathrm dW_u \bigg| \bigg| \mathcal F_s \bigg]
	\le
	\sqrt{ \E\bigg[ \max_{s \le t\le T\wedge(s + \delta)} \bigg| \int^t_s \tilde\sigma_u \, \mathrm dW_u \bigg|^2 \bigg| \mathcal F_s \bigg] }.
\end{align*}
In order to prove that the conditional expectation on the right hand side is bounded, we apply the classical Doob's maximal inequality concerning the conditional measures w.r.t. all sets $A \in \mathcal F_s$ with positive probability and obtain
\begin{align*}
	\sqrt{ \E\bigg[ \max_{s \le t\le T\wedge(s + \delta)} \bigg| \int^t_s \tilde\sigma_u \, \mathrm dW_u \bigg|^2 \bigg| \mathcal F_s \bigg] }
	\le
	2 \sqrt{ \E\bigg[ \bigg| \int_s^{T\wedge(s + \delta)} \tilde\sigma_u \, \mathrm dW_u \bigg|^2 \bigg| \mathcal F_s \bigg] }.
\end{align*}
Then by combining these inequalities and using the It\^o isometry for conditional expectations, we finally obtain
\begin{align*}
	\E\bigg[ \max_{s \le t\le T\wedge(s + \delta)} \bigg| \int^t_s \tilde\sigma_u \, \mathrm dW_u \bigg| \bigg| \mathcal F_s \bigg]
	\le
	2 \sqrt{ m \E\bigg[ \int_s^{T\wedge(s + \delta)} m \overline\sigma^2 \, \du \bigg| \mathcal F_s \bigg] }
	\le
	2 m \overline\sigma \sqrt{ \delta }.
\end{align*}
\end{proof}

\begin{remark}
Obviously, we have the same result for a drift part of an It\^o process, i.e. if $\tilde\mu \in \mathcal L^\infty_{\prog}(\Omega\times[0,T];\R^m)$, then $\{\int^t_0 \tilde\mu_u \, \mathrm du\}_{t\in[0,T]}$ satisfies Condition \hyperlink{C2}{C.2}. However, we cannot expect the weaker Condition \hyperlink{C1}{C.1} under weaker integrability assumptions, as the following example shows. Let $T = 2$, $\tilde\mu_u := \mathbbm 1_{[1,2]}(u) \cdot |W_1|$, $\eps := 1$, let $\delta \in (0,1)$ be arbitrary and let $s := 1$, $\tau := 1 + \delta$ and $V := \mathbbm 1_{\{\delta |W_1| > 1\}}$. Then
\begin{align*}
	\frac {|\E[V \int^{1+\delta}_1 \tilde\mu_u \, \du]|} {\E[|V|]}
	=
	\frac {\E[\mathbbm 1_{\{\delta |W_1| > 1\}} \cdot \delta |W_1|]|} {\E[\mathbbm 1_{\{\delta |W_1| > 1\}}]}
	>
	1.
\end{align*}
\end{remark}

Finally, we prove that the strong Condition \hyperlink{C2}{C.2} holds for a certain class of BSDEs driven by forward SDEs. Specifically, we prove that the solution to the BSDE can be expressed as a {\sl uniformly continuous} function of the forward process and then we apply Lemma \ref{thmftWtpropstarlemma}. The representation of the solution in terms of a \textit{continuous} function has been proven by El Karoui~\cite{elkaroui} already.

For all $(t,x) \in [0,T]\times \R^n$, we consider the following SDE on $[t,T]$,
\begin{align*}
	\mathrm dX^{t,x}_s
	=
	\tilde\mu\big( s, X^{t,x}_s \big) \, \ds
	+ \tilde\sigma\big( s, X^{t,x}_s \big) \, \mathrm dW_s,
	\quad
	X^{t,x}_t
	=
	x,
\end{align*}
where $\tilde\mu, \tilde\sigma \colon [0,T] \times \R^n \to \R^n \times \R^{n\times m}$ are measurable and bounded and satisfy the Lipschitz condition
$$
\big\Vert (\tilde\mu, \tilde\sigma)(t,x_1) - (\tilde\mu, \tilde\sigma)(t,x_2) \big\Vert_\infty \le L \Vert x_1 - x_2\Vert_\infty.
$$
By the previous results, we obtain that $X^{t,x}$ satisfies Condition \hyperlink{C2}{C.2}.

Standard computations show that there exists a constant $C \in (0,\infty)$ such that, for all $t \in [0,T]$ and $x_1,x_2 \in \R^n$,
\begin{align}
	\E\Big[ \sup_{s\in[t,T]} \big\Vert X^{t,x_1}_s - X^{t,x_2}_s \big\Vert_\infty^2 \Big]
	\le
	C \Vert x_1 - x_2\Vert_\infty^2.
	\label{eqDeltaFsdeSolSmall}
\end{align}

\begin{theorem} \label{thmBsdeSolExistsUniqueWithHthm}
Let $\Psi\colon \R^n \to \R$ be bounded, $\psi\colon [0,T]\times\R^n\times\R \to \R$ be continuous and bounded and let both function satisfy the Lipschitz condition
\[
	\big| \Psi(x_1) - \Psi(x_2) \big|
	+ \big| \psi(t,x_1,y_1) - \psi(t,x_2,y_2) \big|
	\le L \big( \Vert x_1 - x_2\Vert_\infty + \vert y_1 - y_2\vert \big).
\]

For $(t,x) \in [0,T]\times\R^n$, let $(Y^{t,x},Z^{t,x}) \in \mathcal L_{\mathcal P}^2(\Omega\times[0,T];\R\times \R^m)$ be the unique solution to the BSDE on $[t,T]$,
\begin{align*}
	- \, \mathrm dY^{t,x}_s
	=
	\psi\big(s, X^{t,x}_s, Y^{t,x}_s\big) \, \ds
	- Z^{t,x}_s \, \mathrm dW_s,
	\quad
	Y^{t,x}_T
	=
	\Psi \big( X^{t,x}_T \big)
\end{align*}
driven by the forward process $X^{t,x}$, which exists due to Theorem 2.1 in \cite{elkaroui}.
Then there exists a uniformly continuous function $h\colon [0,T]\times\R^n \to \R$, which does not depend on $(t,x)$, such that
\begin{align}
	\IP\Big[ h\big(s,X^{t,x}_s\big) = Y^{t,x}_s \mbox{ for all } s \in [t,T] \Big]
	= 1.
	\label{eqthmBsdeSolExistsUniqueWithHthmPHEqYEq1}
\end{align}
In particular, the process $Y^{t,x}$ satisfies Condition \hyperlink{C2}{C.2} on $[t,T]$.
\end{theorem}

\begin{proof}
By Theorem 3.4 in \cite{elkaroui} and the remark following it, there exists a decoupling field $h\colon [0,T]\times\R^n \to \R$ that is $1/2$-Hölder continuous in $t$ and locally Lipschitz continuous in $x$ such that \eqref{eqthmBsdeSolExistsUniqueWithHthmPHEqYEq1} holds. It therefore remains to prove the uniform continuity of $h$. Since
\begin{align*}
	h(t,x)
	= \E\bigg[
		\Psi \big( X^{t,x}_T \big)
		+ \int_t^T \psi\Big(u, X^{t,x}_u, h \big( u,X^{t,x}_u \big) \Big) \, \du
	\bigg],
\end{align*}
we see that $h$ is bounded, say $|h| \le \overline h$. Then
\begin{align*}
	\tilde\omega_h(s,\delta)
	:= \sup_{\Vert x_1-x_2 \Vert_\infty \le \delta} \big| h(s,x_1) - h(s,x_2) \big|
\end{align*}
is finite, and Lemma \ref{thmexistsmodcontconcavelemma} allows us to define its concave envelope $\omega_h(s,\cdot)\colon [0,\infty) \to [0,\infty)$.

Using the Lipschitz continuity of $\psi$,
we obtain, for all $x_1,x_2 \in \R^n$ and $t \in [0,T]$ that
\begin{align*}
	\big| h(t,x_1) - h(t,x_2) \big|
	&\le
	 L \E\Big[ \big\Vert X^{t,x_1}_T - X^{t,x_2}_T \big\Vert_\infty \Big] \\
	& \qquad + L \int_t^T \bigg(
		\E\Big[ \big\Vert X^{t,x_1}_u - X^{t,x_2}_u \big\Vert_\infty \Big]
		 + \E\bigg[ \tilde\omega_h \Big( u, \big\Vert X^{t,x_1}_u - X^{t,x_2}_u \big\Vert_\infty \Big) \bigg]
	\bigg) \, \du.
\end{align*}

Using \eqref{eqDeltaFsdeSolSmall} along with Lemma \ref{thmexistsmodcontconcavelemma}, if $\Vert x_1-x_2 \Vert_\infty \le \delta$, then
\begin{align*}
	\E\Big[ \big\Vert X^{t,x_1}_T - X^{t,x_2}_T \big\Vert_\infty \Big]
	& \le
	C \delta
\end{align*}
and
\begin{align*}
	\E\bigg[ \tilde\omega_h\Big( u, \big\Vert X^{t,x_1}_u - X^{t,x_2}_u \big\Vert_\infty \Big) \bigg]
	& \le
	\omega_h\bigg( u, \E\Big[ \big\Vert X^{t,x_1}_u - X^{t,x_2}_u \big\Vert_\infty \Big] \bigg)
	\le
	\omega_h( u, C \delta ),
\end{align*}
for some constant $C \in (0,\infty)$ that depends only on $\tilde\mu$, $\tilde\sigma$ and $L$. In view of \eqref{eqthmexistsmodcontlemmatildeomegabound},
\begin{align*}
	\omega_h(u,C\delta)
	\le
	(C+1) \tilde\omega_h(u,\delta).
\end{align*}
This shows that
\begin{align*}
	\tilde\omega_h(t,\delta)
	\le&
	L C \delta
	+ L \int_t^T \big(
		C \delta
		+ (C+1) \tilde\omega_h(u,\delta)
	\big) \, \du
	\\=&
	L C \delta (1 + T-t)
	+ L (C+1) \int_t^T \tilde\omega_h(u,\delta) \, \du.
\end{align*}

Hence the Gronwall inequality implies that
$h$ is uniformly Lipschitz continuous in the second variable. Since $h$ is uniformly $1/2$-Hölder continuous in the first variable, this proves the assertion.
\end{proof}

\section{Convergence results for SDEs and random ODEs} \label{SectAppendixConv}

In this appendix, we establish three convergence results for stochastic integral equations that are useful to establish the convergence of our coefficient and state processes. In Subsection \ref{SectAppendixConvSde}, we consider sequences of SDEs parametrized by $\eta\to 0$ with ``positive feedback", i.e. sequences of equations that are driven away from their limits.
In Subsection \ref{SectAppendixConvOde}, we consider sequences of random ODE systems parametrized by $\eta\to 0$ with ``negative feedback", which are driven towards their limit.

For equations with positive feedback, we require a priori information about the terminal value and the increments of the coefficient processes to be bounded independently of their past (see Condition \hyperlink{C1}{C.1} and Condition \hyperlink{C2}{C.2}). These bounds prevent the stochastic process from reaching its terminal value if a large difference between the limit and the pre-limit occurs at some time $s \in (0,T)$. This will imply a.s.~uniform convergence. An almost sure statement under negative feedback cannot be expected. Instead, we follow a pathwise approach and establish uniform convergence in probability.

\subsection{SDEs with positive feedback} \label{SectAppendixConvSde}

\subsubsection{An equation without scaling} \label{SectAppendixConvSdeNotScaled}

We consider a family of real-valued, adapted, continuous stochastic processes $b^\eta=\{b^\eta_t\}_{t \in (0,T)}$ that satisfy the integral equation
\begin{align}
	\mathrm d b^\eta_t
	=
	\Big( a\big( t, b^\eta_t \big) + q^\eta_t \Big) \, \dt
	+ \overline Z^\eta_t \, \mathrm dW_t
	\quad
	\mbox{on } (0,T),
	\label{eqSdeBGeneral}
\end{align}
where $a\colon \Omega\times (0,T) \times \R \to \R$ is adapted and continuous and $\{q^\eta_t\}_{t\in(0,T)}$ is an adapted, real-valued, continuous stochastic process and
$$
\{\overline Z^\eta_t\}_{t\in(0,T)}
\in
\mathcal L_{\prog}^2 \big(\Omega\times(0,T-];\R^m\big).
$$

Our goal is to prove that the processes $\{b^\eta_t\}_{t \in (0,T)}$ with terminal conditions $b^\eta_T = 0$ converge to $0$ as $\eta \to 0$ uniformly on $(0,T)$ if the process $q^\eta$ does and if the mapping $a$ is such that $b^\eta$ is driven away from $0$. Intuitively, the last condition makes it impossible for $b^\eta$ to return to $0$ at the terminal time once it has left a neighbourhood of $0$. Specifically, we assume that the following assumption is satisfied.

\begin{assumption} \label{assthmbGeneralConvlemma}
\hspace{2em} 
\begin{itemize}
	\item[i)] There exists a real constant $\overline a > 0$ such that
	$$
		\sup_{t \in (0,T)}
		\sup_{\eta > 0}
			\Big| a\big( t, b^\eta_t \big) \Big|
		\le \overline a.
	$$
	\item[ii)] For all $\eps > 0$, there exists a $\delta > 0$ such that, for all $\eta>0$,
	$$
		\IP\Big[ \sup_{t \in [T-\delta,T)} \big| b^\eta_t \big| \le \eps \Big]
		=
		1.
	$$
	\item[iii)] There exists a function $\eta_1 \colon (0,\infty) \to (0,\infty)$ such that, for all $\eps>0$ and $\eta \le \eta_1(\eps)$,
	\begin{align*}
		\IP\Big[
			\sup_{t \in (0,T-\eps]} \big| q^\eta_t \big| \le \eps
		\Big]
		= 1.
	\end{align*}
	\item[iv)] There exists an $\eps_1 > 0$ such that
	$$
		\inf_{t \in (0,T)}
		\inf_{b \in [-\eps_1,\eps_1]}
			b \cdot a(t,b)
		\ge 0.
	$$
\end{itemize}
\end{assumption}

\begin{lemma} \label{thmbGeneralConvlemma}
Under Assumption \ref{assthmbGeneralConvlemma}, there exists a function $\eta_0\colon (0,\infty) \to (0,\infty)$, which depends only on $\eta_1$ and the mapping $\eps \mapsto \delta$ introduced in Assumption \ref{assthmbGeneralConvlemma} ii), such that, for all $\eps>0$ and all $\eta \in (0,\eta_0(\eps)]$,
\begin{align*}
	\IP \Big[ \sup_{s \in (0,T)} \big| b^\eta_s \big| \le \eps \Big]
	=
	1.
\end{align*}
\end{lemma}
\begin{proof}
Let $\eps \in (0,\eps_1]$. In a first step, we show that there exists an $\eta_0(\eps) > 0$ such that
\begin{align*}
	\sigma^{\eps,\eta}
	:=
	\inf \Big\{
		t \in (0,T) :
		\IP\big[b^\eta_s < -\eps\big] = 0
		\text{ for all } s \in [t,T)
	\Big\}
\end{align*}
satisfies $\sigma^{\eps,\eta} = 0$ for all $\eta \in (0,\eta_0(\eps)]$.
Since $b^\eta$ is continuous, this will imply that, for all $\eta \in (0,\eta_0(\eps)]$,
$$
	\IP \Big[ \inf_{s \in (0,T)} b^\eta_s \ge -\eps \Big]
	=
	1.
$$
The other direction of the statement can be proven analogously.

We now assume to the contrary that $0 < \sigma^{\eps,\eta}$ and prove that this leads to a contradiction if $\eta$ is small enough. For all $s\in (0,T)$ and $\eta \in (0,\infty)$, we can define the stopping time
\begin{align*}
	\tau^s
	:= \tau^{s,\eps,\eta}
	:= \inf \big\{ u \in [s,T) : b^\eta_u = -\eps/2 \big\}.
\end{align*}
By Assumption \ref{assthmbGeneralConvlemma} \textit{ii)}, there exists a small enough $\delta > 0$ such that, for all $s\in (0,T)$ and $\eta \in (0,\infty)$,
\begin{align}
	\IP\big[ b^\eta_s < -\eps, \tau^s > T - \delta \big]
	=
	0.
\label{eqthmbconv0lemmaTauLeTMinusEtc}
\end{align}
Since $\tau^s$ is a stopping time, due to the square-integrability of $\overline Z^\eta$, by the optional sampling theorem (cf. Theorem 1.3.22 in \cite{karatzas_shreve}) and Assumption \ref{assthmbGeneralConvlemma} \textit{i)}, \textit{iv)}, \eqref{eqSdeBGeneral} and \eqref{eqthmbconv0lemmaTauLeTMinusEtc}, we obtain that, for all $s\in (0,\sigma^{\eps,\eta})$ and $\eta \in (0,\infty)$,
\begin{align*}
	\IP\big[ b^\eta_s < -\eps \big] \frac \eps 2
	\le&
	\E\Big[ \mathbbm 1_{\{b^\eta_s < -\eps\}} \big( b^\eta_{\tau^s} - b^\eta_s \big) \Big]
	\\=&
	\E\bigg[ \mathbbm 1_{\{b^\eta_s < -\eps\}} \bigg(
		\int_s^{\tau^s \wedge \sigma^{\eps,\eta}} a\big( t, b^\eta_t \big) \, \dt
		+ \int_{\tau^s \wedge \sigma^{\eps,\eta}}^{\tau^s} a\big( t, b^\eta_t \big) \, \dt
		+ \int_s^{\tau^s} q^\eta_t \, \dt
	\bigg) \bigg].
\end{align*}

Inside the second integral, if $\sigma^{\eps,\eta} < \tau^s$, it holds $b^\eta_t \in [-\eps,-\eps/2]$. Hence

\begin{align}
\begin{split}
	\IP\big[b^\eta_s < -\eps\big] \frac \eps 2
	\le
	\IP\big[ b^\eta_s < -\eps \big] \Big(
		\big( \sigma^{\eps,\eta} - s \big) \overline a
		+ T \sup_{t\in [0,T-\delta]} \big| q^\eta_t \big|
	\Big).
\end{split}
\label{eqthmbGeneralConvlemmaMain}
\end{align}

The assumption $0 < \sigma^{\eps,\eta}$ implies that for any $\alpha > 0$, we can find
$$
\sigma^{\eps,\eta} - \alpha
\le s
\le \sigma^{\eps,\eta}
$$
with $\IP[ b^\eta_s < -\eps ] > 0$. However, due to Assumption \ref{assthmbGeneralConvlemma} \textit{ii)} and \textit{iii)}, we can choose $\alpha \le \eps/(5\overline a)$, $\delta \le \eps/(5T)$ and $\eta_0(\eps) < \eta_1(\delta)$ such that \eqref{eqthmbGeneralConvlemmaMain} leads to a contradiction because $\IP[ b^\eta_s < -\eps ] > 0$. This implies $0 = \sigma^{\eps,\eta}$.
\end{proof}

\subsubsection{An equation with scaling} \label{SectAppendixConvSdeScaled}

In this section, we establish an abstract convergence result for stochastic processes $\{P^\eta_t\}_{t \in (0,T)}$ indexed by some parameter $\eta > 0$ that satisfy the integral equation
\begin{align}
	\mathrm d ( \psi P^\eta )_t
	=
	\sqrt{\eta^{-1}} \Big(
		a\big(P^\eta_t, P^0_t\big)
		+ q^\eta_t
	\Big) \, \dt
	+ \, \mathrm dL^\eta_t
	+ \overline Z^\eta_t \, \mathrm dW_t
	\quad
	\mbox{on } (0,T),
	\label{eqthmfggeneralconvlemmaSde}
\end{align}
where $a:\mathbb{R}^2 \to \R$ is a measurable mapping,
$\{\psi_t\}_{t\in(0,T)}$,
$\{P^\eta_t\}_{t\in(0,T)}$,
$\{P^0_t\}_{t\in(0,T)}$,
$\{q^\eta_t\}_{t\in(0,T)}$
and $\{L^\eta_t\}_{t\in(0,T)}$
are adapted, real-valued, continuous stochastic processes and $$\{\overline Z^\eta_t\}_{t\in(0,T)} \in \mathcal L_{\prog}^2\big(\Omega\times(0,T-];\R^m\big).$$

Our goal is to prove that the processes $\{P^\eta_t\}_{t \in (0,T)}$ converge to $P^0$ as $\eta \to 0$ uniformly on compact subintervals of $(0,T)$ if the mapping $a(\cdot,\cdot)$ is such that $P^\eta$ is driven away from $P^0$ and if the boundary condition $\lim_{t\to T} ( P^\eta_t - P^0_t ) \ge 0$ holds. If the integral equation \eqref{eqthmfggeneralconvlemmaSde} holds on the whole interval $(0,T]$, then it is enough to assume that $P^\eta_T \ge P^0_T$. When applying the abstract convergence result to the BSDEs \eqref{eqBsdeF} and \eqref{eqBsdeG}, the former condition holds if $N = +\infty$ while the latter holds if $N$ is finite.

\begin{assumption} \label{assthmfggeneralconvlemma}
\hspace{2em} 
\begin{itemize}
	\item[i)] There exist positive constants $\underline \psi, \overline \psi, \overline P$ such that $\underline\psi \le \psi \le \overline\psi$, $|P^0| \le \overline P$ and $P^\eta \ge - \overline P$. Moreover, there exists a function $\eta_1 \colon (0,\infty) \to (0,\infty)$ such that
	$$
	P^\eta|_{\Omega\times(0,T-\eps]} \le \overline P \quad \mbox{for all} \quad \eps > 0 \mbox{ and } \eta \le \eta_1(\eps).
	$$
	\item[ii)] One of the following two ``boundary conditions'' holds:
	\begin{itemize}
		\item[a)] For all $\eta > 0$, there exists a $T_\eta < T$ such that
	$$
		\IP\Big[ \inf_{T_\eta \le t < T} \big( P^\eta_t - P^0_t \big) \ge 0 \Big]
		=
		1.
	$$
	\item[b)] The integral equation \eqref{eqthmfggeneralconvlemmaSde} holds on the whole interval $(0,T]$ and $P^\eta_T \ge P^0_T$. Thus, in particular, $\psi$, $P^\eta$, $P^0$, $q^\eta$ and $L^\eta$ are continuous on $(0,T]$ and $\overline Z^\eta \in \mathcal L_{\prog}^2(\Omega\times(0,T];\R^m)$.
	\end{itemize}
	\item[iii)] There exists a function $\eta_2 \colon (0,\infty) \to (0,\infty)$ such that, for all $\eps>0$ and all $\eta \in (0,\eta_2(\eps)]$,
\begin{align*}
	\IP\Big[
		\inf_{t \in (0,T-\eps]} q^\eta_t \ge -\eps,
		\sup_{t \in (0,T)} q^\eta_t \le \eps
	\Big]
	=
	1.
\end{align*}
	\item[iv)]
	The function $a \colon \R^2 \to \R$ is non-decreasing and differentiable in the first variable and $a(P^0_t(\omega),P^0_t(\omega)) \equiv0$. Moreover, there exist $\overline\eps, \beta > 0$ such that
\[
	\mddx a(x,P^0_t(\omega)) \ge \beta \quad \mbox{for all} \, \big| x-P^0_t(\omega) \big| \le \overline\eps.
\]
	\item[v)] The processes $\psi$ and $P^0$ satisfy Condition \hyperlink{C2}{C.2} and the processes $L^\eta$ $(\eta > 0)$ satisfy Condition \hyperlink{C1}{C.1} uniformly in $\eta>0$.
\end{itemize}
\end{assumption}

\begin{lemma} \label{thmfggeneralconvlemma}
Under Assumption \ref{assthmfggeneralconvlemma},
there exists a mapping $\eta_0\colon (0,\infty) \to (0,\infty)$ that depends only on $\psi$, $\overline P$, $\eta_1$, $\eta_2$, $P^0$, $\overline\eps$, $\beta$ and the mapping $\eps \mapsto \delta$ corresponding to the uniform satisfaction of Condition \hyperlink{C1}{C.1} by $\{L^\eta\}_{\eta}$, such that, for all $\eps>0$ and $\eta \in (0,\eta_0(\eps)]$,
\begin{align*}
	\IP \Big[
		\sup_{s \in (0,T-\eps]} \big( P^\eta_s - P^0_s \big) \le \eps,
		\inf_{s \in (0,T)} \big( P^\eta_s - P^0_s \big) \ge -\eps
	\Big]
	=
	1.
\end{align*}
\end{lemma}
\begin{proof}
Let $p^\eta := P^\eta - P^0$ and let us fix $\eps \in (0,\overline\eps]$. We first show that there exists $\eta_0(\eps) > 0$ such that
\begin{align*}
	\sup_{\eta \le \eta_0(\eps)} \sup_{s \in (0,T-\eps]} \IP\big[ p^\eta_s > \eps \big] = 0.
\end{align*}
Since $p^\eta$ is continuous, this proves that, for all $\eta \le \eta_0(\eps)$,
\[
	\IP \Big[ \sup_{s \in (0,T-\eps]} p^\eta_s \le \eps \Big]=1.
\]
To this end, we fix $\delta \leq \eps/2$ and $s \in (0,T-\eps]$. For all $\eta > 0$, we can define the stopping time
\begin{align*}
	\tau^\eta
	:=
	(s+\delta) \wedge \inf \Big\{ u \in [s,T) : \big| p^\eta_u \big| = \eps/2 \Big\}
	\le
	T - \frac \eps 2.
\end{align*}
Due to the $L^2$-integrability of $\overline Z^\eta$ on $[s,T-\eps/2]$ and by the optional sampling theorem (cf. Theorem 1.3.22 in \cite{karatzas_shreve}),
we obtain
\begin{equation} \label{eqthmfggeneralconvlemmaEEq0}
\begin{split}
	0
	=&
	\E\bigg[ \mathbbm 1_{\{ p^\eta_s > \eps\}} \bigg(
		\sqrt{\eta^{-1}} \int_s^{\tau^\eta} \Big(
			a\big(P^\eta_u, P^0_u\big)
			+ q^\eta_u
		\Big) \, \du
		+ L^\eta_{\tau^\eta} - L^\eta_s
	\\&\quad
		+ \psi_s P^0_s - \psi_{\tau^\eta} P^0_{\tau^\eta}
		+ \psi_s p^\eta_s - \psi_{\tau^\eta} p^\eta_{\tau^\eta}
	\bigg)\bigg].
\end{split}
\end{equation}

We now analyze the various terms in \eqref{eqthmfggeneralconvlemmaEEq0} separately to deduce that $\IP[p^\eta_s > \eps] = 0$ if
\[
	\eta \le \eta_0(\eps)
	:=
	\min \Big( \eta_1(\eps / 2), \eta_2 \big( \min( \delta, \beta \eps/4 ) \big) \Big).
\]

\begin{itemize} \item We first consider the term
\[
	\E\bigg[ \mathbbm 1_{\{ p^\eta_s > \eps\}}
		\sqrt{\eta^{-1}} \int_s^{\tau^\eta} \Big(
			a\big(P^\eta_u, P^0_u\big)
			+ q^\eta_u
		\Big) \, \du
	\bigg].
\]

On $\{p^\eta_s > \eps\}$, we have
that $p^\eta_u \ge \eps/2$ for all $u \in [s,\tau^\eta]$ and hence
\begin{align*}
	a\big(P^\eta_u, P^0_u\big)
	= a\big( P^0_u + p^\eta_u, P^0_u\big)
	\ge a\big( P^0_u + \eps/2, P^0_u\big)
	\ge \frac {\beta \eps} 2.
\end{align*}

Together with our assumption on the process $q^\eta$ and our choice of $\eta$, this implies that
\begin{align*}
	&
	\E\bigg[ \mathbbm 1_{\{ p^\eta_s > \eps\}}
		\sqrt{\eta^{-1}} \int_s^{\tau^\eta} \Big(
			\underbrace{a\big(P^\eta_u, P^0_u\big)}_{\ge \beta \eps/2}
			+ \underbrace{q^\eta_u}_{\ge -\beta \eps/4}
		\Big) \, \du
	\bigg]
	\\\ge&
	\E\bigg[ \mathbbm 1_{\{ p^\eta_s > \eps\}}
		\sqrt{\eta^{-1}} (\tau^\eta - s) \frac {\beta \eps} 4
	\bigg]
	\\\ge&
	\delta \sqrt{\eta^{-1}} \frac {\beta \eps} 4 \IP\big[ p^\eta_s > \eps \, \text{and} \, \tau^\eta = s + \delta \big].
\end{align*}

\item By Assumption \ref{assthmfggeneralconvlemma} \textit{v)}, Lemma \ref{thmC2ImpliesC1lemma} and Theorem \ref{thmXYSatisfiesStarPropertythm}, we can choose for any $\eps_0 > 0$ some $\delta > 0$, which does not depend on $\eta$, such that
\begin{equation} \label{MLY}
	\Big| \E \big[ \mathbbm 1_{\{ p^\eta_s > \eps\}} ( Y_{\tau^\eta} - Y_s ) \big] \Big|
	\le
	\eps_0 \IP\big[ p^\eta_s > \eps \big],
	\quad \mbox{for} \, Y = \psi P^0, L^{\eta},
\end{equation}
and
\begin{align*}
	\E \big[ | \psi_{\tau^\eta} - \psi_s | \big| \mathcal F_s \big]
	\le
	\eps_0.
\end{align*}

\item Using \eqref{MLY}, along with the fact that $P^\eta|_{\Omega\times (0,T-\eps]} \le \overline P$ due to Assumption \ref{assthmfggeneralconvlemma} \textit{i)}, yields
\begin{align*}
	&
	\E\Big[ \mathbbm 1_{\{ p^\eta_s > \eps\}} \big( \psi_s p^\eta_s - \psi_{\tau^\eta} p^\eta_{\tau^\eta} \big) \Big]
	\\ \ge&
	\E\Big[ \mathbbm 1_{\{ p^\eta_s > \eps \, \text{and} \, {\tau^\eta} < s + \delta \}} \big( \psi_s \eps/2 + (\psi_s - \psi_{\tau^\eta}) \eps/2 \big) \Big]
	+ \E\Big[ \mathbbm 1_{\{ p^\eta_s > \eps \, \text{and} \, {\tau^\eta} = s + \delta \}} \big( \psi_s \eps - \psi_{\tau^\eta} \underbrace{p^\eta_{\tau^\eta}}_{\le 2 \overline P} \big) \Big]
	\\\ge&
	\frac \eps 2 \underline\psi \IP\big[ p^\eta_s > \eps \, \text{and} \, {\tau^\eta} < s + \delta \big]
	- \frac {\eps \eps_0} 2 \IP\big[ p^\eta_s > \eps \big]
	+ \big( \eps \underline\psi - 2 \overline\psi \overline P \big) \IP\big[ p^\eta_s > \eps \, \text{and} \, {\tau^\eta} = s + \delta \big].
\end{align*}
\end{itemize}

Altogether this yields
\begin{align*}
	0
	\ge&
	\IP\big[ p^\eta_s > \eps \, \text{and} \, {\tau^\eta} = s + \delta \big] \bigg(
		\delta \sqrt{\eta_0(\eps)^{-1}} \frac {\beta \eps} 4
		- 2 \eps_0
		- \frac {\eps \eps_0} 2
		+ \eps \underline\psi - 2 \overline\psi \overline P
	\bigg)
	\nonumber
	\\&
	+ \IP\big[ p^\eta_s > \eps \, \text{and} \, {\tau^\eta} < s + \delta \big] \bigg(
		- 2 \eps_0
		+ \frac \eps 2 \underline\psi
		- \frac {\eps \eps_0} 2
	\bigg).
\end{align*}
Now, if we choose first $\eps_0$ and then $\eta_0(\eps)$ small enough, the coefficients that multiply the probabilities become positive. Hence both probabilities must be equal to zero and so we have $\IP[ p^\eta_s > \eps] = 0$ if $\eta \le \eta_0(\eps)$. Analogously, we can prove that $\IP[ p^\eta_s < -\eps] = 0$ for all $s \in (0,T)$. The main difference is that ${\tau^\eta} < T-\delta$ does not hold on the set $\{p^\eta_s < -\eps\}$. Instead, we only obtain $\tau^\eta < T$ (if Assumption \ref{assthmfggeneralconvlemma} \textit{ii)} \textit{b)} holds) or $\tau^\eta < T_\eta$ (if Assumption \ref{assthmfggeneralconvlemma} \textit{ii)} \textit{a)} holds).
\end{proof}

\subsection{ODE systems with negative feedback} \label{SectAppendixConvOde}

In this section, we establish a convergence result for ODEs with random coefficients and ``nearly singular'' driver. Specifically, we consider pairs of continuously differentiable stochastic processes $(X^\eta,Z^\eta)$, which satisfy the ODE system
\begin{align*}
\dot X^\eta_t
=&
- A^\eta_t \big( X^\eta_t + B^\eta_t Z^\eta_t \big), \quad
X^\eta_0
=
x_0,
\\
\dot Z^\eta_t
=&
C_t X^\eta_t + D_t Z^\eta_t,
\quad
Z^\eta_0
=
z_0
\end{align*}
on some time interval $[0,S]$ for all $\eta > 0$, where $A^\eta$, $B^\eta$, $C$ and $D$ are continuous, adapted, real-valued stochastic processes. We assume that the process $A^\eta$ converges in probability to infinity as $\eta \to 0$, that the process $B^\eta$ converges to a process $B^0$ in probability and that the processes $(X^\eta,Z^\eta)$ are uniformly bounded in probability.

\begin{assumption} \label{assOdeConv}
There exists a continuous adapted process $B^0$, such that for all $\eps,\delta \in (0,\infty)$, there exists some $\eta_0 = \eta_0(\eps, \delta) > 0$ such that, for all $\eta \in (0,\eta_0]$,
\begin{align}
	\IP\bigg[
		\inf_{t\in [0,S]} A^\eta_t \ge \frac 1 \eps,
		\sup_{t\in [0,S]} \big| B^\eta_t - B^0_t \big| \le \eps
	\bigg]
	\ge
	1 - \delta.
	\label{eqassOdeConvABConv}
\end{align}
Moreover, for all $\delta \in (0,1)$, there exists an $L > 0$ such that, for all $\eta>0$,
\begin{equation} \label{boundXZ}
	\IP\bigg[ \sup_{t\in[0,S]} \Big( \big| X^\eta_t \big| + \big| Z^\eta_t \big| \Big) \le L \bigg] \ge 1 - \delta.
\end{equation}
\end{assumption}

In terms of the process $B^0$, we define
\begin{align*}
X^0_t & := - B^0_t Z^0_t, \\
Z^0_t & :=
z_0 \exp\bigg( \int^t_0 \big( - C_s B^0_s + D_s \big) \, \ds \bigg).
\end{align*}

Our goal is to prove that $(X^\eta,Z^\eta) \to (X^0,Z^0)$ as $\eta \to 0$ in a suitable sense. Since the initial condition of the processes $X^\eta$ does not vary with $\eta$ and since $A^\eta_0 \uparrow \infty$, we cannot expect convergence in $t=0$. Instead, we first prove that $\frac{\dot X^\eta}{A^\eta}$ converges to $0$ on any compact subinterval of $(0,S]$ with lower, respectively upper convergence near the origin.

\begin{lemma} \label{thmDdtXEtaScaledConv0Generallemma}
If Assumption \ref{assOdeConv} holds and if $x_0 + B^0_0 z_0 > 0$,
then for all $\eps,\delta \in (0,\infty)$, there exists $\eta_1=\eta_1(\eps,\delta) > 0$ such that, for all $\eta \in (0,\eta_1]$,
\begin{align*}
	\IP\bigg[
		\sup_{t\in [0,S]} \frac {\dot X^\eta_t} {A^\eta_t} \le \eps,
		\inf_{t\in [\eps,S]} \frac {\dot X^\eta_t} {A^\eta_t} \ge -\eps
	\bigg]
	\ge
	1 - \delta.
\end{align*}
If $x_0 + B^0_0 z_0 < 0$,
then the intervals $[0,S]$ and $[\eps,S]$ in the above statement are to be swapped.

The mapping $\eta_1$ depends only on $B^0$, $C$, $D$, $x_0$, $z_0$, and the mappings $(\eps,\delta) \mapsto \eta_0(\eps,\delta)$ and $\delta \mapsto L$ introduced in Assumption \ref{assOdeConv}.
\end{lemma}
\begin{proof}
Let $y_0 := x_0 + B^0_0 z_0$. We assume w.l.o.g.~that $y_0 > 0$. Then there exists $\alpha_0 > 0$ s.t.
\begin{align*}
	y_0 - \alpha | z_0 |
	> 0 \quad \mbox{for all} \, \alpha \leq \alpha_0.
\end{align*}
Let us now put
\begin{align*}
	M^{L,\eta}
	:=&
	\bigg\{ \sup_{t \in [0,S]} \Big( \big| B^0_t \big| + |C_t| + |D_t| + \big| X^\eta_t \big| + \big| Z^\eta_t \big| \Big) \le L \bigg\},
	\\
	M^{\alpha,\eta}
	:=&
	\bigg\{
		\inf_{t\in[0,S]} A^\eta \ge \frac 1 \alpha,
		\sup_{t\in[0,S]} \big| B^\eta_t - B^0_t \big| \le \alpha
	\bigg\},
	\\
	M^{\beta,\nu}
	:=&
	\bigg\{ \sup_{s,t\in[0,S],|s-t|\le\nu} \big| B^0_s - B^0_t \big| \le \beta \bigg\}.
\end{align*}

Due to Assumption \ref{assOdeConv} and Lemma \ref{thmStarStarlemma}, for all $\delta>0$, $\beta>0$ and $\alpha \leq \alpha_0$, there exist constants $L_0(\delta), \eta_0(\alpha,\delta)$ and $\nu_0(\beta,\delta)$ such that
\[
	\inf_{L \ge L_0} \inf_\eta \IP[M^{L,\eta}] \geq 1-\frac{\delta}{6}, \qquad
	\inf_{\eta \leq \eta_0}\IP[M^{\alpha,\eta}] \geq 1-\frac{\delta}{6}, \qquad
	\inf_{\nu \le \nu_0} \IP[M^{\beta,\nu}] \geq 1-\frac{\delta}{6}.
\]
Moreover, for all $\alpha \le \alpha_0$ and $\IP$-a.a.~$\omega \in M^{\alpha,\eta}$,
\begin{align}
	\dot X^\eta_0(\omega)
	\le - A^\eta_0(\omega) \Big( y_0 - \big| B^\eta_0(\omega) - B^0_0(\omega) \big| | z_0 | \Big)
	\leq - \frac 1 \alpha \big( y_0 - \alpha |z_0| \big) < 0.
	\label{eqthmDdtXEtaScaledConv0GenerallemmadX0Le0}
\end{align}

We are now ready to prove that, for all $\eps,\delta>0$, there exists $\eta_1>0$ such that, for all $\eta \le \eta_1$,
\[
	\IP \bigg[ \sup_{t\in [0,S]} \frac {\dot X^\eta_t} {A^\eta_t} \le \eps \bigg] \geq 1- \frac{\delta}{2}.
\]
To this end, we take an $\omega \in M^{L,\eta} \cap M^{\alpha,\eta} \cap M^{\beta,\nu}$ where \eqref{eqthmDdtXEtaScaledConv0GenerallemmadX0Le0} holds and assume that there exists some $t \in [0,S]$ such that $\frac{\dot X^\eta_t(\omega)}{A^\eta_t(\omega)} \ge \eps$. It is enough to show that such a $t$ cannot exist for $L$ large and $\alpha,\beta,\nu,\eta$ small enough. Since $\frac{\dot X^\eta(\omega)}{A^\eta(\omega)}$ is continuous and $\frac{\dot X^\eta_0(\omega)}{A^\eta_0(\omega)} < 0$ if $\alpha \le \alpha_0$, we can choose a minimal $t \in [0,S]$ with the property that $\frac{\dot X^\eta_t(\omega)}{A^\eta_t(\omega)} = \eps$. Let
\begin{align*}
	s
	:=
	\sup\bigg\{ u \in (0,t) : \frac{\dot X^\eta_u(\omega)}{A^\eta_u(\omega)} = \eps/2 \bigg\}.
\end{align*}
Then, $\frac{\dot X^\eta_s(\omega)}{A^\eta_s(\omega)} = \frac{\eps}{2}$ and $\frac{\dot X^\eta_u(\omega)}{A^\eta_u(\omega)} > \frac{\eps}{2}$ for all $u \in (s,t]$. We now distinguish two cases.
\begin{itemize}
	\item Case 1: $t - s \ge \nu$. Then, the fact that $\omega \in M^{L,\eta} \cap M^{\alpha,\eta}$ yields
\begin{align}
	2 L
	\ge
	X^\eta_t(\omega) - X^\eta_s(\omega)
	=
	\int^t_s \frac {\dot X^\eta_u(\omega)} {A^\eta_u(\omega)} A^\eta_u(\omega) \, \du
	\ge
	\frac {\nu \eps} {2\alpha}
	\label{eqthmDdtXEtaScaledConv0GenerallemmaPart1Case1}.
\end{align}

	\item Case 2: $t - s < \nu$. In this case, we have
\begin{align*}
	\frac \eps 2
	=&
	\frac {\dot X^\eta_t(\omega)} {A^\eta_t(\omega)} - \frac {\dot X^\eta_s(\omega)} {A^\eta_s(\omega)}
	=
	- X^\eta_t(\omega)
	- B^\eta_t(\omega) Z^\eta_t(\omega)
	+ X^\eta_s(\omega)
	+ B^\eta_s(\omega) Z^\eta_s(\omega).
\end{align*}
Using that $\omega \in M^{\alpha,\eta}$ yields $-X_t^\eta(\omega)+X^\eta_s(\omega) = - \int^t_s \frac {\dot X^\eta_u(\omega)} {A^\eta_u(\omega)} A^\eta_u(\omega) \, \du < 0$. Using that $\omega \in M^{L,\eta} \cap M^{\beta,\nu}$ yields
\begin{equation}
\begin{split}
	\frac{\eps}{2} \le &
	- B^\eta_t(\omega) \int^t_s \dot Z^\eta_u(\omega) \, \du
	- Z^\eta_s(\omega) \big( B^\eta_t(\omega) - B^\eta_s(\omega) \big)
	\\\le&
	2L^2 \nu \Big(
		\big| B^0_t(\omega) \big|
		+ \big| B^\eta_t(\omega) - B^0_t(\omega) \big|
	\Big)
	\\&
	+ L \Big( \big| B^\eta_t(\omega) - B^0_t(\omega) \big| + \big| B^0_t(\omega) - B^0_s(\omega) \big| + \big| B^0_s(\omega) - B^\eta_s(\omega) \big| \Big)
	\\\le&
	2 L^2 \nu ( L + \alpha )
	+ L ( 2\alpha + \beta ).
	\label{eqthmDdtXEtaScaledConv0GenerallemmaPart1Case2}
\end{split}
\end{equation}
\end{itemize}

Now, we first choose $L \geq L_0(\delta)$, then $\beta > 0$ such that $3L\beta < \frac{\eps}{4}$, then $\nu \leq \nu_0(\beta,\delta)$ such that $2 L^2 \nu( L + \alpha_0 ) < \frac{\eps}{4}$, then $\alpha \leq \alpha_0$ such that $\frac{\nu \eps}{2 \alpha} > 2 L$ and $\alpha \leq \beta$ and finally $\eta_1(\eps,\delta) \leq \eta_0(\alpha,\delta)$. Then both \eqref{eqthmDdtXEtaScaledConv0GenerallemmaPart1Case1} and \eqref{eqthmDdtXEtaScaledConv0GenerallemmaPart1Case2} are violated. As a result,
\begin{align*}
M^{L,\eta}
\cap M^{\alpha,\eta}
\cap M^{\beta,\nu}
\subseteq
\bigg\{ \sup_{t\in [0,S]} \frac {\dot X^\eta_t} {A^\eta_t} \le \eps \bigg\}
\end{align*}
and hence
\begin{align}
&
\IP\bigg[ \sup_{t\in [0,S]} \frac {\dot X^\eta_t} {A^\eta_t} \le \eps \bigg]
\ge 1 - \frac{\delta}{2}.
\label{eqthmDdtXEtaScaledConv0GenerallemmaPart1Result}
\end{align}

It remains to prove that, for suitably chosen parameters
\begin{align*}
M^{L,\eta}
\cap M^{\alpha,\eta}
\cap M^{\beta,\nu}
\subseteq
\bigg\{ \inf_{t\in [\eps,S]} \frac {\dot X^\eta_t} {A^\eta_t} \ge -\eps \bigg\}.
\end{align*}
To this end, we fix $\omega \in M^{L,\eta} \cap M^{\alpha,\eta} \cap M^{\beta,\nu}$ and assume that there exists a minimal $t \in [\eps,S]$ such that $\frac{\dot X^\eta_t(\omega)}{A^\eta_t(\omega)} \le -\eps$. Using that $\omega \in M^{L,\eta} \cap M^{\alpha,\eta}$, it is straightforward to show that, for all sufficiently small $\alpha$,
there exists an $r \in [0,\eps]$ such that $\frac{\dot X^\eta_r(\omega)}{A^\eta_r(\omega)} \ge -\eps/2$.
Hence, we can define
\begin{align*}
s
:=
\sup \bigg\{ u \in [0,t) : \frac{\dot X^\eta_u(\omega)}{A^\eta_u(\omega)} = -\eps/2 \bigg\}.
\end{align*}
Then $0 \leq s < t$ and $\frac{\dot X^\eta_u(\omega)}{A^\eta_u(\omega)} \in [-\eps,-\eps/2]$ for all $u \in [s,t]$. We can now use the same arguments as in the first step to conclude that

\begin{align}
\IP\bigg[ \inf_{t\in [\eps,S]} \frac {\dot X^\eta_t} {A^\eta_t} \ge -\eps \bigg]
\ge
1 - \frac \delta 2.
\label{eqthmDdtXEtaScaledConv0GenerallemmaPart2Result}
\end{align}

Combining \eqref{eqthmDdtXEtaScaledConv0GenerallemmaPart1Result} and \eqref{eqthmDdtXEtaScaledConv0GenerallemmaPart2Result} yields the desired result.
\end{proof}

\begin{lemma} \label{thmZConv0Generallemma}
Let Assumption \ref{assOdeConv} be satisfied and let $x_0 + B^0_0 z_0 \ne 0$. Then for all $\eps,\delta \in (0,\infty)$, there exists $\eta_2=\eta_2(\eps,\delta) > 0$ such that, for all $\eta \in (0,\eta_2]$,
\begin{align*}
	\IP\Big[ \sup_{t\in [0,S]} \big| Z^\eta_t - Z^0_t \big| \le \eps \Big]
	\ge
	1 - \delta.
\end{align*}

The mapping $\eta_2$ depends only on $B^0$, $C$, $D$, $x_0$, $z_0$, and the mappings $(\eps,\delta) \mapsto \eta_0(\eps,\delta)$ and $\delta \mapsto L$ introduced in Assumption \ref{assOdeConv}.
\end{lemma}

\begin{proof}
Since
\begin{align*}
	Z^\eta_t - Z^0_t
	=&
	\exp\bigg( \int^t_0 \big( D_u - B^\eta_u C_u \big) \, \du \bigg)
	\int^t_0
		\exp\bigg( - \int_0^s \big( D_u - B^\eta_u C_u \big) \, \du \bigg) \cdot
\\& \quad\quad
		\bigg(
			\big( B^0_s - B^\eta_s \big) C_s Z^0_s
			- C_s \cdot \frac{\dot X^\eta_s}{A^\eta_s}
		\bigg)
	\, \ds,
\end{align*}
the assertion follows from Assumption \ref{assOdeConv} and Lemma \ref{thmDdtXEtaScaledConv0Generallemma}.
\end{proof}

The next theorem follows by Assumption \ref{assOdeConv}, Lemma \ref{thmDdtXEtaScaledConv0Generallemma} and Lemma \ref{thmZConv0Generallemma}.

\begin{theorem} \label{thmXConv0Generalthm}
If Assumption \ref{assOdeConv} holds and $x_0 + B^0_0 z_0 > 0$,
then for all $\eps,\delta \in (0,\infty)$, there exists $\eta_3 = \eta_3(\eps,\delta) > 0$ such that, for all $\eta \in (0,\eta_3]$,
\begin{align*}
	\IP\bigg[
		\sup_{t\in [0,S]} \max\Big(
			\big| Z^0_t - Z^\eta_t \big|,
			X^0_t - X^\eta_t
		\Big) \le \eps,
		\inf_{t\in [\eps,S]} \big( X^0_t - X^\eta_t \big) \ge -\eps
	\bigg]
	\ge
	1 - \delta.
\end{align*}
If $x_0 + B^0_0 z_0 < 0$,
then the intervals $[0,S]$ and $[\eps,S]$ in the statement are to be swapped.

The mapping $\eta_3$ depends only on $B^0$, $C$, $D$, $x_0$, $z_0$, and the mappings $(\eps,\delta) \mapsto \eta_0(\eps,\delta)$ and $\delta \mapsto L$ introduced in Assumption \ref{assOdeConv}.
\end{theorem}

\end{appendix}

\bibliography{bibdata}

\end{document}